\documentclass{tac}
\pdfoutput=1 

\newif{\ifarxiv} 
\arxivtrue

\newif{\iftac}
\tacfalse

\iftac
\arxivfalse
\fi

\newif{\ifchanges} 
\changestrue

\ifarxiv
\changesfalse
\fi

\usepackage[square]{natbib}
\usepackage{graphicx} 
\usepackage{amsfonts,amsmath,amssymb}
\ifarxiv
\usepackage[english]{babel}
\else
\usepackage[T5]{fontenc} 
\usepackage[vietnam,english]{babel} 
\fi
\usepackage{enumitem}

\usepackage[all]{xy}

\usepackage[colorlinks=true]{hyperref}
\hypersetup{allcolors=[rgb]{0.1,0.1,0.4}}

\author{Jean Goubault-Larrecq}

\address{Universit\'e Paris-Saclay, CNRS, ENS Paris-Saclay,
  Laboratoire M\'ethodes
  Formelles, 91190, Gif-sur-Yvette, France
}

\title[Distributing Retractions and Weak Distributive Laws]{Distributing Retractions, Weak Distributive Laws and
  Applications to Monads of Hyperspaces, Continuous Valuations and Measures}

\copyrightyear{2025}

\keywords{weak distributive law, monad, continuous valuation, Radon measure, hyperspace, Smyth hyperspace, Hoare powerspace, Plotkin hyperspace, quasi-lens, lens}
\amsclass{18C15, 54B20, 28A33, 46E27}
  
\eaddress{jgl@lmf.cnrs.fr}

\ifarxiv
\newcommand\acuteabreve{\rlap{\raisebox{.9ex}{$\scriptstyle\:'$}}{\u
    a}}
\else
\newcommand\acuteabreve{\'{\u a}} 
\fi
\newcommand\eqdef{\mathrel{\buildrel \text{def}\over=}}

\newcommand\Topcat{\mathbf{Top}}
\newcommand\Setcat{\mathbf{Set}}
\newcommand\Dcpo{\mathbf{Dcpo}}
\newcommand\Cat{\mathbf{Cat}}
\newcommand\Cont{\mathbf{Cont}}
\newcommand\SComp{\mathbf{SComp}}
\newcommand\KHaus{\mathbf{KHaus}}
\newcommand\dG{{\text{\textsf{d}}}}

\newcommand\Smyth{{\mathcal Q}}
\newcommand\Hoare{{\mathcal H}}
\newcommand\Plotkin{\mathcal P\ell} 
\newcommand\Plotkinn{\Plotkin_\Vt}
\newcommand\Vt{\mathsf{V}}
\newcommand\SV{\Smyth_{\Vt}}
\newcommand\HV{\Hoare_\Vt}
\newcommand\TEMleq{\sqsubseteq^{\text{TEM}}}
\newcommand\Val{{\mathbf V}}
\newcommand\Lform{{\mathcal L}}
\newcommand\one{{\mathbf 1}}
\newcommand{\interior}[1]{int ({#1})} 
\newcommand\diff{\smallsetminus}

\newcommand\Pred{\mathbb{P}}
\newcommand\Demon{{\mathtt{D}}}
\newcommand\Angel{{\mathtt{A}}}
\newcommand\Nature{{\mathtt{P}}}
\newcommand\DN{{\Demon\Nature}}
\newcommand\AN{{\Angel\Nature}}
\newcommand\ADN{{\Angel\Demon\Nature}}
\newcommand\pl{{\mathrm{p}}}

\newcommand\catc{{\mathbf{C}}} 
\newcommand\dto{\Rightarrow}
\newcommand{\identity}[1]{\mathrm{id}_{#1}}
\newcommand\id[1]{1_{#1}}
\newcommand\ext[1]{\underline{#1}}
\newcommand\upc{\mathop{\uparrow}}
\newcommand\dc{\mathop{\downarrow}}
\newcommand\Open{\mathcal O}
\newcommand{\real}{\mathbb{R}}
\newcommand\Rp{\real_+}
\newcommand\creal{\overline{\real}_+}
\newcommand\conv{\mathop{\mathrm{conv}}}
\newcommand\patch{{\text{\textsf{patch}}}}

\newtheorem{theorem}{Theorem} 
\newtheorem{lemma}[theorem]{Lemma}
\newtheorem{proposition}[theorem]{Proposition}
\newtheorem{definition}[theorem]{Definition}

\ifchanges
\newcommand{\red}[1]{\textcolor{red}{#1}}
\newenvironment{longred}{\par\color{red}}{\par}
\else
\newcommand{\red}[1]{#1}
\newenvironment{longred}{\par}{\par}
\fi

\newtheoremrm{rem}{Remark}

\begin{document}

\maketitle
\begin{abstract}
  Given two monads $S$, $T$ on a category where idempotents split, and
  a weak distributive law between them, one can build a combined monad
  $U$.  Making explicit what this monad $U$ is requires some effort.
  When we already have an idea what $U$ should be, we show how to
  recognize that $U$ is indeed the combined monad obtained from $S$
  and $T$: it suffices to exhibit what we call a distributing
  retraction of $ST$ onto $U$.  We show that distributing retractions
  and weak distributive laws are in one-to-one correspondence, in a
  2-categorical setting.  
  We give three
  applications, where $S$ is the Smyth, Hoare or Plotkin hyperspace
  monad, $T$ is a monad of continuous valuations, and $U$ is a monad
  of previsions or of forks, depending on the case.  As a byproduct,
  this allows us to describe the algebras of monads of superlinear,
  resp.\ sublinear previsions.  In the category of compact Hausdorff
  spaces, the Plotkin hyperspace monad is sometimes known as the
  Vietoris monad, the monad of probability valuations coincides with
  the Radon monad, and we infer that the associated combined monad is
  the monad of normalized forks.
\end{abstract}



\noindent
\begin{minipage}{0.25\linewidth}
  \includegraphics[scale=0.2]{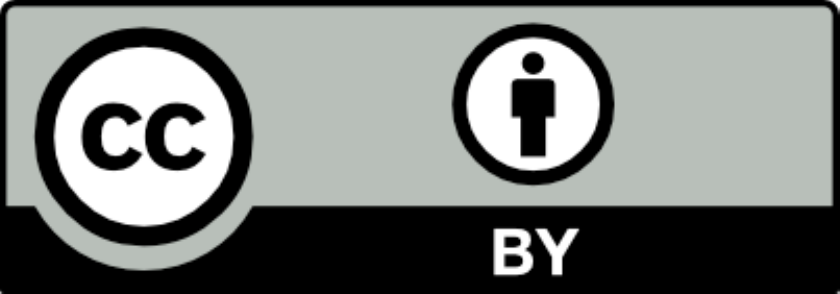}
\end{minipage}
\begin{minipage}{0.74\linewidth}
  \scriptsize
  For the purpose of Open Access, a CC-BY public copyright licence has
  been applied by the authors to the present document and will be
  applied to all subsequent versions up to the Author Accepted
  Manuscript arising from this submission.
\end{minipage}

\section{Introduction}
\label{sec:intro}

Jon Beck introduced the notion of distributive laws
$\lambda \colon TS \to ST$ of a monad $S$ over a monad $T$, and showed
that any such distributive law allowed for the construction of a
composite monad $ST$ \citep{Beck:distr}.  He also showed that there is
a one-to-one correspondence between distributive laws as above,
extensions of $T$ to the Kleisli category of $S$, and liftings of $S$
to the category of $T$-algebras.

It was once dreamed that one could perhaps use this construction in
order to define new monads, combining various side-effects in the
semantics of programming languages.  But distributive laws are rare,
and, as a case in point, there is no distributive law of the monad of
non-deterministic choice over the monad of discrete probability
measure on $\Setcat$ \citep{VW:distr}.  The same argument shows that
there is no distributive law of any of the monads of non-deterministic
choice (the Smyth powerdomain monad $\Smyth$, the Hoare powerdomain
monad $\Hoare$, the Plotkin powerdomain monad $\Plotkin$) over the
monad $\Val$ of continuous valuations (or $\Val_1$ [probability
valuations], or $\Val_{\leq 1}$ [subprobability valuations],
implementing probabilistic choice).

Gabriella B\"ohm was the first to study weakenings of the notion of
distributive law between monads \citep{Bohm:weak:monads}, in a general
$2$-categorical context.  Given two monads $S$ and $T$, a distributive
law $\lambda \colon TS \to ST$ is a natural transformation satisfying
four equations expressing compatibility with the units and
multiplications of each of the two monads.  \emph{Weak} distributive
laws are obtained by dropping one or the other of those equations.
Richard Garner observed that one of these kinds of weak distributive
laws, where compatibility with the unit of the $T$ monad is ignored,
could be applied to algebraic-effect type monads
\citep{Garner:weak:distr}.  This was pursused by Alexandre Goy in his
PhD thesis \citep{Goy:PhD}, by him and Daniela Petri\c san
\citep{goypetr-dp} and more recently by Quentin Aristote
\citep{Aristote:mono:WDL}.  This is the kind that we are interested in
here.

The main application we have in mind is the following.  Taking $S$ to
be the Smyth hyperspace monad $\SV$ over $\Topcat$, where $\SV X$ is
the collection of non-empty compact saturated subsets of a topological
space $X$, and $T$ to be the monad $\Val$ of continuous valuations, we
claim that there is a weak distributive law of $S$ over $T$.  This
yields a combined monad $U$, and we claim that $U$ is the monad
$\Pred_\DN$ of superlinear previsions of \citep{Gou-csl07}, once
properly recast as a monad on $\Topcat$, not just its full subcategory
$\Dcpo$.  We will come back to these notions in
Section~\ref{sec:smyth-hypersp-cont}.  The point is that, while
defining the weak distributive law is not complicated, checking that
it is a weak distributive law is more technical, and building the
combined monad and comparing it to the monad of superlinear previsions
is difficult.  (See \citep{JGL:wdistr}.)  Instead of carrying the
burden of building the combined monad $U$ from $S$ and $T$, we propose
to recognize $U$ as being obtained from $S$ and $T$ by some weak
distributive law (which will remain implicit), by exhibiting what we
call a distributing retraction of $ST$ onto $U$.  We show that
distributing retractions and weak distributive laws are two faces of
the same coin in Section~\ref{sec:distr-retr}.

The rest of the paper consists in applications.
Section~\ref{sec:smyth-hypersp-cont} is concerned with the three
monads $\SV$, $\Val_\bullet$ and $\Pred_\DN^\bullet$ mentioned above.
Section~\ref{sec:hoare-hypersp-cont} replaces $\SV$ by $\HV$, and
$\Pred_\DN^\bullet$ by the monad $\Pred_\AN^\bullet$ of sublinear
previsions of \citep{Gou-csl07}.  In those two sections, we take the
opportunity of characterizing the algebras of the $\Pred_\DN^\bullet$
and $\Pred_\AN^\bullet$ monads; this falls directly off the theory of
weak distributive laws.  Finally, Section~\ref{sec:plotk-powerd-cont}
is about the Plotkin hyperspace $\Plotkinn$ in lieu of $\SV$, and the
monad $\Pred_\ADN^\bullet$ of forks of \citep{Gou-csl07} in place of
$\Pred_\DN^\bullet$.  As a byproduct, we obtain (and generalize) the
weak distributive law of the Vietoris monad $\mathsf V_*$ over the
Radon monad $\mathsf R$ found by 
\citet{Aristote:mono:WDL} on the category $\KHaus$ of compact
Hausdorff spaces, and we show that the combined monad is
$\Pred_\ADN^1$.

We claim that it is significantly easier to recognize $U$ from a
distributing retraction than to build a weak distributive law between
$S$ and $T$, then the combined monad, and finally recognizing it as
the desired monad $U$.  The unpublished manuscript \citep{JGL:wdistr},
where we took the former approach, is long and technical, and quite a
bit more than the application sections of the present paper.

\section{Distributing retractions}
\label{sec:distr-retr}

While our main interest is on actual monads on a category, it is not
more complicated to reason in (strict) 2-categories \citep{KS:2cat}.
We write horizontal composition of 1-cells and 2-cells by
juxtaposition and vertical composition using $\circ$, so that our
notation will look like familiar notation for composition of functors
(e.g., $ST$) and of natural transformations (e.g.,
$\mu^T \circ T \eta^T$).  We write $\id \catc$ for the identity 1-cell
on the 0-cell $\catc$.  The \emph{interchange law} reads
$(g' \circ f') (g \circ f) = g'g \circ f'f$ for all 2-cells
$f \colon A \dto B$, $g \colon B \dto C$, $f' \colon A' \dto B'$ and
$g' \colon B' \dto C'$.  A 1-cell $F$ has an identity 2-cell, which we
also write as $F$.  As a special case of the interchange law, for all
1-cells $A, B, C, D \colon \catc \to \catc$ on the same 0-cell
$\catc$, for all 2-cells $f \colon A \dto B$ and $g \colon C \dto D$,
$g B \circ C f = (g \circ C) (B \circ f) = gf$, while
$D f \circ g A = (D \circ g) (f \circ A) = gf$, so
$g B \circ C f = D f \circ g A$ ($=gf$).  We will call the latter the
\emph{naturality} property of $g$.

A \emph{monad} in a 2-category $\mathcal K$ is a tuple
$(\catc, T, \eta^T, \mu^T)$ where $\catc$ is a 0-cell in $\mathcal K$,
$T \colon \catc \to \catc$ is a 1-cell,
$\eta^T \colon \id \catc \dto T$ and $\mu^T \colon TT \dto T$ are
2-cells subject to the following conditions:
$\mu^T \circ T \eta^T = T$, $\mu^T \circ \eta^T T = T$,
$\mu^T \circ T \mu^T = \mu^T \circ \mu^T T$.  We sometimes omit
$\catc$ and write $(T, \eta^T, \mu^T)$, saying that the latter is a
monad \emph{on} $\catc$, and we sometimes even call $T$ itself the
monad.  The 2-cell $\eta^T$ is the \emph{unit} of the monad $T$, and
$\mu^T$ is its \emph{multiplication}.

The usual monads on a category $\catc$ are the monads in the
2-category $\Cat$ of categories, functors and natural transformations.

Every monad $(T, \eta^T, \mu^T)$ admits an \emph{extension operation},
which maps every 1-cell $f \colon F \to TG$ to $f^{\ext T} \colon TF
\to TG$, defined as $\mu^T \circ T f$.  This will serve below.

A \emph{retraction} $\xymatrix{F \ar@<1ex>[r]^r & G \ar@<1ex>[l]^s}$
between two 1-cells $F$ and $G$ is given by two 2-cells $r$ and $s$ as
indicated, satisfying $r \circ s = G$.  In $\Cat$, a retraction is
better called a \emph{natural retraction}, in order to stress the fact
that it is comprised of two natural transformations that form
retractions at each object.

\begin{definition}
  \label{defn:STU:retr}
  Let $S$, $T$, $U$ be three monads on the same object $\catc$ of a
  2-category $\mathcal K$.  A \emph{distributing retraction} of $ST$
  onto $U$ is a retraction
  $\xymatrix{ST \ar@<1ex>[r]^r & U \ar@<1ex>[l]^s}$ satisfying the
  following laws:
  \begin{align}
    \label{eq:STU:eta}
    r \circ \eta^S \eta^T
    & = \eta^U \\
    \label{eq:STU:muS}
    r \circ \mu^S T & = \mu^U \circ ir
    \\
    \label{eq:STU:muT}
    r \circ S \mu^T & = \mu^U \circ rj, 
  \end{align}
  where $i \colon S \to U$ abbreviates $r \circ S \eta^T$ and $j
  \colon T \to U$ abbreviates $r \circ \eta^S T$, and:
  \begin{align}
    \label{eq:STU:etaS:conv}
    e \circ \eta^S T
    & = \eta^S T \\
    \label{eq:STU:SmuT:conv}
    e \circ S \mu^T
    & = S \mu^T \circ e T
    \\
    \label{eq:STU:muUS:conv}
    e \circ (s \circ \mu^U \circ j U)^{\ext S}
    & = (s \circ \mu^U \circ j U)^{\ext S} \circ e U,
  \end{align}
  where $\_^{\ext S}$ is the extension operation of the monad $S$ and
  $e$ abbreviates $s \circ r$.
\end{definition}

That probably looks complex, especially compared with the notion of
weak distributive law, which only requires three axioms (see below).
We will see in Section~\ref{sec:smyth-hypersp-cont} and subsequent
sections that the axioms of Definition~\ref{defn:STU:retr} are easy to
check in practice.  For now, let us observe that axioms
(\ref{eq:STU:eta}), (\ref{eq:STU:muS}) and (\ref{eq:STU:muT}) specify
compatibility conditions of the $r$ part with the units and
multiplications of the various monads.  The $s$ part is not involved,
and only intervenes in (\ref{eq:STU:etaS:conv}),
(\ref{eq:STU:SmuT:conv}) and (\ref{eq:STU:muUS:conv}) through the
idempotent $e \eqdef s \circ r$.  In applications to come, $e$ will
compute some form of convex hull, and one may think of
(\ref{eq:STU:etaS:conv}) as saying that $\eta^S T$ is convex-valued,
and of (\ref{eq:STU:SmuT:conv}) and (\ref{eq:STU:muUS:conv}) as saying
that $S \mu^T$ and $(s \circ \mu^U \circ j U)^{\ext S}$ preserve
\red{convex hulls}.
(In the case considered by Garner where $S$ is the
powerset monad, $T$ is the ultrafilter monad, and $U$ is the Vietoris
monad on the category of compact Hausdorff spaces, $e$ would compute
topological closure instead.)

A \emph{weak distributive law} \citep{Bohm:weak:monads,
  Garner:weak:distr} of a monad $S$ over a monad $T$ (on an object
$\catc$ in a 2-category $\mathcal K$) is a 1-cell
$\lambda \colon TS \to ST$ that satisfies the following laws:
\begin{align}
  \label{eq:distr:eta:1}
  \lambda \circ T \eta^S & = \eta^S T \\
  \label{eq:distr:mu:1}
  \lambda \circ T \mu^S & = \mu^S T \circ S \lambda \circ
                              \lambda S \\
  \label{eq:distr:mu:2}
  \lambda \circ \mu^T S & = S \mu^T \circ \lambda T \circ T \lambda.
\end{align}
A weak distributive law is a \emph{distributive law} if it
additionally satisfies the following:
\begin{align}
  \label{eq:distr:eta:2}
  \lambda \circ \eta^T S & = S \eta^T.
\end{align}

\subsection{Every distributing retraction defines a weak distributive
  law}
\label{sec:every-distr-retr}

Given two monads $T$ and $U$ on the same 0-cell $\catc$ of a
2-category $\mathcal K$, a \emph{monad morphism} from $T$ to $U$ is a
1-cell $f \colon T \to U$ such that $f \circ \eta^T = \eta^U$ and
$f \circ \mu^T = \mu^U \circ ff$.  There are several conventions in
the literature as to whether $f$ should be a 1-cell from $T$ to $U$ or
from $U$ to $T$, and this definition serves to clarify which notion we
mean.

\begin{lemma}
  \label{lemma:STU:retr}
  Let $S$, $T$, $U$ be three monads on the same 0-cell $\catc$ of a
  2-category $\mathcal K$, and
  $\xymatrix{ST \ar@<1ex>[r]^r & U \ar@<1ex>[l]^s}$ be a distributing
  retraction, with $i$ and $j$ defined as in
  Definition~\ref{defn:STU:retr}.  Then:
  \begin{enumerate}
  \item\label{it:i} $i$ is a monad morphism from $S$ to $U$;
  \item\label{it:j} $j$ is a monad morphism from $T$ to $U$;
  \item\label{it:sj} $s \circ j = \eta^S T$;
  \item\label{it:jeta} $j \circ \eta^T = \eta^U$;
  \item\label{it:muij} $\mu^U \circ ij = r$
  \end{enumerate}
\end{lemma}
\begin{proof}
  1. We have
  $i \circ \eta^S = r \circ S \eta^T \circ \eta^S = r \circ \eta^S
  \eta^T$ (by naturality of $\eta^S$) $= \eta^U$ (by
  (\ref{eq:STU:eta})).  Next,
  $i \circ \mu^S = r \circ S \eta^T \circ \mu^S$ (by definition of
  $i$) $= r \circ \mu^S T \circ SS\eta^T$ (by naturality of $\mu^S$)
  $= \mu^U \circ ir \circ SS \eta^T$ (by (\ref{eq:STU:muS}))
  $= \mu^U \circ i U \circ S r \circ SS\eta^T$ (by naturality of $i$)
  $= \mu^U \circ i U \circ S i$ (by definition of $i$)
  $= \mu^U \circ ii$ (by naturality of $i$).

  2.  Similarly,
  $j \circ \eta^T = r \circ \eta^S T \circ \eta^T = r \circ \eta^S
  \eta^T$ (by naturality of $\eta^S$) $= \eta^U$ by
  (\ref{eq:STU:eta}), and
  $j \circ \mu^T = r \circ \eta^S T \circ \mu^T$ (by definition of
  $j$) $= r \circ S \mu^T \circ \eta^S TT$ (by naturality of $\eta^S$)
  $= \mu^U \circ rj \circ \eta^S TT$ (by (\ref{eq:STU:muT}))
  $= \mu^U \circ U j \circ r T \circ \eta^S TT$ (by naturality of $r$)
  $= \mu^U \circ U j \circ j T$ (by definition of $j$)
  $= \mu^U \circ jj$ (by naturality of $j$).

  3.  This is just a rephrasing of (\ref{eq:STU:etaS:conv}).

  4.  This is a rephrasing of (\ref{eq:STU:eta}), using the naturality
  of $\eta^S$.

  5.  We first compose both sides of (\ref{eq:STU:muS}) by $S s$ on
  the right.  On the right-hand side,
  $\mu^U \circ ir \circ S s = \mu^U \circ i U \circ S r \circ S s =
  \mu^U \circ i U$, by naturality of $i$ and the fact that $r$ and $s$
  form a retraction.  Hence
  $\mu^U \circ i U = r \circ \mu^S T \circ S s$.  Using this,
  $\mu^U \circ i U \circ S j = r \circ \mu^S T \circ S s \circ S j = r
  \circ \mu^S T \circ S \eta^S T$ (by item~\ref{it:sj} above) $= r$
  (by the monad equations for $S$).  We conclude since
  $\mu^U \circ i U \circ S j = \mu^U \circ ij$, by naturality of $i$.
\end{proof}

\begin{proposition}
  \label{prop:STU:wdistr}
  Let $S$, $T$, $U$ be three monads on the same 0-cell $\catc$ of a
  2-category $\mathcal K$, and
  $\xymatrix{ST \ar@<1ex>[r]^r & U \ar@<1ex>[l]^s}$ be a distributing
  retraction.  We define $i \eqdef r \circ S \eta^T \colon S \to U$,
  $j \eqdef r \circ \eta^S T \colon T \to U$ and
  $e \eqdef s \circ r \colon ST \to ST$, as in
  Definition~\ref{defn:STU:retr}.  Then:
  \begin{align}
    \label{eq:lambda}
    \lambda \eqdef s \circ \mu^U \circ ji \colon TS \to ST
  \end{align}
  is a weak distributive law of $S$ over $T$.
\end{proposition}
\begin{proof}
  We first check (\ref{eq:distr:eta:1}).  We have
  $\lambda \circ T \eta^S = s \circ \mu^U \circ ji \circ T \eta^S = s
  \circ \mu^U \circ j U \circ T i \circ T \eta^S$, by naturality of
  $j$.  Now
  $i \circ \eta^S = r \circ S \eta^T \circ \eta^S = r \circ \eta^S
  \eta^T$ (by naturality of $\eta^S$) $= \eta^U$ (by
  (\ref{eq:STU:eta})), so
  $\lambda \circ T \eta^S = s \circ \mu^U \circ j U \circ T \eta^U$.
  This is equal to $s \circ \mu^U \circ U \eta^U \circ j$ by
  naturality of $j$, hence to $s \circ j$ since $U$ is a monad, and
  then to $\eta^S T$ by Lemma~\ref{lemma:STU:retr}, item~\ref{it:sj}.

  The other equalities need more work.

  We claim that:
  \begin{align}
    \label{eq:STU:wdistr:ii}
    r \circ \lambda \circ T \mu^S
    &
      \text{ and }
      r \circ \mu^S T \circ S \lambda \circ \lambda S
      \text{ are both equal to }
      \mu^U \circ \mu^U U \circ jii.
  \end{align}
  Indeed,
  \begin{align*}
    r \circ \lambda \circ T \mu^S
    & = \mu^U \circ ji \circ T \mu^S_X
    & \text{since $r, s$ form a retraction} \\
    & = \mu^U \circ j (i \circ \mu^S)
    & \text{by the interchange law} \\
    & = \mu^U \circ j (\mu^U \circ ii)
    & \text{since $i$ is a monad morphism
      (Lemma~\ref{lemma:STU:retr}, item~\ref{it:i})} \\
    & = \mu^U \circ U \mu^U \circ jii
    & \text{by the interchange law} \\
    & = \mu^U \circ \mu^U U \circ jii
  \end{align*}
  since $U$ is a monad, while:
  \begin{align*}
    & r \circ \mu^S T \circ S \lambda \circ \lambda S \\
    & = \mu^U \circ ir \circ S \lambda \circ
      \lambda S
    & \text{by (\ref{eq:STU:muS})} \\
    & = \mu^U \circ i U \circ S r \circ S \lambda \circ
      \lambda S
    & \text{by  naturality of $i$} \\
    & = \mu^U \circ i U \circ S \mu^U \circ S ji \circ \lambda S
    & \text{by def.\ of $\lambda$ and since $r, s$ form a retraction} \\
    & = \mu^U \circ U \mu^U \circ i UU \circ S ji \circ \lambda S
    & \text{by naturality of $i$} \\
    & = \mu^U \circ \mu^U U \circ i UU
      \circ S ji \circ \lambda S
    & \text{since $U$ is a monad} \\
    & = \mu^U \circ \mu^U U \circ ij U
      \circ ST i \circ \lambda S
    & \text{by the interchange law} \\
    & = \mu^U \circ r U
      \circ ST i \circ \lambda S
    & \text{by Lemma~\ref{lemma:STU:retr}, item~\ref{it:muij}} \\
    & = \mu^U \circ U i \circ r S \circ \lambda S
    & \text{by naturality of $r$} \\
    & = \mu^U \circ U i \circ \mu^U S \circ ji S
    & \text{by def.\ of $\lambda$ and since $r, s$ form a retraction} \\
    & = \mu^U \circ \mu^U U \circ UU i \circ ji S
    & \text{by naturality of $i$} \\
    & = \mu^U \circ \mu^U U \circ jii
    & \text{by the interchange law.}
  \end{align*}

  Following the intuition we sketched earlier, let us say that a
  1-cell $f \colon Z \to STY$ is \emph{convex-valued} if and only if
  $e \circ f = f$.

  Then $\lambda$ is convex-valued, since
  $e \circ \lambda = s \circ r \circ \lambda = s \circ \mu^U \circ ji$
  (by definition of $\lambda$ and since $r, s$ form a retraction)
  $= \lambda$.

  Both sides of (\ref{eq:distr:mu:1}) are convex-valued, too.  The
  left-hand side is because $\lambda$ is convex-valued, and therefore
  so is any composition $\lambda \circ f$.  The right-hand side is
  equal to
  $\mu^S \circ S (s \circ \mu^U \circ j U) \circ ST i \circ \lambda S$
  (using the interchange law), namely to
  $(s \circ \mu^U \circ j U)^{\ext S} \circ ST i \circ \lambda S$.
  Now:
  \begin{align*}
    & e \circ (s \circ \mu^U \circ j U)^{\ext S} \circ ST i \circ
    \lambda S \\
    & = (s \circ \mu^U \circ j U)^{\ext S} \circ e \circ ST i \circ
      \lambda S
    & \text{by (\ref{eq:STU:muUS:conv})} \\
    & = (s \circ \mu^U \circ j U)^{\ext S} \circ ST i \circ e \circ
      \lambda S
    & \text{by naturality of $e$} \\
    & = (s \circ \mu^U \circ j U)^{\ext S} \circ ST i \circ
      \lambda S
    & \text{since $\lambda$ is convex-valued,}
  \end{align*}
  and we have noticed earlier that the latter is equal to the
  right-hand side of (\ref{eq:distr:mu:1}).

  From this, we deduce that (\ref{eq:distr:mu:1}) holds.  In general,
  if $f$ and $g$ are convex-valued, it suffices to show that
  $r \circ f = r \circ g$ in order to assert that $f=g$.  Indeed, if
  $r \circ f = r \circ g$, then
  $e \circ f = s \circ r \circ f = s \circ r \circ g = e \circ g$, and
  since $f$ and $g$ are convex-valued, this entails $f=g$.  In the
  present case, $f=\lambda \circ T \mu^S$,
  $g = \mu^S T \circ S \lambda \circ \lambda S$, and
  $r \circ f = r \circ g$ is (\ref{eq:STU:wdistr:ii}).

  For the last of the three equations, we claim that:
  \begin{align}
    \label{eq:STU:wdistr:vi}
    r \circ \lambda \circ \mu^T S
    & \text{ and }
      r \circ S \mu^T \circ \lambda T \circ T \lambda
      \text{ are both equal to }
      \mu^U \circ \mu^UU \circ jji.
  \end{align}
  For the first one,
  \begin{align*}
    & r \circ \lambda \circ \mu^T S \\
    & = \mu^U \circ ji \circ \mu^T S
    & \text{by def.\ of $\lambda$ and since $r, s$ is a retraction} \\
    & = \mu^U \circ j U \circ T i \circ \mu^T S
    & \text{by naturality of $j$} \\
    & = \mu^U \circ j U \circ \mu^T U \circ TT i
    & \text{by naturality of $\mu^T$} \\
    & = \mu^U \circ \mu^U U \circ jj U \circ TT i
    & \text{since $j$ is a monad morphism (Lemma~\ref{lemma:STU:retr},
      item~\ref{it:j})} \\
    & = \mu^U \circ \mu^U U \circ jji
    & \text{by the interchange law.}
  \end{align*}
  For the second one,
  \begin{align*}
    & r \circ S \mu^T \circ \lambda T \circ T \lambda \\
    & = \mu^U \circ rj \circ \lambda T \circ T \lambda
    & \text{by (\ref{eq:STU:muT})} \\
    & = \mu^U \circ U j \circ r T \circ \lambda T \circ T \lambda
    & \text{by naturality of $r$} \\
    & = \mu^U \circ U j \circ
      \mu^U T \circ ji T
      \circ T \lambda
    & \text{by def. of $\lambda$ and since $r, s$ form a retraction} \\
    & = \mu^U \circ \mu^U U \circ UU j \circ ji T
      \circ T \lambda
    & \text{by naturality of $\mu^U$} \\
    & = \mu^U \circ \mu^U U \circ U ij \circ j ST \circ T \lambda
    & \text{by the interchange law} \\
    & = \mu^U \circ \mu^U U \circ U ij
      \circ U \lambda \circ j TS
    & \text{by naturality of $j$} \\
    & = \mu^U \circ U \mu^U \circ U ij
      \circ U \lambda \circ j TS
    & \text{since $U$ is a monad} \\
    & = \mu^U \circ U r
      \circ U \lambda \circ j TS
    & \text{by Lemma~\ref{lemma:STU:retr}, item~\ref{it:muij}} \\
    & = \mu^U \circ U \mu^U \circ U ji
      \circ j TS
    & \text{by def.\ of $\lambda$ and since $r, s$ form a retraction} \\
    & = \mu^U \circ \mu^U U \circ jji
    & \text{by naturality of $j$.}
  \end{align*}
  By (\ref{eq:STU:wdistr:vi}),
  $r \circ \lambda \circ \mu^T S = r \circ S \mu^T \circ \lambda T
  \circ T \lambda$, namely $r$ composed with either side of
  (\ref{eq:distr:mu:2}) yields the same value.  Hence, in order to
  show (\ref{eq:distr:mu:2}), it suffices to show that both sides of
  (\ref{eq:distr:mu:2}) are convex-valued.

  We remember that $\lambda$ is convex-valued, hence also the
  left-hand side $\lambda \circ \mu^T S$ of (\ref{eq:distr:mu:2}).
  For the right-hand side, we have:
  \begin{align*}
    e \circ S \mu^T \circ \lambda T \circ T \lambda
    & = S \mu^T \circ e
      \circ \lambda T \circ T \lambda
    & \text{by (\ref{eq:STU:SmuT:conv})} \\
    & =  S \mu^T \circ \lambda T \circ T \lambda,
  \end{align*}
  since $\lambda$ is convex-valued.
\end{proof}

While we are at it, we observe the following.
\begin{proposition}
  \label{prop:STU:wdistr:distr}
  With the same assumptions as in Proposition~\ref{prop:STU:wdistr},
  $\lambda$ is a distributive law if and only if $e$ is the identity,
  if and only if $r, s$ is an isomorphism.
\end{proposition}
\begin{proof}
  We first notice that $\lambda \circ \eta^T S = e \circ S \eta^T$.
  Indeed:
  \begin{align*}
    \lambda \circ \eta^T S
    & = s \circ \mu^U \circ ji \circ \eta^T S \\
    & = s \circ \mu^U \circ j U \circ S i \circ \eta^T S
    & \text{by naturality of $j$} \\
    & = s \circ \mu^U \circ j U \circ \eta^T U \circ i
    & \text{by naturality of $\eta^T$} \\
    & = s \circ \mu^U \circ r U \circ \eta^S TU \circ \eta^T U \circ i
    & \text{by definition of $j$} \\
    & = s \circ \mu^U \circ r U \circ \eta^S \eta^T U \circ i
    & \text{by naturality of $\eta^S$} \\
    & = s \circ \mu^U \circ \eta^U U \circ i
    & \text{by (\ref{eq:STU:eta})} \\
    & = s \circ i
    & \text{since $U$ is a monad} \\
    & = s \circ r \circ S \eta^T
    & \text{by definition of $i$} \\
    & = e \circ S \eta^T
    & \text{by definition of $e$.}
  \end{align*}

  If $\lambda$ is a distributive law, then (\ref{eq:distr:eta:2})
  holds, so $\lambda \circ \eta^T S = S \eta^T$, and therefore
  $S \eta^T = e \circ S \eta^T$.  Then
  $S \mu^T \circ e T \circ S \eta^T T = S \mu^T \circ S \eta^T T$ is
  the identity (by the monad equations for $T$), but also
  $S \mu^T \circ e T \circ S \eta^T T = e \circ S \mu^T \circ S \eta^T
  T$ (by (\ref{eq:STU:SmuT:conv})) $= e$, since $T$ is a monad.
  Therefore $e$ is the identity, or equivalently, $r, s$ form an
  isomorphism.

  Conversely, if $e$ is the identity, then $\lambda \circ \eta^T S = e
  \circ S \eta^T$ means that $\lambda \circ \eta^T S = S \eta^T$,
  which is exactly (\ref{eq:distr:eta:2}); hence $\lambda$ is a
  distributive law.
\end{proof}

We also note that a distributing retraction defines the monad $U$ in a
unique way from $S$ and $T$.
\begin{proposition}
  \label{prop:STU:unique}
  With the same assumptions as in Proposition~\ref{prop:STU:wdistr},
  \begin{align*}
    \eta^U & \eqdef r \circ \eta^S \eta^T \\
    \mu^U & \eqdef r \circ \mu^S \mu^T \circ S \lambda T \circ ss.
  \end{align*}
\end{proposition}
\begin{proof}
  The first equation is (\ref{eq:STU:eta}).  For the second one, we
  observe that:
  \begin{align*}
    r \circ \mu^S \mu^T
    & = r \circ \mu^S T \circ SS \mu^T
    & \text{by naturality of $\mu^S$} \\
    & = \mu^U \circ ir \circ SS \mu^T
    & \text{by (\ref{eq:STU:muS})} \\
    & = \mu^U \circ U r \circ i ST \circ SS \mu^T
    & \text{by naturality of $i$} \\
    & = \mu^U \circ U r \circ US \mu^T \circ i STT
    & \text{by naturality of $i$} \\
    & = \mu^U \circ U \mu^U \circ U rj \circ i STT
    & \text{by (\ref{eq:STU:muT})} \\
    & = \mu^U \circ U \mu^U \circ irj
    & \text{by naturality of $i$} \\
    & = \mu^U \circ \mu^U U \circ irj
    & \text{since $U$ is monad.}
  \end{align*}
  Then:
  \begin{align*}
    & r \circ \mu^S \mu^T \circ S \lambda T  \\
    & = \mu^U \circ \mu^U U \circ irj \circ S s T \circ S \mu^U T \circ S ji T
    & \text{by the above and def.\ of $\lambda$} \\
    & = \mu^U \circ \mu^U U \circ i U j \circ S \mu^U T \circ S ji T
    & \text{by the interchange law, $(r, s)$ is a retraction}
    \\
    & = \mu^U \circ \mu^U U \circ U \mu^U U \circ i UU j  \circ S ji T
    & \text{by the interchange law} \\
    & = \mu^U \circ \mu^U U \circ \mu^U UU \circ ijij 
    & \text{since $U$ is a monad} \\
    & = \mu^U \circ U \mu^U \circ \mu^U UU \circ ijij 
    & \text{since $U$ is a monad} \\
    & = \mu^U \circ \mu^U \mu^U \circ ijij 
    & \text{by naturality of $\mu^U$} \\
    & = \mu^U \circ rr
    & \text{by Lemma~\ref{lemma:STU:retr}, item~\ref{it:muij}.}
  \end{align*}
  Composing with $ss$ on the right, and since $r, s$ form a retraction,
  it follows that
  $r \circ \mu^S \mu^T \circ S \lambda T \circ ss = \mu^U$.
\end{proof}

\subsection{Every weak distributive law defines a distributing
  retraction}
\label{sec:every-weak-distr}

Given two monads $S$ and $T$ on the same 0-cell $\catc$ of a
2-category $\mathcal K$, and a weak distributive law $\lambda$ of $S$
over $T$, there is a combined monad $U$ that one obtains from $ST$ by
splitting idempotents \citep{Bohm:weak:monads}.  Let us recall the
construction.

An idempotent (on $\catc$) is a 2-cell $e \colon F \to F$ such that
$e \circ e = e$.  Given a retraction
$\xymatrix{F \ar@<1ex>[r]^r & G \ar@<1ex>[l]^s}$, $s \circ r$ is an
idempotent.  Conversely, given an idempotent $e$, a \emph{splitting}
of $e$ is a retraction $r, s$ such that $e = s \circ r$.

Given a weak distributive law of $S$ over $T$, we form the following
idempotent:
\begin{align}
  \label{eq:e}
  e & \eqdef S \mu^T \circ \lambda T \circ \eta^T ST
\end{align}
from $ST$ to $ST$.  This is an idempotent because of the following.
\begin{align*}
  e \circ e
  & = S \mu^T \circ \lambda T \circ \eta^T ST
    \circ S \mu^T \circ \lambda T \circ \eta^T ST
  \\
  & = S \mu^T \circ \lambda T \circ TS \mu^T
    \circ \eta^T STT \circ \lambda T \circ \eta^T ST
  & \text{by naturality of $\eta^T$} \\
  & = S \mu^T \circ ST \mu^T \circ \lambda TT
    \circ \eta^T STT \circ \lambda T \circ \eta^T ST
  & \text{by naturality of $\lambda$} \\
  & = S \mu^T \circ S \mu^T T \circ \lambda TT
    \circ \eta^T STT \circ \lambda T \circ \eta^T ST
  & \text{since $T$ is a monad} \\
  & = S \mu^T \circ S \mu^T T \circ \lambda TT
    \circ T \lambda T \circ \eta^T TST \circ \eta^T ST
  & \text{by naturality of $\eta^T$} \\
  & = S \mu^T \circ \lambda T \circ \mu^T ST
    \circ \eta^T TST \circ \eta^T ST
  & \text{by (\ref{eq:distr:mu:2})} \\
  & = S \mu^T \circ \lambda T \circ \eta^T ST
  & \text{since $T$ is a monad} \\
  & = e.
\end{align*}

Assuming a splitting $\xymatrix{ST \ar@<1ex>[r]^r & U \ar@<1ex>[l]^s}$
of $e$, the \emph{combined monad} $(U, \eta^U, \mu^U)$ is defined by:
\begin{align}
  \label{eq:etaU}
  \eta^U & \eqdef r \circ \eta^S \eta^T \\
  \label{eq:muU}
  \mu^U & \eqdef r \circ \mu^S \mu^T \circ S \lambda T \circ ss,
\end{align}
while $U$ itself is given by the splitting.

\begin{remark}
  \label{rem:U}
  In a category $\catc$ (a 0-cell in the 2-category $\Cat$), $e$ is a
  natural transformation from $ST$ to $ST$, and it is enough to find a
  \emph{objectwise splitting} of $e$, namely a collection of
  retractions
  $\xymatrix{STX \ar@<1ex>[r]^{r_X} & UX \ar@<1ex>[l]^{s_X}}$ such
  that $s_X \circ r_X = e_X$ for each object $X$.  This defines
  objects $UX$, and then $U$ extends to a functor (a 1-cell) by the
  requirement that for every morphism $f \colon X \to Y$,
  $Uf \eqdef r_Y \circ ST f \circ s_X$.  It is an easy exercise to
  show that $U$ is indeed a functor, and that $r$ and $s$ are natural
  transformations (2-cells) from $ST$ to $U$ and from $U$ to $ST$
  respectively.
\end{remark}

\begin{lemma}
  \label{lemma:lambda:aux}
  Let $S$ and $T$ be two monads on the same 0-cell $\catc$ of a
  2-category $\mathcal K$, and a weak distributive law $\lambda$ of
  $S$ over $T$.  Let $e$ be defined as in (\ref{eq:e}).  Then:
  \begin{enumerate}
  \item\label{it:vi} $e \circ \lambda = \lambda$;
  \item\label{it:vii} $e \circ \eta^S T = \eta^S T$;
  \item\label{it:viii} $e \circ S \mu^T = S \mu^T \circ e T$.
  \item\label{it:ix} $e \circ \mu^S T \circ S \lambda = \mu^S \circ S
    \lambda \circ e S$;
  \item\label{it:xii}
    $S \mu^T \circ \lambda T \circ T e = S \mu^T \circ \lambda T$.
  \end{enumerate}
\end{lemma}
In other words, reading ``$e \circ f = f$'' as ``$f$ is
convex-valued'', $e \circ f = f \circ e$ (or $= f \circ e T$, or
$= f \circ e S$) as ``$f$ preserves convexity'',
the first four items above state that $\lambda$ and $\eta^S T$ are
convex-valued, and that $S \mu^T$ and $\mu^S T \circ S \lambda$
(namely, the extension $\lambda^{\ext S}$) preserve convexity.

\begin{proof}
  \ref{it:vi}.
  \begin{align*}
    e \circ \lambda
    & = S \mu^T \circ \lambda T \circ \eta^T ST \circ \lambda
    & \text{by definition of $e$}
    \\
    & = S \mu^T \circ \lambda T \circ T \lambda \circ
      \eta^T TS
    & \text{by naturality of $\eta^T$} \\
    & = \lambda \circ \mu^T S \circ \eta^T TS
    & \text{by (\ref{eq:distr:mu:2})} \\
    & = \lambda
    & \text{since $T$ is a monad.}
  \end{align*}

  \ref{it:vii}.
  \begin{align*}
    e \circ \eta^S T
    & = S \mu^T \circ \lambda T \circ \eta^T ST \circ
      \eta^S T
    & \text{by definition of $e$}
    \\
    & = S \mu^T \circ \lambda T \circ T \eta^S T \circ
      \eta^T T
    & \text{by naturality of $\eta^T$} \\
    & = S \mu^T \circ \eta^S TT \circ \eta^T T
    & \text{by (\ref{eq:distr:eta:1})} \\
    & = S \mu^T \circ S \eta^T T \circ \eta^S T
    & \text{by naturality of $\eta^S$} \\
    & = \eta^S T
    & \text{since $T$ is a monad.}
  \end{align*}

  \ref{it:viii}.
  \begin{align*}
    e \circ S \mu^T
    & = S \mu^T \circ \lambda T \circ \eta^T ST \circ S
      \mu^T
    & \text{by definition of $e$}
    \\
    & = S \mu^T \circ \lambda T \circ 
      TS \mu^T \circ \eta^T STT
    & \text{by naturality of $\eta^T$} \\
    & = S \mu^T \circ ST \mu^T \circ \lambda TT \circ 
      \eta^T STT
    & \text{by naturality of $\lambda$} \\
    & = S \mu^T \circ S \mu^T T \circ \lambda TT \circ 
      \eta^T STT
    & \text{since $T$ is a monad} \\
    & = S \mu^T \circ e T
    & \text{by definition of $e$.}
  \end{align*}

  \ref{it:ix}.
  \begin{align*}
    & \mu^S T \circ S \lambda \circ e S \\
    & = \mu^S T \circ S \lambda \circ S \mu^T S
      \circ \lambda TS \circ \eta^T STS
    & \text{by definition of $e$} \\
    & = \mu^S T \circ  SS \mu^T \circ S \lambda T \circ ST \lambda
      \circ \lambda TS \circ \eta^T STS
    & \text{by (\ref{eq:distr:mu:2})} \\
    & = \mu^S T \circ
      SS \mu^T \circ S \lambda T \circ \lambda ST \circ TS \lambda
      \circ \eta^T STS
    & \text{by naturality of $\lambda$} \\
    & = \mu^S T \circ
      SS \mu^T \circ S \lambda T \circ \lambda ST \circ
      \eta^T SST \circ S \lambda
    & \text{by naturality of $\eta^T$} \\
    & = 
      S \mu^T \circ \mu^S TT \circ
      S \lambda T \circ \lambda ST \circ
      \eta^T SST \circ S \lambda
    & \text{by naturality of $\mu^S$} \\
    & = 
      S \mu^T \circ \lambda T \circ T \mu^S T \circ
      \eta^T SST \circ S \lambda
    & \text{by (\ref{eq:distr:mu:1})} \\
    & = 
      S \mu^T \circ \lambda T \circ
      \eta^T ST \circ \mu^S T \circ
      S \lambda
    & \text{by naturality of $\eta^T$} \\
    & = e \circ \mu^S T \circ S \lambda
    & \text{by definition of $e$.}
  \end{align*}

  \ref{it:xii}.
  \begin{align*}
    S \mu^T \circ \lambda T \circ T e
    & = S \mu^T \circ \lambda T \circ TS \mu^T
      \circ T \lambda T \circ T \eta^T ST
    & \text{by definition of $e$} \\
    & = S \mu^T  \circ ST \mu^T\circ \lambda TT
      \circ T \lambda T \circ T \eta^T ST
    & \text{by naturality of $\lambda$} \\
    & = S \mu^T  \circ S \mu^T T \circ \lambda TT
      \circ T \lambda T \circ T \eta^T ST
    & \text{since $T$ is a monad} \\
    & = S \mu^T  \circ \lambda T \circ \mu^T ST
      \circ T \eta^T ST
    & \text{by (\ref{eq:distr:mu:2})} \\
    & = S \mu^T  \circ \lambda T
    & \text{since $T$ is a monad.}
  \end{align*}
\end{proof}

\begin{lemma}
  \label{lemma:ABC}
  Let $S$ and $T$ be two monads on the same 0-cell $\catc$ of a
  2-category $\mathcal K$, and a weak distributive law $\lambda$ of
  $S$ over $T$.  Let $e$ be defined as in (\ref{eq:e}), and let us define:
  \begin{align}
    \label{eq:A}
    A & \eqdef \mu^S \mu^T \circ S \lambda T \colon STST \dto ST \\
    \label{eq:B}
    B & \eqdef A \circ ST s \colon STU \dto ST\\
    \label{eq:C}
    C & \eqdef A \circ s ST \colon UST \dto ST,
  \end{align}
  so that $\mu^U = r \circ A \circ ss$ by (\ref{eq:muU}), or
  equivalently $\mu^U = r \circ B \circ s U$, or equivalently
  $\mu^U = r \circ C \circ U s$.  Then:
  \begin{enumerate}
  \item\label{it:x} $A \circ e ST = e \circ A$;
  \item\label{it:xi} $B \circ e U = e \circ B$;
  \item\label{it:xvi} $A \circ ST \eta^S T = S \mu^T$.
  \end{enumerate}
\end{lemma}

\begin{proof}
  \ref{it:x}.
    \begin{align*}
    A \circ e ST
    & = S \mu^T \circ \mu^S TT \circ S \lambda T
      \circ e ST
    & \text{by def.\ of $A$ and naturality of $\mu^S$} \\
    & = S \mu^T \circ e T \circ \mu^S TT \circ S \lambda T
    & \text{by Lemma~\ref{lemma:lambda:aux}, item~\ref{it:ix}} \\
    & = e \circ  S \mu^T \circ \mu^S TT \circ S \lambda T
    & \text{by Lemma~\ref{lemma:lambda:aux}, item~\ref{it:viii}} \\
    & = e \circ A
    & \text{by def.\ of $A$ and naturality of $\mu^S$.}
  \end{align*}

  \ref{it:xi}.
  \begin{align*}
    B \circ e U
    & = A \circ ST s \circ e U
    & \text{by definition of $B$} \\
    & = A \circ e ST \circ ST s
    & \text{by naturality of $e$} \\
    & = e \circ A \circ ST s = e \circ B
    & \text{by definition of $B$.}
  \end{align*}

    \ref{it:xvi}.
    \begin{align*}
      A \circ ST \eta^S T
      & = \mu^S \mu^T \circ S \lambda T
        \circ ST \eta^S T
      & \text{by definition of $A$} \\
      & = S \mu^T \circ \mu^S TT \circ S \lambda T
        \circ ST \eta^S T
      & \text{by naturality of $\mu^S$} \\
      & = S \mu^T \circ \mu^S TT \circ S \eta^S T
      & \text{by (\ref{eq:distr:eta:1})} \\
      & = S \mu^T
      & \text{since $S$ is a monad.}
    \end{align*}

\end{proof}

\begin{remark} 
  \label{rem:U:monad}
  We will not need item~\ref{it:x} of Lemma~\ref{lemma:ABC} in order
  to show that $r, s$ form a distributing retraction.  However, this
  is one step in one possible proof that $U$ is a monad (in its final
  step, see item~\ref{it:monad:3} below).  Contrarily to the
  construction of \citep{Bohm:weak:monads}, this proof does not require
  that $\mathcal K$ admits Eilenberg-Moore constructions for monads,
  nor that idempotent 2-cells in $\mathcal K$ split---only that the
  specific idempotent $e$ of (\ref{eq:e}) splits.  Since that
  alternate proof is a bit lengthy and does not bring much, we only
  sketch it here.  We show that:
  \begin{enumerate}[label=(\alph*)]
  \item\label{it:xii:bis}
    $\mu^S \mu^T \circ S \lambda T \circ ST e = \mu^S \mu^T \circ S
    \lambda T$ (naturality of $\mu^S$, Lemma~\ref{lemma:ABC}
    item~\ref{it:xii}, naturality of $\mu^S$);
  \item\label{it:xiii:prelim}
    $e \circ \mu^S T = \mu^S \mu^T \circ S \lambda T \circ \lambda ST
    \circ \eta^T SST$ (naturality of $\eta^T$, (\ref{eq:distr:mu:1}),
    naturality of $\mu^S$);
  \item\label{it:xiii} $e \circ \mu^S T \circ S e = e \circ \mu^S T$
    (\ref{it:xiii:prelim}, naturality of $\eta^T$, of $\lambda$, then
    \ref{it:xii:bis} and \ref{it:xiii:prelim});
  \item\label{it:xiv} $A \circ \eta^S \eta^T ST = e$
    (definition of $A$, naturality of $\mu^S$ and $\eta^S$, $S$ is
    a monad, definition of $e$);
  \item\label{it:xv} $B \circ \eta^S \eta^T U = s$
    (definition of $B$, naturality of $\eta^S \eta^T$, \ref{it:xiv},
    definition of $e$, $(r, s)$ is a retraction);
  \item\label{it:xvii} $A \circ ST \eta^S \eta^T$ is the identity
    (naturality of $\eta^S$, Lemma~\ref{lemma:ABC} item~\ref{it:xvi},
    $T$ is a monad);
  \item\label{it:xviii:prelim} $s \circ \eta^U = \eta^S \eta^T$
    (definition of $\eta^U$ (\ref{eq:etaU}), naturality of $\eta^S$,
    Lemma~\ref{lemma:lambda:aux} item~\ref{it:vii}, naturality of $\eta^S$);
  \item\label{it:xviii} $C \circ U s \circ U \eta^U = s$
    (\ref{it:xviii:prelim}, definition of $C$, naturality of $s$, \ref{it:xvii});
  \item\label{it:xix:1}
    $S \lambda T \circ ST s \circ A U = \mu^S S \mu^T T \circ D$,
    where
    $D \eqdef SS \lambda TT \circ S \lambda STT \circ STS \lambda T
    \circ STST s$ (definition of $A$, naturality of $\mu^S \mu^T$, of
    $\lambda$, of $\mu^S$, (\ref{eq:distr:mu:2}), naturality of
    $\mu^S$, of $\lambda$, of $\mu^S$, and definition of $D$);
  \item\label{it:xix:2}
    $S \lambda T \circ ST A \circ STST s = S \mu^S T \mu^T \circ D$
    (definition of $A$, naturality of $\mu^S$, of $\lambda$,
    (\ref{eq:distr:mu:1}), naturality of $\mu^S$, definition of $D$);
  \item\label{it:xix:3}
    $\mu^S \mu^T \circ \mu^S S \mu^T T = \mu^S \mu^T \circ S \mu^S T
    \mu^T$ (interchange law, $\mu^S$ and $\mu^T$ are monads);
  \item\label{it:xix} $B \circ A U = A \circ ST A \circ STST s$
    (definition of $B$, \ref{it:xix:1}, \ref{it:xix:2},
    \ref{it:xix:3}, definition of $A$);
  \item first monad law: $\mu^U \circ \eta^U U$ is the identity (write
    $\mu^U$ as $r \circ B \circ s U$, then use the definition of
    $\eta^U$, (\ref{eq:etaU}), $e = s \circ r$, Lemma~\ref{lemma:ABC}
    item~\ref{it:xi}, $r \circ e = r$, \ref{it:xv}, the fact that $(r, s)$ is a retraction);
  \item second monad law: $\mu^U \circ U \eta^U$ is the identity
    (write $\mu^U$ as $r \circ C \circ U s$, then use \ref{it:xviii}
    and the fact that $(r, s)$ is a retraction);
  \item\label{it:monad:3} third monad law:
    $\mu^U \circ \mu^U U = \mu^U \circ U \mu^U$ (simplify the
    left-hand side to $r \circ A \circ ST A \circ sss$ using
    $\mu^U = r \circ B \circ s U$, $\mu^U = r \circ A \circ ss$, the
    definition of $e$, Lemma~\ref{lemma:ABC} item~\ref{it:xi}, the
    fact that $(r, s)$ is a retraction, \ref{it:xix}, and the
    naturality of $ss$, and simplify the right-hand side to the same
    value using $\mu^U = r \circ B \circ ss$, the naturality of $s$,
    $e = s \circ r$, Lemma~\ref{lemma:ABC} item~\ref{it:x}, the
    interchange law and $e \circ s = s$).
    %
  \end{enumerate}
\end{remark}

\begin{proposition}
  \label{prop:wdistr:STU}
  Let $S$ and $T$ be two monads on the same 0-cell $\catc$ of a
  2-category $\mathcal K$, and a weak distributive law $\lambda$ of
  $S$ over $T$.  Let $e$ be defined as in (\ref{eq:e}), and let $\xymatrix{ST
    \ar@<1ex>[r]^r & U \ar@<1ex>[l]^s}$ be a splitting of $e$.
  Then:
  \begin{enumerate}
  \item\label{it:STU:1} $(U, \eta^U, \mu^U)$ is a monad, where $\eta^U$ and $\mu^U$
    are defined in (\ref{eq:etaU}) and (\ref{eq:muU}) respectively;
  \item\label{it:STU:2} $r, s$ form a distributing retraction of $ST$ over $U$;
  \item\label{it:STU:3} $\lambda$ is equal to $s \circ \mu^U \circ ji$.
  \end{enumerate}
\end{proposition}
\begin{proof}
  \ref{it:STU:1}.  This is Remark~\ref{rem:U:monad}.

  \ref{it:STU:2}.  We first claim that
  $B \circ S \eta^T U = \mu^S T \circ S s$, where $B$ is defined as in
  Lemma~\ref{lemma:ABC}.  Indeed,
  \begin{align*}
    B \circ S \eta^T U
    & = \mu^S \mu^T \circ S \lambda T \circ ST s
      \circ S \eta^T U
    & \text{by definition of $B$} \\
    & = \mu^S T \circ SS \mu^T  \circ S \lambda T \circ ST s
      \circ S \eta^T U
    & \text{by naturality of $\mu^S$} \\
    & = \mu^S T \circ SS \mu^T  \circ S \lambda T
      \circ S \eta^T ST \circ S s
    & \text{by naturality of $\eta^T$} \\
    & = \mu^S T \circ S e \circ S s
    & \text{by definition of $e$ (\ref{eq:e})} \\
    & = \mu^S T \circ S s 
    & \text{since $e \circ s = s \circ r \circ s$.}
  \end{align*}

  Let $i \eqdef r \circ S \eta^T$ and $j \eqdef r \circ \eta^S T$.
  Our second claim is that $s \circ \mu^U \circ j U = S \mu^T \circ \lambda T \circ T s$.
  Indeed:
  \begin{align*}
    & s \circ \mu^U \circ j U \\
    & = s \circ r \circ B \circ s U \circ r U \circ
      \eta^S TU
    & \text{since $\mu^U = r \circ B \circ s U$ and by def.\ of $j$} \\
    & = e \circ B \circ e U \circ \eta^S TU
    & \text{since $e = s \circ r$} \\
    & = e \circ e \circ B \circ \eta^S TU
    & \text{by Lemma~\ref{lemma:ABC}, item~\ref{it:xi}} \\
    & = e \circ B \circ \eta^S TU
    & \text{since $e$ is idempotent} \\
    & = e \circ \mu^S \mu^T  \circ S \lambda T
      \circ ST s \circ \eta^S TU
    & \text{by definition of $B$} \\
    & = e \circ \mu^S \mu^T \circ
      \eta^S STT \circ \lambda T
      \circ T s
    & \text{by naturality of $\eta^S$} \\
    & = e \circ S \mu^T \circ \mu^S TT \circ
      \eta^S STT \circ \lambda T
      \circ T s
    & \text{by naturality of $\mu^S$} \\
    & = e \circ S \mu^T \circ \lambda T
      \circ T s
    & \text{since $S$ is a monad} \\
    & = S \mu^T \circ e \circ \lambda T
      \circ T s
    & \text{by Lemma~\ref{lemma:lambda:aux}, item~\ref{it:viii}} \\
    & = S \mu^T \circ \lambda T \circ T s
    & \text{by Lemma~\ref{lemma:lambda:aux}, item~\ref{it:vi}.}
  \end{align*}

  Our third claim is that $e \circ S \eta^T = \lambda \circ \eta^T
  S$.  Indeed:
  \begin{align*}
    e \circ S \eta^T
    & = S \mu^T \circ \lambda T \circ \eta^T ST \circ S
      \eta^T
    & \text{by definition of $e$} \\
    & = S \mu^T \circ \lambda T \circ TS \eta^T \circ
      \eta^T S
    & \text{by naturality of $\eta^T$} \\
    & = S \mu^T \circ ST \eta^T \circ \lambda \circ
      \eta^T S
    & \text{by naturality of $\lambda$} \\
    & = \lambda \circ \eta^T S
    & \text{since $\mu^T$ is a monad.}
  \end{align*}

  We can now verify the six laws of Definition~\ref{defn:STU:retr}.
  (\ref{eq:STU:eta}) is simply the definition of $\eta^U$, see
  (\ref{eq:etaU}).  (\ref{eq:STU:etaS:conv}) is
  Lemma~\ref{lemma:lambda:aux}, item~\ref{it:vii}.
  (\ref{eq:STU:SmuT:conv}) is Lemma~\ref{lemma:lambda:aux},
  item~\ref{it:viii}.  As far as (\ref{eq:STU:muS}) is concerned,
  \begin{align*}
    \mu^U \circ ir
    & = \mu^U \circ i U \circ S r
    & \text{by naturality of $i$} \\
    & = r \circ B \circ s U \circ r U \circ S \eta^T U
      \circ S r
    & \text{since $\mu^U = r \circ B \circ s U$ and by definition of $i$} \\
    & = r \circ B \circ e U \circ S \eta^T U
      \circ S r
    & \text{by definition of $e$} \\
    & = r \circ e \circ B \circ S \eta^T U
      \circ S r
    & \text{by Lemma~\ref{lemma:ABC}, item~\ref{it:xi}} \\
    & = r \circ B \circ S \eta^T U
      \circ S r
    & \text{since $r \circ e = r \circ s \circ r = r$} \\
    & = r \circ \mu^S T \circ S s \circ S r
    & \text{by our first claim above} \\
    & = r \circ \mu^S T \circ S e
    & \text{by definition of $e$} \\
    & = r \circ e \circ \mu^S T
    & \text{by Lemma~\ref{lemma:lambda:aux}, item~\ref{it:xii}} \\
    & = r \circ \mu^S T
    & \text{since $r \circ e = r$.}
  \end{align*}
  For (\ref{eq:STU:muT}), we compute:
  \begin{align*}
    \mu^U \circ rj
    & = r \circ A \circ ss \circ rj
    & \text{since $\mu^U = r \circ A \circ ss$} \\
    & = r \circ A \circ e (s \circ j)
    & \text{by the interchange law, and $s \circ r = e$} \\
    & = r \circ A \circ e (e \circ \eta^S T)
    & \text{since $s \circ j = s \circ r \circ \eta^S T= e \circ
      \eta^S T$} \\
    & = r \circ A \circ e \, \eta^S T
    & \text{by Lemma~\ref{lemma:lambda:aux}, item~\ref{it:vii}} \\
    & = r \circ A \circ ST \eta^S T \circ e T
    & \text{by naturality of $e$} \\
    & = r \circ S \mu^T \circ e T
    & \text{by Lemma~\ref{lemma:ABC}, item~\ref{it:xvi}} \\
    & = r \circ e \circ S \mu^T
    & \text{by Lemma~\ref{lemma:lambda:aux}, item~\ref{it:viii}} \\
    & = r \circ S \mu^T
    & \text{since $r \circ e = r \circ s \circ r = r$.}
  \end{align*}

  Finally, we deal with (\ref{eq:STU:muUS:conv}).   We have:
  \begin{align*}
    (s \circ \mu^U \circ j U)^{\ext S}
    & = \mu^S T \circ SS \mu^T \circ S \lambda T \circ ST
      s
    & \text{by our second claim above and def.\ of $\_^{\ext S}$} \\
    & = S \mu^T \circ \mu^S TT \circ S \lambda T
      \circ ST s
    & \text{by naturality of $\mu^S$,} \\
  \end{align*}
  and we use this on the first and last lines of the following
  derivation:
  \begin{align*}
    e \circ (s \circ \mu^U \circ j U)^{\ext S}
    & = e \circ S \mu^T \circ \mu^S TT \circ S \lambda T
      \circ ST s
    & \\
    & = S \mu^T \circ e T \circ \mu^S TT \circ S \lambda T
      \circ ST s
    & \text{by Lemma~\ref{lemma:lambda:aux}, item~\ref{it:viii}} \\
    & = S \mu^T \circ \mu^S TT \circ S \lambda T
      \circ e ST \circ ST s
    & \text{by Lemma~\ref{lemma:lambda:aux}, item~\ref{it:ix}} \\
    & = S \mu^T \circ \mu^S TT \circ S \lambda T
      \circ ST s \circ e U
    & \text{by naturality of $e$} \\
    & = (s \circ \mu^U \circ j U)^{\ext S} \circ e U,
  \end{align*}

  \ref{it:STU:3}. 
  \begin{align*}
    s \circ \mu^U \circ ji
    & = s \circ \mu^U \circ j U \circ T i
    & \text{by naturality of $j$} \\
    & = S \mu^T \circ \lambda T
      \circ T s \circ T r \circ TS \eta^T
    & \text{by the second claim above and the def.\ of $i$} \\
    & = S \mu^T \circ \lambda T
      \circ T e \circ TS \eta^T
    & \text{since $e = s \circ r$} \\
    & = S \mu^T \circ \lambda T \circ T\lambda \circ T \eta^T S
    & \text{by the third claim above} \\
    & = \lambda \circ \mu^T S  \circ T \eta^T S
    & \text{by (\ref{eq:distr:mu:2})} \\
    & = \lambda
    & \text{since $T$ is a monad.}
  \end{align*}
\end{proof}

Let us summarize.
\begin{theorem}
  \label{thm:STU:wdistr}
  Let $\catc$ be a 0-cell in a 2-category $\mathcal K$.  There is a
  one-to-one correspondence between
  \begin{enumerate}
  \item tuples $(S, T, U, r, s)$ where $S$, $T$, $U$ are monads on
    $\catc$ and $\xymatrix{ST \ar@<1ex>[r]^r & U \ar@<1ex>[l]^s}$ is a
    distributing retraction;
  \item tuples $(S, T, U, \lambda, r, s)$ where $S$ and $T$ are monads
    on $\catc$, $U$ is a 1-cell on $\catc$, $\lambda$ is a weak
    distributive law of $S$ over $T$, and
    $\xymatrix{ST \ar@<1ex>[r]^r & U \ar@<1ex>[l]^s}$ is a splitting
    of the idempotent
    $e \eqdef S \mu^T \circ \lambda T \circ \eta^T ST$.
  \end{enumerate}
  From 1 to 2, $\lambda \eqdef s \circ \mu^U \circ ji$, where
  $i \eqdef r \circ S \eta^T$ and $j \eqdef r \circ \eta^S T$.  From 2
  to 1, $\eta^U \eqdef r \circ \eta^S \eta^T$ and
  $\mu^U \eqdef r \circ \mu^S \mu^T \circ S \lambda T \circ ss$.
\end{theorem}
\begin{proof}
  This is Proposition~\ref{prop:STU:wdistr} from 1 to 2 and
  Proposition~\ref{prop:wdistr:STU}, items~\ref{it:STU:1}
  and~\ref{it:STU:2} from 2 to 1.  The two constructions are inverse
  of each other by Proposition~\ref{prop:wdistr:STU},
  item~\ref{it:STU:3} and by Proposition~\ref{prop:STU:unique}.
\end{proof}

\begin{longred}
  \subsection{Equivalent axioms}
  \label{sec:equivalent-axioms}

  In order to verify that a given retraction
  $\xymatrix{ST \ar@<1ex>[r]^r & U \ar@<1ex>[l]^s}$ is distributing, it
  is sometimes interesting to verify different, but equivalent
  axioms.  The following look like more complicated reformulations
  of (\ref{eq:STU:SmuT:conv}) and (\ref{eq:STU:muUS:conv}), but they
  will be easier to check.

  \begin{lemma}
    \label{lemma:STU:conv}
    Let $(S, \eta^S, \mu^S)$, $(T, \eta^T, \mu^T)$ and
    $(U, \eta^U, \mu^U)$ be three monads on a category $\catc$ and
    $\xymatrix{ST \ar@<1ex>[r]^r & U \ar@<1ex>[l]^s}$ be a natural
    retraction.  Let $e \eqdef s \circ r$.  In the presence of
    (\ref{eq:STU:muT}), (\ref{eq:STU:SmuT:conv}) is equivalent to:
    \begin{align}
      \label{eq:STU:SmuT:conv'}
      e \circ S \mu^T \circ e T
      & = S \mu^T \circ e T
    \end{align}
    In the presence of (\ref{eq:STU:muS}), (\ref{eq:STU:muUS:conv}) is
    equivalent to:
    \begin{align}
      \label{eq:STU:muUS:conv'}
      e \circ (s \circ \mu^U \circ j U)^{\ext S} \circ e U
      & = (s \circ \mu^U \circ j U)^{\ext S} \circ e U.
    \end{align}
  \end{lemma}
  \begin{proof}
    As usual, we let $i \eqdef r \circ S \eta^T \colon S \to U$ and
    $j \eqdef r \circ \eta^S T \colon T \to U$.  In the first case, we
    realize that:
    \begin{align*}
      r \circ S \mu^T \circ e T
      & = \mu^U \circ rj \circ e T
      & \text{by (\ref{eq:STU:muT})} \\
      & = \mu^U \circ (r \circ e)  j
      & \text{by the interchange law} \\
      & = \mu^U \circ rj
      & \text{since $r \circ e = r \circ s \circ r = r$} \\
      & = r \circ S \mu^T
      & \text{by (\ref{eq:STU:muT}).}
    \end{align*}
    Composing with $s$ on the left, we obtain that $e \circ S \mu^T
    \circ e T = e \circ S
    \mu^T$, from which the equivalence of (\ref{eq:STU:SmuT:conv}) and
    of (\ref{eq:STU:SmuT:conv'}) follows.

    In the second case, let
    $A \eqdef r \circ (s \circ \mu^U \circ j U)^{\ext S}$.  Then:
    \begin{align*}
      A
      & = r \circ \mu^S T \circ S s \circ S \mu^U \circ S j U \\
      & = \mu^U \circ ir \circ S s \circ S \mu^U \circ S j U
      & \text{by (\ref{eq:STU:muS})} \\
      & = \mu^U \circ i U \circ S \mu^U \circ S j U
      & \text{by the interchange law and since $r \circ s =
        U$} \\
      & = \mu^U \circ U \mu^U \circ ijU
      & \text{by naturality of $i$ and the interchange law} \\
      & = \mu^U \circ \mu^U U \circ ij U
      & \text{since $U$ is a monad} \\
      & = \mu^U \circ \mu^U U \circ irU \circ S \eta^S TU
      & \text{by definition of $i$ and the interchange law} \\
      & = \mu^U \circ r U \circ \mu^S TU \circ S \eta^S TU
      & \text{by (\ref{eq:STU:muT})} \\
      & = \mu^U \circ r U
      & \text{since $U$ is a monad.}
    \end{align*}
    Since $r$ is a retraction,
    $A \circ e U = A \circ s U \circ r U = \mu^U \circ r U = A$.
    Composing with $s$ on the left yields that
    $e \circ (s \circ \mu^U \circ j U)^{\ext S} \circ e U = e \circ (s
    \circ \mu^U \circ j U)^{\ext S}$; the equivalence between
    (\ref{eq:STU:muUS:conv}) and (\ref{eq:STU:muUS:conv'}) follows.
  \end{proof}
\end{longred}

\section{Smyth hyperspace, continuous valuations and superlinear
  previsions}
\label{sec:smyth-hypersp-cont}

For every topological space $X$, let $\Smyth X$ be the set of
non-empty compact saturated subsets of $X$.  (A subset is
\emph{saturated} if and only if it is the intersection of its open
neighborhoods, or equivalently if it is upwards closed in the
specialization preordering $\leq$ of $X$, defined as $x \leq y$ if and
only if every open neighborhood of $x$ contains $y$.  A space is $T_0$
if and only if its specialization preordering is antisymmetric.  We
refer to \citep{JGL-topology} for all questions topological.)  The
\emph{upper Vietoris} topology on that set has basic open subsets
$\Box U$ consisting of those non-empty compact saturated subsets of
$X$ that are included in $U$, where $U$ ranges over the open subsets
of $X$.  We write $\SV X$ for the resulting topological space.  Its
specialization ordering is reverse inclusion $\supseteq$.

The $\SV$ construction was studied by a number of people, starting
with Mike Smyth \citep{Smyth:power:pred}, and later by Andrea Schalk
\citep[Section~7]{Schalk:diss} who studied not only this, but also the
variant with the Scott topology, and a localic counterpart.  See also
\citep[Sections~6.2.2, 6.2.3]{AJ:domains} or
\citep[Section~IV-8]{GHKLMS:contlatt}, where the accent is rather on
the Scott topology of $\supseteq$.  In a computer-science setting,
$\SV$ is a model of \emph{demonic} non-determinism.

There is a monad $(\SV, \eta^\Smyth, \mu^\Smyth)$ on $\Topcat$, and
this will be our first monad ($S$).  For every continuous map
$f \colon X \to Y$, $\SV f$ maps every $Q \in \SV X$ to $\upc f [Q]$,
where $f [Q]$ denotes $\{f (x) \mid x \in Q\}$, and $\upc$ denotes
upward closure.  The unit $\eta^\Smyth$ is defined by
$\eta^\Smyth_X (x) \eqdef \upc x$ for every $x \in X$.  For every
continuous map $f \colon X \to \SV Y$, there is an extension
$f^\sharp \colon \SV X \to \SV Y$, defined by
$f^\sharp (Q) \eqdef \bigcup_{x \in Q} f (x)$.  The multiplication
$\mu^\Smyth$ is defined by
$\mu^\Smyth_X \eqdef \identity {\SV X}^\sharp$, namely
$\mu^\Smyth_X (\mathcal Q) \eqdef \bigcup \mathcal Q$
\citep[Proposition~7.21]{Schalk:diss}.  

Before we go to our second monad of interest, we need the following.
A \emph{dcpo} (short for \emph{directed complete partial order}) is a
poset $X$ in which every directed family $D$ has a supremum $\sup D$;
a family $D$ is \emph{directed} if and only if $D \neq \emptyset$ and
any two elements of $D$ have an upper bound in $D$.  A function
$f \colon X \to Y$ between dcpos is \emph{Scott-continuous} if and
only if it is monotonic and preserves suprema of directed families.  A
subset $U$ of $X$ is \emph{Scott-open} if and only if it is
upwards-closed and intersects every directed family $D$ such that
$\sup D \in U$.  The Scott-open sets form the \emph{Scott topology} on
$X$.  See \citep{GHKLMS:contlatt,AJ:domains,JGL-topology} for notions
of dcpos, Scott-continuous and other notions of domain theory.

Our second monad, $T$, is that of continuous valuations, a close
cousin of measures.  Let $\Open X$ denote the lattice of open subsets
of a space $X$.  This is in particular a dcpo; so is
$\creal \eqdef \Rp \cup \{\infty\}$, with its usual ordering, which
places $\infty$ above all non-negative real numbers.  A
\emph{continuous valuation} on a space $X$ is a map
$\nu \colon \Open X \to \creal$ that is \emph{strict}
($\nu (\emptyset)=0$), 
\emph{modular} (for all $U, V \in \Open X$,
$\nu (U) + \nu (V) = \nu (U \cup V) + \nu (U \cap V)$) and
Scott-continuous.  The latter means that $U \subseteq V$ implies
$\nu (U) \leq \nu (V)$ and that for every directed family
${(U_i)}_{i \in I}$ of open subsets of $X$,
$\nu (\bigcup_{i \in I} U_i) = \sup_{i \in I} \nu (U_i)$.
We say that $\nu$ is a \emph{probability} valuation if and only if
$\nu (X)=1$, and a \emph{subprobability} valuation if and only if
$\nu (X) \leq 1$.

Let $\Val X$ denote the space of continuous valuations on a space $X$,
with the following \emph{weak topology}.  That is defined by subbasic
open sets $[U > r] \eqdef \{\nu \in \Val X \mid \nu (U) > r\}$, where
$U \in \Open X$ and $r \in \Rp$.  We define its subspace $\Val_1 X$ of
probability valuations and $\Val_{\leq 1} X$ (subprobability)
similarly.  In general, we will write $\Val_\bullet X$, where
$\bullet$ is nothing, ``$\leq 1$'', or ``$1$''.

Continuous valuations are an alternative to measures that have become
popular in domain theory \citep{jones89,Jones:proba}.  The first
results that connected continuous valuations and measures are due to
\citet{saheb-djahromi:meas} and 
\citet{Lawson:valuation}.  More recent results include the following.
In one direction, every measure on the Borel $\sigma$-algebra of $X$
induces a continuous valuation on $X$ by restriction to the open sets,
if $X$ is hereditarily Lindel\"of (namely, if every directed family of
open sets contains a cofinal monotone sequence).  This is an easy
observation, and one half of Adamski's theorem
\citep[Theorem~3.1]{Adamski:measures}, which states that a space is
hereditary Lindel\"of if and only if every measure on its Borel
$\sigma$-algebra restricts to a continuous valuation on its open sets.
In the other direction, every continuous valuation on a space $X$
extends to a measure on the Borel sets provided that $X$ is an
LCS-complete space \citep[Theorem~1]{dBGLJL:LCScomplete}.  An
\emph{LCS-complete} space is a space homeomorphic to a $G_\delta$
subset of a locally compact sober space---$G_\delta$ meaning a
countable intersection of open sets.  Every continuous dcpo in its
Scott topology, being locally compact and sober, is LCS-complete.
Every Polish space, and more generally every quasi-Polish space, is
LCS-complete; see \citep{dBGLJL:LCScomplete}.

There is a monad $(\Val_\bullet, \eta, \mu)$ on $\Topcat$.  For every
continuous map $f \colon X \to Y$, $\Val_\bullet f$ sends every
$\nu \in \Val_\bullet X$ to the \emph{pushforward} continuous
valuation $f [\nu]$, defined by $f [\nu] (V) \eqdef \nu (f^{-1} (V))$
for every $V \in \Open Y$.  The unit $\eta$ is defined by
$\eta_X (x) \eqdef \delta_x$ for every $x \in X$, for every space $X$;
$\delta_x$ is the \emph{Dirac valuation} at $x$, and maps every open
subset $U$ of $X$ to $1$ if $x \in U$, to $0$ otherwise.  For every
continuous map $f \colon X \to \Val_\bullet Y$, its extension
$f^\dagger \colon \Val_\bullet X \to \Val_\bullet Y$ is defined by
$f^\dagger (\nu) (V) \eqdef \int_{x \in X} f (x) (V) \,d\nu$ for every
$\nu \in \Val_\bullet X$ and for every $V \in \Open Y$, and then
$\mu_X \eqdef \identity {\Val_\bullet X}^\dagger$.  Integration is
best defined by a Choquet formula \citep[Section~4]{Tix:bewertung}.
Explicitly, let $\Lform X$ be the space of all lower semicontinuous
functions $h$ from $X$ to $\creal$, that is, the set of all continuous
maps from $X$ to $\creal$, equipped with its Scott topology of the
pointwise ordering.  For every $h \in \Lform X$,
$\int_{x \in X} h (x) \,d\nu$ is defined as the indefinite Riemann
integral $\int_0^\infty \nu (h^{-1} (]t, \infty])) \,dt$.  The
following \emph{change of variable} formula follows immediately: for
every continuous map $f \colon X \to Y$, for every $h \in \Lform Y$,
for every $\nu \in \Val_\bullet X$,
$\int_{y \in Y} h (y) \,df[\nu] = \int_{x \in X} h (f (x)) \,d\nu$.
We also note that, for every $h \in \Lform X$, for every
$\xi \in \Val_\bullet {\Val_\bullet X}$,
\begin{align}
  \label{eq:chgvar}
  \int_{x \in X} h (x) \,d\mu_X (\xi)
  & = \int_{\nu \in \Val_\bullet X}
    \left(\int_{x \in X} h (x) \,d\nu\right)\,d\xi
\end{align}
See \citep[Hilfssatz~6.1]{kirch93} for a proof.

Integration allows one to give another presentation of the weak
topology on $\Val_\bullet X$: this is also the topology generated by
subbasic open sets
$[h > r] \eqdef \{\nu \in \Val_\bullet X \mid \int_{x \in X} h (x)
\,d\nu > r\}$, where $h \in \Lform X$ and $r \in \Rp$.  This was
proved in \citep[Theorem~3.3]{Jung:scs:prob} when $\bullet$ is ``$1$''
or ``$\leq 1$'', but the same argument also shows it when $\bullet$ is
nothing, see \citep[Satz~8.5]{kirch93} or \citep[Lemma~4.9]{kirch93}.

Our third monad, $U$, is the monad of superlinear previsions of
\citep{Gou-csl07}, except that this paper studied it on the category of
dcpos and Scott-continuous maps instead of $\Topcat$; the study of
superlinear previsions on $\Topcat$ was initiated in
\citep{JGL-mscs16}, but it was not stated that they formed a monad
there, although, as we will see below, this is pretty elementary.

A \emph{prevision} on $X$ is a Scott-continuous functional
$F \colon \Lform X \to \creal$ such that $F (ah)=aF(h)$ for all
$a \in \Rp$ and $h \in \Lform X$.  It is \emph{linear} (resp.,
\emph{superlinear}, \emph{sublinear}) if and only if
$F (h+h')=F (h)+F (h')$ (resp., $\geq$, $\leq$) for all
$h, h' \in \Lform X$.  It is \emph{subnormalized} (resp.,
\emph{normalized}) iff $F (\one+h) \leq 1+F (h)$ (resp., $=$) for
every $h \in \Lform X$, where $\one$ is the constant function with
value $1$.  We write $\Pred_\DN X$ for the space of superlinear
previsions on $X$, $\Pred_\DN^{\leq 1} X$ for its subspace of
subnormalized superlinear previsions, and $\Pred_\DN^1 X$ for its
subspace of normalized superlinear previsions; synthetically, we use
the notation $\Pred_\DN^\bullet X$.  Similarly, we denote by
$\Pred_\Nature^\bullet X$ the corresponding spaces of linear
previsions and by $\Pred_\AN^\bullet X$ the corresponding spaces of
sublinear previsions.  The topology on each of those spaces is
generated by subbasic open sets
$[h > r] \eqdef \{F \in \Pred_\DN^\bullet X \text{ (resp.,
  $\Pred_\Nature^\bullet X$, $\Pred_\AN^\bullet X$)} \mid F (h) >
r\}$, $h \in \Lform X$, $r \in
\Rp$. 

There is a $\Pred_\DN^\bullet$ endofunctor on $\Topcat$, and its
action on morphisms $f \colon X \to Y$ is given by
$\Pred_\DN^\bullet f (F) = (h \in \Lform Y \mapsto F (h \circ f))$.
The $\Pred_\Nature^\bullet$ and $\Pred_\AN^\bullet$ endofunctors are
defined similarly.  Those are all functor parts of monads.  Those
monads were made explicit in \citep[Proposition~2]{Gou-csl07}, on the
category of dcpos (resp., pointed dcpos) and Scott-continuous maps,
except that spaces of previsions were given the Scott topology of the
pointwise ordering.  The formulae for unit and multiplication is the
same below.
In order to verify the monad equations, we will rely on Ernest Manes'
equivalent characterization through units and extension operations
\citep{Manes:algth}: a monad $S$ is equivalently given by objects $SX$
and morphisms $\eta^S_X \colon X \to SX$, one for each object $X$,
plus an extension $f^{\ext S} \colon SX \to SY$ for every morphism
$f \colon X \to SY$, satisfying:
\begin{enumerate}[label=(\roman*),leftmargin=*]
\item\label{it:Manes:eta} ${(\eta^S_X)}^{\ext S} = \identity {SX}$;
\item\label{it:Manes:dagger:1} for every morphism
  $f \colon X \to SY$, $f^{\ext S} \circ \eta^S_X = f$;
\item\label{it:Manes:dagger:2} for all morphisms $f \colon X \to SY$
  and $g \colon Y \to SZ$,
  $g^{\ext S} \circ f^{\ext S} = {(g^{\ext S} \circ f)}^{\ext S}$.
\end{enumerate}
Then $S f$ can be recovered, for every morphism $f \colon X \to Y$, as
$(\eta^S_Y \circ f)^{\ext S}$, and $\mu^S_X$ as
$\identity {SX}^{\ext S}$.

\begin{proposition}
  \label{prop:DN:monad}
  Let $\bullet$ be nothing, ``$\leq 1$'', or ``$1$''.  There is a
  monad
  $(\Pred_\DN^\bullet, \allowbreak \eta^\DN, \allowbreak \mu^\DN)$
  (resp.,
  $(\Pred_\Nature^\bullet, \allowbreak \eta^\Nature, \allowbreak
  \mu^\Nature)$,
  $(\Pred_\AN^\bullet, \allowbreak \eta^\AN, \allowbreak \mu^\AN)$) on
  $\Topcat$, whose unit and multiplication are defined at every
  topological space $X$ by:
  \begin{itemize}
  \item for every $x \in X$, $\eta^\DN_X (x)$ (resp.,
    $\eta^\Nature_X$, $\eta^\AN_X$)
    $\eqdef (h \in \Lform X \mapsto h (x))$,
  \item for every
    $\mathcal F \in \Pred_\DN^\bullet {\Pred_\DN^\bullet X}$ (resp.,
    in $\Pred_\Nature^\bullet {\Pred_\Nature^\bullet X}$),
    $\mu^\DN_X (\mathcal F)$ (resp., $\mu^\Nature_X$, $\mu^\AN_X$)
    $\eqdef (h \in \Lform X \mapsto \mathcal F (F \mapsto F (h)))$,
    where $F$ ranges over $\Pred_\DN^\bullet X$, resp.\
    $\Pred_\Nature^\bullet X$, $\Pred_\AN^\bullet X$.
  \end{itemize}
  For every continuous map $f \colon X \to \Pred_\DN^\bullet Y$
  (resp., $\Pred_\Nature^\bullet Y$, $\Pred_\AN^\bullet Y$), its
  extension $f^{\ext \DN}$ is defined by
  $f^{\ext \DN} (F) \eqdef (h \in \Lform Y \mapsto F (x \in X \mapsto f (x)
  (h)))$ (resp., and the same formula for $f^{\ext \Nature}$, $f^{\ext \AN}$).
\end{proposition}
\begin{proof}
  Unit.  The functional $(h \in \Lform X \mapsto h (x))$ is a linear
  prevision, hence also a superlinear and a sublinear prevision, as
  one checks easily.  It is also normalized, hence subnormalized.  The
  map $\eta^\DN_X$ is continuous, since the inverse image of $[h > r]$
  is $h^{-1} (]r, \infty])$.

  Extension.  For every $h \in \Lform X$, the map
  $x \in X \mapsto f (x) (h)$ is lower semicontinuous, since the
  inverse image of $]r, \infty]$ is $f^{-1} ([h > r])$.  Hence
  $F (x \in X \mapsto f (x) (h))$ makes sense.  The map
  $(h \in \Lform Y \mapsto F (x \in X \mapsto f (x) (h)))$ itself is
  Scott-continuous, because $F$ and $f (x)$ are.  Hence
  $f^{\ext \DN} (F)$ is well-defined.  We have
  $(f^{\ext \DN})^{-1} ([h > r]) = [(x \in X \mapsto f (x) (h)) > r]$, so
  $f^{\ext \DN}$ is continuous.  Additionally, when $F$ and every $f (x)$
  is superlinear (resp., linear, sublinear), then
  $f^{\ext \DN} (F) (h+h') \geq F (x \in X \mapsto f (x) (h) + f (x) (h'))
  \geq f^{\ext \DN} (F) (h) + f^{\ext \DN} (F) (h')$ (resp. $=$, $\leq$); we
  have implicitly used the fact that $F$, being continuous, is
  monotonic.  When $\bullet$ is ``$\leq 1$'', namely if $F$ and every
  $f (x)$is subnormalized, then
  $f^{\ext \DN} (F) (\one + h) \leq F (x \in X \mapsto 1+ f (x) (h)) \leq 1
  + f^{\ext \DN} (F) (h)$, so $f^{\ext \DN} (F)$ is subnormalized; similarly, if
  $\bullet$ is ``$1$'', then $f^{\ext \DN} (F)$ is normalized.

  The formulation for multiplication follows immediately.  We also
  check that the action of the functor on morphisms $f$ gives back the
  expected formula $F \mapsto (h \in \Lform Y \mapsto F (h \circ
  f))$.  Finally, we check Manes' equations.  We only deal with $\DN$,
  as the $\AN$ and $\Nature$ cases are identical.
  \begin{itemize}
  \item[\ref{it:Manes:eta}] $(\eta^\DN_X)^{\ext \DN}$ maps $F$ to
    $(h \in \Lform X \mapsto F (x \in X \mapsto \eta^\DN_X (x) (h))) =
    F$.
  \item[\ref{it:Manes:dagger:1}] for every continuous map
    $f \colon X \to \Pred_\DN^\bullet Y$, for every $x \in X$,
    $(f^{\ext \DN} \circ \eta^\DN_X) (x) = (h \in \Lform Y \mapsto
    \eta^\DN_X (x) (x \in X \mapsto f (x) (h))) = f (x)$.
  \item[\ref{it:Manes:dagger:2}] for all continuous maps
    $f \colon X \to \Pred_\DN^\bullet Y$ and
    $g \colon Y \to \Pred_\DN^\bullet Z$, for every
    $F \in \Pred_\DN^\bullet X$, for every $h \in \Lform Z$,
    $(g^{\ext \DN} \circ f^{\ext \DN}) (F) (h) = f^{\ext \DN} (F) (y \in Y \mapsto g
    (y) (h)) = F (x \in X \mapsto f (x) (y \in Y \mapsto g (y) (h)))$,
    while
    $(g^{\ext \DN} \circ f)^{\ext \DN} (F) (h) = F (x \in X \mapsto (g^{\ext \DN}
    \circ f) (x) (h)) = F (x \in X \mapsto f (x) (y \in Y \mapsto g
    (y) (h)))$.
  \end{itemize}
%

\end{proof}

There is a homeomorphism between $\Pred_\Nature^\bullet X$ and
$\Val_\bullet X$, for every space $X$: in one direction,
$\nu \in \Val_\bullet X$ is mapped to its \emph{integration
  functional} $(h \in \Lform X \mapsto \int_{x \in X} h (x) \,d\nu)$,
and in the other direction any $G \in \Pred_\Nature^\bullet X$ is
mapped to $(U \in \Open X \mapsto G (\chi_U))$, where $\chi_U$ is the
characteristic map of $U$.  We note in particular that
$\int_{x \in X} \chi_U (x) \,d\nu = \nu (U)$.  This was first proved
by 
\citet[Satz~8.1]{kirch93}, under the additional assumption
that $X$ is core-compact; Tix later observed that this assumption was
unnecessary \citep[Satz~4.16]{Tix:bewertung}.

\begin{lemma}
  \label{lemma:V=PN}
  Let $f_X \colon \Val_\bullet X \to \Pred_\Nature^\bullet X$ map
  $\nu$ to $(h \in \Lform X \mapsto \int_{x \in X} h (x) \,d\nu)$.
  Then $f$ is a monad isomorphism from $\Val_\bullet$ onto
  $\Pred_\Nature^\bullet$.
\end{lemma}
\begin{proof}
  For every space $X$, $f_X$ is a bijection, and the inverse image of
  $[h > r]$ (seen as a subbasic open subset of
  $\Pred_\Nature^\bullet X$) is
  $[h > r] = \{\nu \in \Val_\bullet X \mid \int_{x \in X} h (x) \,d\nu
  > r\}$, showing that it is a homeomorphism.  It remains to see that
  $f$ is a monad homomorphism, namely that
  $f \circ \eta = \eta^\Nature$ and
  $f \circ \mu = \mu^\Nature \circ ff$.  For the first one,
  let $x \in X$, then
  $(f_X \circ \eta_X) (x) = f_X (\delta_x) = (h \in \Lform X \mapsto
  \int_{x' \in X} h (x') \,d\delta_x) = (h \in \Lform X \mapsto h (x))
  = \eta^\Nature_X$.  For the second one, for every $\xi \in
  \Val_\bullet \Val_\bullet X$, for every $h \in \Lform X$,
  \begin{align*}
    (f_X \circ \mu_X) (\xi) (h)
    & = \int_{x \in X} h (x) \,d\mu_X (\xi)
      = \int_{\nu \in \Val_\bullet X} \left(\int_{x \in X} h (x) \,d\nu\right) \,d\xi,
  \end{align*}
  while:
  \begin{align*}
    (\mu^\Nature_X \circ (ff)_X) (\xi) (h)
    &= (ff)_X (\xi) (F \mapsto F (h)) \\
    & = \Pred_\Nature^\bullet f_X (f_{\Val_\bullet X} (\xi)) (F
      \mapsto F (h)) \\
    & = f_{\Val_\bullet X} (\xi) ((F \mapsto F (h)) \circ f_X) \\
    & = f_{\Val_\bullet X} (\xi) (\nu \mapsto f_X (\nu) (h))
    = \int_{\nu \in \Val_\bullet X} f_X (\nu) (h) \,d\xi,
  \end{align*}
  which is the same thing.
\end{proof}

We will also need the following later on.
\begin{lemma}
  \label{lemma:muT:lin}
  Let $\bullet$ be nothing, ``$\leq 1$'' or ``$1$'', and $X$ be a
  topological space.  Then $\mu^\Nature_X$ is \emph{affine}: for all
  $\mathcal G_1, \mathcal G_2 \in \Pred_\Nature^\bullet
  \Pred_\Nature^\bullet X$, for every $a \in [0, 1]$,
  $\mu^\Nature_X (a \mathcal G_1 + (1-a) \mathcal G_2) = a
  \mu^\Nature_X (\mathcal G_1) + (1-a) \mu^\Nature_X (\mathcal G_2)$.
\end{lemma}
\begin{proof}
  Letting $G$ range implicitly over $\Pred_\Nature^\bullet X$ and $h$
  over $\Lform X$, we have
  $\mu^\Nature_X (a \mathcal G_1 + (1-a) \mathcal G_2) = (h \mapsto
  (\alpha \mathcal G_1 + (1-a) \mathcal G_2) (G \mapsto G (h))) = (h
  \mapsto \alpha \mathcal G_1 (G \mapsto G (h)) + (1-a) \mathcal G_2
  (G \mapsto G (h))) = a \mu^\Nature_X (\mathcal G_1) + (1-a)
  \mu^\Nature_X (\mathcal G_2)$.
\end{proof}

\subsection{The distributing retraction $r_\DN$, $s_\DN^\bullet$}
\label{sec:r_dn-s_dnbullet-form}

We need to use a few results from \citep{JGL-mscs16}.  For every space
$X$, there is a map
$r_\DN \colon \SV \Pred_\Nature^\bullet X \to \Pred_{\DN}^\bullet X$,
defined by $r_\DN (Q) (h) \eqdef \min_{G \in Q} G (h)$, and a map
$s_\DN^\bullet \colon \Pred_{\DN}^\bullet X \to \SV
\Pred_\Nature^\bullet X$ defined by
$s_\DN^\bullet (F) \eqdef \{G \in \Pred_\Nature^\bullet X \mid G \geq
F\}$; $G \geq F$ abbreviates for every $h \in \Lform X$,
$G (h) \geq F (h)$.  Both are continuous, and they form a retraction:
$r_\DN \circ s_\DN^\bullet = \identity {\Pred_{\DN}^\bullet X}$.  This
is the contents of Proposition~3.22 of \citep{JGL-mscs16}.

Additionally, $r_\DN$ and $s_\DN^\bullet$ restrict to a homeomorphism
between $\Pred_\DN^\bullet X$ and the subspace
$\SV^{cvx} \Pred_\Nature^\bullet X$ of $\SV \Pred_\Nature^\bullet X$
consisting of its convex elements \citep[Theorem~ 4.15]{JGL-mscs16}.
In particular, $s_\DN^\bullet$ only takes convex values, and in order
to show that two elements $\mathcal Q$, $\mathcal Q'$ of
$\SV \Pred_\Nature^\bullet X$ are equal, it suffices to show that they
are convex and have the same image under $r_\DN$.

One result that is missing in \citep{JGL-mscs16} is the following.
\begin{lemma}
  \label{lemma:rDP:nat}
  The transformations $r_\DN$ and $s_\DN^\bullet$ are natural.
\end{lemma}
\begin{proof}
  Let $f \colon X \to Y$ be any continuous map.  For $r_\DN$, we show
  that for every $Q \in \SV \Pred_\Nature^\bullet X$, for every
  $h \in \Lform Y$,
  $r_{\DN} (\SV (\Pred_\DN^\bullet f) (Q)) (h) = \Pred_\DN^\bullet f
  (r_{\DN} (Q)) (h)$.  The left-hand side is equal to
  $\min_{G' \in \SV \Pred_\DN^\bullet (f) (Q)} G' (h) = \min_{G' \in
    \upc \{\Pred_\DN^\bullet (f) (G) \mid G \in Q\}} \allowbreak G'
  (h) = \min_{G' \in \{\Pred_\DN^\bullet (f) (G) \mid G \in Q\}}
  \allowbreak G' (h) = \min_{G \in Q} \Pred_\DN^\bullet f (G) (h)$,
  while the right-hand side is
  $r_{\DN} (Q) (h \circ f) = \min_{G \in Q} G (h \circ f)$, and those
  are equal.

  For $s_\DN^\bullet$, we must show that for every
  $F \in \Pred_\DN^\bullet X$,
  $s_{\DN}^\bullet (\Pred_\DN^\bullet (f) (F)) = \SV \Pred_\DN^\bullet
  (f) \allowbreak (s_{\DN}^\bullet (F))$.  We apply the trick
  described above: we show that both sides are convex, and have the
  same image under $r_\DN$.  The left-hand side is convex because
  $s_\DN^\bullet$ only takes convex values.  For the right-hand side,
  we consider any two elements $G'_1$, $G'_2$ of
  $\SV \Pred_\DN^\bullet (f) (s_{\DN}^\bullet (F))$, and
  $r \in [0, 1]$.  By definition, there are elements $G_1$, $G_2$ of
  $s_{\DN}^\bullet (F)$ such that
  $\Pred_\DN^\bullet f (G_1) \leq G'_1$ and
  $\Pred_\DN^\bullet f (G_2) \leq G'_2$.  Since $s_{\DN}^\bullet (F)$
  is convex, $rG_1+ (1-r)G_2 \in s_{\DN}^\bullet (F)$.  It is easy to
  see that
  $r G'_1 + (1-r) G'_2 \geq r \Pred_\DN^\bullet f (G_1) + (1-r)
  \Pred_\DN^\bullet f (G_2) = \Pred_\DN^\bullet f (rG_1 + (1-r) G_2)$,
  so
  $r G'_1 + (1-r) G'_2 \in \SV (\Pred_\DN^\bullet f) (s_{\DN}^\bullet
  (F))$.  It remains to see that
  $s_{\DN}^\bullet (\Pred_\DN^\bullet (f) (F)) = \SV \Pred_\DN^\bullet
  (f) \allowbreak (s_{\DN}^\bullet (F))$ have the same image under
  $r_\DN$.  The image of the left-hand side is
  $\Pred_\DN^\bullet (f) (F)$, while that of the right-hand side is
  $\Pred_\DN^\bullet f (r_{\DN} (s_{\DN}^\bullet (F))$ (by naturality
  of $r_\DN$), hence to $\Pred_\DN^\bullet f (F)$.
\end{proof}

\begin{theorem}
  \label{thm:DN:STU}
  $\xymatrix{\SV \Val_\bullet \ar@<1ex>[r]^{r_\DN} &
    \Pred_{\DN}^\bullet \ar@<1ex>[l]^{s_\DN^\bullet}}$ is a
  distributing retraction on $\Topcat$.
\end{theorem}
\begin{proof}
  Using Lemma~\ref{lemma:V=PN}, we replace $\Val_\bullet$ by the
  isomorphic monad $\Pred_\Nature^\bullet$.
  Naturality is by Lemma~\ref{lemma:rDP:nat}.

  Let $S \eqdef \SV$, $T \eqdef \Pred_\Nature^\bullet$,
  $U \edef \Pred_\DN^\bullet$, $r \eqdef r_\DN$,
  $s \eqdef s_\DN^\bullet$.  We first observe that for every space
  $Y$, $j_Y \eqdef r_Y \circ \eta^S_{TY}$ is defined by: for every
  $G \in TY$, $j_Y (G) = r (\eta^S_{TY} (G)) = r (\upc G) = G \in UY$;
  hence $j_Y$ is simply the inclusion of $TY$ into $UY$.

  We claim that the function $i_Y \eqdef r_Y \circ S \eta^T_Y$ maps
  every $Q \in \SV Y$ to
  $(h \in \Lform Y \mapsto \min_{y \in Q} h (y))$.  Indeed,
  $i_Y (Q) (h) = r_Y (\SV \eta_Y (Q)) (h) = \min_{y \in Q, G \geq
    \eta_Y (y)} G (h) = \min_{y \in Q} h (y)$.  Note that we equate
  $\eta_Y (y) = \delta_y$ with $h \in \Lform Y \mapsto h (y)$,
  something we will continue to do.

  It remains to check the six laws of Definition~\ref{defn:STU:retr}.

  (\ref{eq:STU:eta}) For every $x \in X$,
  $(r \circ \eta^S \eta^T)_X (x) = r (\upc \delta_x)$ is the prevision
  that maps $h \in \Lform X$ to
  $\min_{G \in \upc \delta_x} G (h) = \delta_x (h) = h (x)$, namely
  $\eta^\DN_X (x) = \eta^U_X (x)$.

  
  (\ref{eq:STU:muS}) For every $\mathcal Q \in SST X$, $(r_X \circ
  \mu^S_{TX}) (\mathcal Q)$ maps every $h \in \Lform X$ to
  $r_X (\bigcup \mathcal Q) (h) = \min_{G \in \bigcup \mathcal Q}
  G (h) = \min_{Q \in \mathcal Q} \min_{G \in Q} G (h)$, while:
  \begin{align*}
    (\mu^U_X \circ (ir)_X) (\mathcal Q) (h)
    & = \mu^\DN_X ((ir)_X (\mathcal Q)) (h) \\
    & = (ir)_X (\mathcal Q) (F \mapsto F (h)) \\
    & = \Pred_\DN^\bullet r (i_{\SV
      \Pred_\Nature^\bullet X} (\mathcal Q)) (F \mapsto F (h)) \\
    & = i_{\SV \Pred_\Nature^\bullet X} (\mathcal Q)
      ((F \mapsto F (h)) \circ r) \\
    & = i_{\SV \Pred_\Nature^\bullet X} (\mathcal Q)
      (F \mapsto F (h)) \circ r) \\
    & = \min_{Q \in \mathcal Q} ((F \mapsto F (h)) \circ r) (Q)
    = \min_{Q \in \mathcal Q} r (Q) (h)
    = \min_{Q \in \mathcal Q} \min_{G \in Q} G (h).
  \end{align*}
  Therefore $r \circ \mu^S T = \mu^U \circ ir$.

  (\ref{eq:STU:muT}) For every $\mathcal Q \in STT X$,
  $(r_X \circ S \mu^T_X) (\mathcal Q)$ maps every $h \in \Lform X$ to
  $r_X (\upc \mu^T_X [\mathcal Q]) (h) = \min_{G \in \upc \mu^T_X
    [\mathcal Q]} G (h) = \min_{\mathcal G \in \mathcal Q}
  \mu^\Nature_X (\mathcal G) (h) = \min_{\mathcal G \in \mathcal Q}
  \mathcal G (G \mapsto G (h))$, while, keeping the convention that
  $G$ ranges over $\Pred_\Nature^\bullet X$ and $F$ over
  $\Pred_\DN^\bullet X$:
  \begin{align*}
    (\mu^U_X \circ (rj)_X) (\mathcal Q) (h)
    & = \mu^\DN_X ((rj)_X (\mathcal Q)) (h) \\
    & = (rj)_X (\mathcal Q) (F \mapsto F (h)) \\
    & = \Pred_\DN^\bullet j (r_{\Pred_\Nature^\bullet X} (\mathcal Q))
      (F \mapsto F (h)) \\
    & = r_{\Pred_\Nature^\bullet X} (\mathcal Q) ((F \mapsto F (h))
      \circ j) \\
    & = r_{\Pred_\Nature^\bullet X} (\mathcal Q) (G \mapsto G (h))
    & \text{since $j$ is mere inclusion} \\
    & = \min_{\mathcal G \in \mathcal Q} \mathcal G (G \mapsto G (h)).
  \end{align*}
  Therefore $r \circ S \mu^T = \mu^U \circ rj$.

  For the remaining three laws, we recall that $r$ and $s$ define a
  homeomorphism between $\Pred_\DN^\bullet X$ and
  $\SV^{cvx} \Pred_\Nature^\bullet X$ \citep[Theorem~ 4.15]{JGL-mscs16}.
  Letting $e \eqdef s \circ r$, it follows that $e (Q) = Q$ if and
  only if $Q$ is convex.

  (\ref{eq:STU:etaS:conv}) states that $\eta^S_{TX}$ always returns a
  convex element of $ST X = \SV {\Pred^\bullet_\Nature X}$, and
  indeed, for every $G \in TX$, $\eta^S_{TX} = \upc G$ is convex.

  (\ref{eq:STU:SmuT:conv}).  \red{Using Lemma~\ref{lemma:STU:conv}, we
    verify (\ref{eq:STU:SmuT:conv'}) instead.  This states that
    $e_X (S \mu^T_X (e_{TX} (\mathcal Q))) = S\mu^T (e_{TX} (\mathcal
    Q))$, namely that $S\mu^T (e_{TX} (\mathcal Q))$ is convex, for
    every $\mathcal Q \in STTX$.  But $e_{TX} (\mathcal Q)$ is convex,
    and $\mu^T_X$ is affine by Lemma~\ref{lemma:muT:lin}, so the image
    $\mu^T_X [e_{TX} (\mathcal Q)]$ is convex, and therefore also its
    upward closure
    $\upc \mu^T_X [e_{TX} (\mathcal Q)] = S \mu^T_X (e_{TX} (\mathcal
    Q))$.  }


  (\ref{eq:STU:muUS:conv}).  \red{Using Lemma~\ref{lemma:STU:conv}, we
    verify (\ref{eq:STU:muUS:conv'}) instead.  Let
    $f \eqdef s_X \circ \mu^U_X \circ j_{UX}$.  We need to show that
    $e \circ f^{\ext S} \circ e U = f^{\ext S} \circ e U$.  As in the
    previous case, it suffices to show that $f^{\ext S}$ maps} convex
  sets to convex
  sets.  
  For every $\mathcal G \in TUX$,
  $f (\mathcal G) = s_X (\mu^U_X (\mathcal G)) = \{G \in TX \mid G
  \geq (h \in \Lform X \mapsto \mu^U_X (\mathcal G))\} = \{G \in TX
  \mid G \geq (h \in \Lform X \mapsto \mathcal G (F \in UX \mapsto F
  (h)))\}$.  It follows that for every $\mathcal Q \in STUX$,
  $f^{\ext S} (\mathcal Q) = \{G \in TX \mid \exists \mathcal G \in
  \mathcal Q, G \geq (h \in \Lform X \mapsto \mathcal G (F \in UX
  \mapsto F (h)))\}$.  Now, if $\mathcal Q$ is convex, then for all
  $G_1, G_2 \in f^{\ext S} (\mathcal Q)$, for every $a \in [0, 1]$,
  $a G_1 + (1-a) G_2$ is also in $f^{\ext S}$: by definition,
  $G_1 \geq (h \in \Lform X \mapsto \mathcal G_1 (F \in UX \mapsto F
  (h)))$ for some $\mathcal G_1 \in \mathcal Q$,
  $G_2 \geq (h \in \Lform X \mapsto \mathcal G_2 (F \in UX \mapsto F
  (h)))$ for some $\mathcal G_2 \in \mathcal Q$, and therefore
  $a G_1 + (1-a) G_2 \geq (h \in \Lform X \mapsto \mathcal G (F \in UX
  \mapsto F (h)))$, where
  $\mathcal G \eqdef a \mathcal G_1 + (1-a) \mathcal G_2 \in \mathcal
  Q$.
\end{proof}

\subsection{The weak distributive law $\lambda^\DN$}
\label{sec:weak-distr-law}

We elucidate the corresponding weak distributive law.
\begin{proposition}
  \label{prop:DN:wdistr}
  The weak distributive law associated with the distributing
  retraction of Theorem~\ref{thm:AN:STU} is given at each object $X$
  of $\Topcat$ by:
  \begin{align}
    \nonumber
    \lambda^\DN_X & \colon \Val_\bullet \SV X \to \SV \Val_\bullet
                       X \\
    \label{lambda:sharp:h}
                     & \quad \xi \mapsto \{\nu \in \Val_\bullet X \mid \forall
                       h \in \Lform X, \int_{x \in X} h (x) \,d\nu \geq \int_{Q \in \SV X}
                       \min_{x \in Q} h (x) \,d\xi\} \\
    \label{lambda:sharp:U}
                     & \qquad\quad =
                       \{\nu \in \Val_\bullet X \mid \nu (U) \geq \xi (\Box
                       U) \text{ for every open subset $U$ of $X$}\}.
  \end{align}
\end{proposition}
\begin{proof}
  As in Theorem~\ref{thm:DN:STU}, we replace $\Val_\bullet$ by the
  isomorphic monad $\Pred_\Nature^\bullet$.  Relying on
  (\ref{eq:lambda}),
  $\lambda^\DN_X = s_\DN^\bullet \circ \mu^\DN_X \circ (ji)_X$,
  where $i$ and $j$ are computed as in the proof of
  Theorem~\ref{thm:DN:STU}, namely $j_Y$ is the inclusion of
  $\Pred_\Nature^\bullet Y$ into $\Pred_\DN^\bullet Y$ and $i_Y$ maps
  every $Q \in \SV Y$ to
  $(h \in \Lform Y \mapsto \min_{x \in Q} h (x))$, for every space
  $Y$.  For every $\mathcal G \in \Pred_\Nature^\bullet \SV X$, for
  every $H \in \Lform \Pred_\DN^\bullet X$,
  $(ji)_X (\mathcal G) (H) = \Pred_\DN^\bullet i (j_{\SV X} (\mathcal
  G)) (H) = j_{\SV X} (\mathcal G) (H \circ i) = \mathcal G (H \circ
  i)$.  Therefore, for every $h \in \Lform X$,
  $(\mu^\DN_X \circ (ji)_X) (\mathcal G) (h) = (ji)_X (\mathcal G) (F
  \mapsto F (h)) = \mathcal G (Q \mapsto i (Q) (h)) = \mathcal G (Q
  \mapsto \min_{x \in Q} h (x))$.  Hence
  $\lambda^\DN_X (\mathcal G)$ is the collection of
  $G \in \Pred_\Nature^\bullet X$ such that
  $G (h) \geq \mathcal G (Q \mapsto \min_{x \in Q} h (x))$ for every
  $h \in \Lform X$.  Through the isomorphism
  $\Val_\bullet \cong \Pred_\Nature^\bullet$, for every
  $\xi \in \Val_\bullet \SV X$,
  $\lambda^\DN_X (\xi) = \{\nu \in \Val_\bullet X \mid \forall h
  \in \Lform X, \int_{x \in X} h (x) \,d\nu \geq \int_{Q \in \SV X}
  \min_{x \in Q} h (x) \,d\xi\}$.

  Let us show (\ref{lambda:sharp:U}).  For every
  $\nu \in \lambda^\DN_X (\xi)$, by taking the characteristic map
  $\chi_U$ of an arbitrary open subset $U$ of $X$ (that is indeed in
  $\Lform X$), $\int_{x \in X} h (x) \,d\nu = \nu (U)$, while
  $\int_{Q \in \SV X} \min_{x \in Q} h (x) \,d\xi = \int_{Q \in \SV X}
  \chi_{\Box U} (Q) \,d\xi = \xi (\Box U)$, so
  $\nu (U) \geq \xi (\Box U)$ for every open subset $U$ of $X$.
  Conversely, we assume that $\nu (U) \geq \xi (\Box U)$ for every
  open subset $U$ of $X$.  For every $h \in \Lform X$, let
  $h_* \colon \SV X \to \creal$ map $Q$ to $\min_{x \in Q} h (x)$.
  Then $h_* \in \Lform {\SV X}$, something that we have left implicit
  until now: the min is attained because $Q$ is compact and non-empty,
  and for every $t \in \Rp$,
  $h_*^{-1} (]t, \infty]) = \Box h^{-1} (]t, \infty])$.  Using the
  Choquet formula,
  $\int_{Q \in \SV X} h_* (Q) \,d\xi = \int_0^\infty \xi (h_*^{-1}
  (]t, \infty]))\, dt = \int_0^\infty \xi (\Box h^{-1} (]t,
  \infty]))\, dt \leq \int_0^\infty \nu (h^{-1} (]t, \infty]))\, dt =
  \int_{x \in X} h (x) \,d\nu$.  Since $h$ is arbitrary,
  $\nu \in \lambda^\DN_X (\xi)$.
\end{proof}

\begin{remark}
  \label{rem:DN:wdistr} 
  When $\xi$ is a finite linear combination
  $\sum_{i=1}^n a_i \delta_{Q_i}$, where furthermore each $Q_i$ is the
  upward closure of a non-empty finite set
  $E_i \eqdef \{x_{ij} \mid 1\leq j\leq m_i\}$, then we claim that:
  \begin{align}
    \label{lambda:sharp:fin}
    \lambda^\DN_X (\xi)
    &= \upc \conv \left\{\sum_{i=1}^n a_i \delta_{x_{i
      f(i)}} \mid f \in \prod_{i=1}^n \{1, \cdots, m_i\}\right\},
  \end{align}
  where $\conv$ denotes \emph{convex hull}, namely the smallest convex
  superset of its argument (equivalently, $\conv E$ is the collection
  of all \emph{convex combinations} $\sum_{k=1}^p b_k x_k$ where
  $p \geq 1$, $b_k \in \Rp$, $x_k \in E$ and $\sum_{k=1}^p b_k=1$),
  and $\prod_{i=1}^n \{1, \cdots, m_i\}$ is the set of choice
  functions, mapping each index $i \in \{1, \cdots, n\}$ to an index
  $j \in \{1, \cdots, m_i\}$.  In order to see this, we use
  Proposition~\ref{prop:DN:wdistr}: $\nu \in \lambda^\DN_X (\xi)$ if
  and only if its integration functional
  $G \eqdef \int_{x \in X} \_ (x) \,d\nu$ is such that $G (h)$ is
  larger than or equal to
  $\int_{Q \in \SV X} \min_{x \in Q} h (x) \,d\xi = \sum_{i=1}^n a_i
  \min_{j=1}^{m_i} h (x_{ij})$ for every $h \in \Lform X$.  Now
  $\sum_{i=1}^n a_i \min_{j=1}^{m_i} h (x_{ij}) = \min_{f \in
    \prod_{i=1}^n \{1, \cdots, m_i\}} \sum_{i=1}^n a_i h (x_{i f (i)})
  = r_\DN (\mathcal E)$, where $\mathcal E$ is the collection of
  (integration functionals of) the continuous valuations
  $\sum_{i=1}^n a_i \delta_{x_{i f(i)}}$, and where $f$ ranges over
  $\prod_{i=1}^n \{1, \cdots, m_i\}$.  Hence
  $\lambda^\DN_X (\xi) = s_\DN^\bullet (r_\DN (\mathcal E))$.  By
  \citep[Proposition~4.20]{JGL-mscs16}, $s_\DN^\bullet \circ r_\DN$
  maps every compact saturated subset $Q$ of $\Val_\bullet X$ to its
  compact saturated convex hull, namely to the smallest compact
  saturated convex set that contains it; when $Q$ is finite, this
  coincides with $\upc \conv Q$
  \citep[Lemma~4.10~(c)]{Keimel:topcones2}, so
  $\lambda^\DN_X (\xi) = \upc \conv \mathcal E$.
\end{remark}

\subsection{The algebras of $\Pred_\DN^\bullet$}
\label{sec:algebr-pred_dnb}

Given two monads $S$ and $T$ on the same category $\catc$, a weak
distributive law $\lambda$ of $S$ over $T$, the algebras
$\gamma \colon UX \to X$ of the combined monad $U$ on a given object
$X$ of $\catc$ are in one-to-one correspondence with the pairs of
algebras $(\alpha, \beta)$ of an $S$-algebra $\alpha \colon SX \to X$
and a $T$-algebra $\beta \colon TX \to X$ such that the following
diagram commutes.
\begin{align}
  \label{eq:lambda:distr}
  \xymatrix{
  TSX \ar[r]^{\lambda_X}
  \ar[d]_{T\alpha}
  & STX \ar[r]^{S\beta} & SX \ar[d]^{\alpha}
  \\
  TX \ar[rr]_{\beta} && X
                        }
\end{align}
In one direction, $\gamma$ is obtained as
$\alpha \circ S \beta \circ s_X$, where $r, s$ form a splitting of the
idempotent (\ref{eq:e}).  See \citep[Lemma~14]{Garner:weak:distr} or
\citep[Proposition~3.7]{Bohm:weak:monads}.  Additionally, the morphisms
of $U$-algebras are exactly the morphisms in $\catc$ that are both
$S$-algebra and $T$-algebra morphisms.

The algebras of the $\SV$ monad were elucidated by Schalk on the
category $\Topcat_0$ of $T_0$ topological spaces
\citep[Proposition~7.23]{Schalk:diss}: for every $T_0$ topological space
$X$, there is a $\SV$-algebra structure on $X$ if and only if every
non-empty compact saturated subset $Q$ of $X$ has an greatest lower
bound $\bigwedge Q$ with respect to the specialization ordering and
the map $\bigwedge \colon \SV X \to X$ is continuous.  Then
$\bigwedge$ is itself the unique $\SV$-algebra on $X$.  Additionally,
the morphisms of $\SV$-algebras are those continuous maps $f$ that
preserve such infima, in the sense that
$\bigwedge f [Q] = f (\bigwedge Q)$ for every non-empty compact
saturated set $Q$.

Putting up together everything we know, we therefore obtain the following.
\begin{theorem}
  \label{thm:DN:alg}
  For every $T_0$ topological space $X$, there is a structure of a
  $\Pred_\DN^\bullet$-algebra on $X$ if and only if every non-empty
  compact saturated subset $Q$ of $X$ has a greatest lower bound
  $\bigwedge Q$ and $\bigwedge \colon \SV X \to X$ is continuous, and
  there is a $\Val_\bullet$-algebra structure
  $\beta \colon \Val_\bullet X \to X$ such that
  $\beta \circ \Val_\bullet \bigwedge = \bigwedge \circ \SV \beta
  \circ \lambda^\DN_X$.

  If $\bigwedge$ exists and is continuous, then there is a one-to-one
  correspondence between $\Val_\bullet$-algebra structures $\beta$ on
  $X$ such that
  $\beta \circ \Val_\bullet \bigwedge = \bigwedge \circ \SV \beta
  \circ \lambda^\DN_X$ and $\Pred_\DN^\bullet$-algebra structures
  $\gamma$ on $X$, given by
  $\gamma \eqdef \bigwedge \circ \SV \beta \circ s_\DN^\bullet$.

  The morphisms $f$ of $\Pred_\DN^\bullet$-algebras are the morphisms
  of underlying $\Val_\bullet$-algebras that preserve infima of
  non-empty compact saturated subsets $Q$.
\end{theorem}
There is no definitive characterization of $\Val_\bullet$-algebras on
$\Topcat$.  When $\bullet$ is nothing, one of the best known results
can be found in \citep{GLJ:bary}.  A structure of
$\Val_\bullet$-algebra on $X$ implies a structure of a (topological,
sober, weakly locally convex) \emph{cone} on $X$.  (See
\citet{Keimel:topcones2} for an introduction to semitopological and
topological cones.  Briefly, a cone is a set with an addition
operation $+$ and scalar multiplication by non-negative real numbers
that obey the laws of vector spaces that can be expressed without
minus signs.)  More precise characterizations of
$\Val_\bullet$-algebras were obtained in subcategories of topological
spaces in \citep{GLJ:Valg}.  The $\SV$-algebras on subcategories of
locally compact sober spaces are the continuous characterized by
Schalk as the continuous inf-semilattices
\citep[Proposition~7.28]{Schalk:diss}, namely the continuous dcpos that
have a Scott-continuous binary infimum operation $\wedge$, with their
Scott topology.

On a continuous dcpo $X$, every continuous valuation is
point-continuous in the sense of Reinhold Heckmann
\citep[Theorem~4.1]{heckmann96}, so the monad $\Val$ coincides with
the monad $\Val_p$ of point-continuous valuations on the category
$\Cont$ of continuous dcpo and (Scott-)continuous maps.  The
$\Val_p$-algebras $\beta \colon \Val_p X \to X$ are the weakly locally
convex sober topological cones $X$, and $\beta$ is uniquely determined
by the formula
$\beta (\sum_{i=1}^n a_i \delta_{x_i}) \eqdef \sum_{i=1}^n a_i x_i$
\citep[Theorem~5.11]{GLJ:bary}; the morphisms of $\Val_p$-algebras are
the linear continuous maps \citep[Theorem~5.12]{GLJ:bary}.  By a
theorem of Claire Jones
\citep[Corollary~5.4]{Jones:proba}, the \emph{simple valuations}
$\sum_{i=1}^n \delta_{x_i}$ form a basis of the dcpo $\Val X$
($=\Val_p X$), which implies that continuous maps from $\Val X$ to any
space are uniquely determined by their values on simple valuations.

Using Remark~\ref{rem:DN:wdistr}, we obtain the following.  A
\emph{d-cone} is a cone with a dcpo structure that makes addition and
scalar multiplication (Scott-)continuous.  A continuous d-cone is one
whose underlying dcpo is continuous.  Every continuous d-cone is
topological \citep[Corollary 6.9~(c)]{Keimel:topcones2}, sober
\citep[Proposition 8.2.12 $(b)$]{JGL-topology}, and locally convex
\citep[Lemma 6.12]{Keimel:topcones2} hence weakly locally convex
\citep[Proposition 3.10]{GLJ:Valg}.  A \emph{continuous inf-semilattice
  d-cone} is a continuous d-cone that is also a continuous
inf-semilattice.
\begin{corollary}
  \label{corl:DN:alg:lcsober}
  The $\Pred_\DN$-algebras in the category of locally compact sober
  spaces are the continuous inf-semilattice d-cones in which scalar
  multiplication and addition distribute over binary infima.  The maps
  of $\Pred_\DN$-algebras are the linear, inf-preserving
  Scott-continuous maps.
\end{corollary}
\begin{proof}
  The condition
  $\beta \circ \Val_\bullet \bigwedge = \bigwedge \circ \SV \beta
  \circ \lambda^\DN_X$ reduces to the fact that both sides coincide on
  simple valuations $\xi \eqdef \sum_{i=1}^n a_i \delta_{Q_i}$, where
  each $Q_i$ is the upward closure of a non-empty finite set
  $E_i \eqdef \{x_{ij} \mid 1\leq j\leq m_i\}$.  Then
  $\beta (\Val_\bullet \bigwedge (\xi)) = \beta (\sum_{i=1}^n a_i
  \delta_{\bigwedge_{j=1}^{m_i} x_{ij}}) = \sum_{i=1}^n a_i
  \bigwedge_{j=1}^{m_i} x_{ij}$, while:
  \begin{align*}
    \bigwedge (\SV \beta (\lambda^\DN_X (\xi)))
    & = \bigwedge \left(\SV \beta \left(\upc \conv \left\{\sum_{i=1}^n a_i \delta_{x_{i
      f(i)}} \mid f \in \prod_{i=1}^n \{1, \cdots,
      m_i\}\right\}\right)\right),
  \end{align*}
  by Remark~\ref{rem:DN:wdistr}.  For every finite set $\mathcal E$ of
  simple valuations, we realize that
  $\SV \beta (\upc \conv \mathcal E) = \upc \conv \beta [\mathcal E]$.
  Indeed, an element $y$ of the left-hand side is larger than or equal
  to $\beta (x)$ for some $x \in \upc \conv \mathcal E$, hence larger
  than or equal to the image of some convex combination of elements of
  $\mathcal E$.  But $\beta$ preserves convex combinations of simple
  valuations, so $y$ is larger than or equal to a convex combination
  of elements of $\beta [\mathcal E]$.  Conversely, every convex
  combination of elements of $\beta [\mathcal E]$ is the image under
  $\beta$ of a convex combination of elements of $\mathcal E$, so
  $\conv \beta [\mathcal E] \subseteq \SV \beta (\upc \conv \mathcal
  E)$, and therefore
  $\upc \conv \beta [\mathcal E] \subseteq \SV \beta (\upc \conv
  \mathcal E)$.

  Additionally, for every non-empty finite subset $E$ of points,
  $\bigwedge \upc \conv E = \bigwedge E$.  We have
  $\bigwedge \upc \conv E \leq \bigwedge E$ since
  $E \subseteq \upc \conv E$.  Conversely, let $x \eqdef \bigwedge E$.
  For every convex combination $\sum_{k=1}^p b_k x_k$ of elements of
  $E$, $x \leq x_k$ for every $k$, so
  $x \leq 1 . x = (\sum_{k=1}^p b_k) x = \sum_{k=1}^p b_k x \leq
  \sum_{k=1}^p b_k x_k$.  Hence
  $x \leq \bigwedge \conv E = \bigwedge \upc \conv E$.

  It follows that:
  \begin{align*}
    \bigwedge (\SV \beta (\lambda^\DN_X (\xi)))
    & = \bigwedge \beta \left[\left\{\sum_{i=1}^n a_i \delta_{x_{i
      f(i)}} \mid f \in \prod_{i=1}^n \{1, \cdots,
      m_i\}\right\}\right] \\
    & = \bigwedge_{f \in \prod_{i=1}^n \{1, \cdots, m_i\}}
      \sum_{i=1}^n a_i x_{i f (i)}.
  \end{align*}
  Hence the condition
  $\beta \circ \Val_\bullet \bigwedge = \bigwedge \circ \SV \beta
  \circ \lambda^\DN_X$ reduces to:
  \begin{align*}
    \sum_{i=1}^n a_i \bigwedge_{j=1}^{m_i} x_{ij}
    & = \bigwedge_{f \in \prod_{i=1}^n \{1, \cdots, m_i\}} \sum_{i=1}^n a_i x_{i f (i)}    
  \end{align*}
  for every
  $\xi \eqdef \sum_{i=1}^n a_i \delta_{\upc \{x_{ij} \mid 1\leq j\leq
    m_i\}}$.  When $n=1$, this states that scalar multiplication
  distributes over non-empty finite infima
  ($a_1 \bigwedge_{j=1}^{m_1} x_{1j} = \bigwedge_{j=1}^{m_1} a_1
  x_{1j}$), or equivalently over binary infima.  When $a_i=1$ for
  every $i$, this states that non-empty finite sums distributes over
  non-empty finite infima, or equivalently, addition distributes over
  binary infima.  Conversely, it is easy to see that if both scalar
  multiplication and addition distribute over non-empty finite infima,
  then
  $\sum_{i=1}^n a_i \bigwedge_{j=1}^{m_i} x_{ij} = \bigwedge_{f \in
    \prod_{i=1}^n \{1, \cdots, m_i\}} \sum_{i=1}^n a_i x_{i f (i)}$
  holds.
\end{proof}

\section{Hoare hyperspace, continuous valuations and sublinear
  previsions}
\label{sec:hoare-hypersp-cont}

Just as $\SV$ is a model of demonic non-determinism, the Hoare
powerdomain is a model of \emph{angelic} non-determinism, and is
obtained as follows.  For every topological space $X$, let $\Hoare X$
be the set of non-empty closed subsets of $X$.  The \emph{lower
  Vietoris} topology on that set has subbasic open subsets $\Diamond
U$ consisting of those non-empty closed subsets that intersect $U$,
for each open subset $U$ of $X$.  We write $\HV X$ for the resulting
topological space.  Its specialization ordering is (ordinary)
inclusion $\subseteq$.

The $\HV$ construction was studied by 
\citet[Section~6]{Schalk:diss}, as well as a localic variant and a
variant with the Scott topology.  See also \citep[Sections~6.2.2,
6.2.3]{AJ:domains} or \citep[Section~IV-8]{GHKLMS:contlatt}.

There is a monad $(\HV, \eta^\Hoare, \mu^\Hoare)$ on $\Topcat$.  For
every continuous map $f \colon X \to Y$, $\HV f$ maps every
$C \in \HV X$ to $cl (f [C])$.  The unit $\eta^\Hoare$ is defined by
$\eta^\Hoare_X (x) \eqdef \dc x$ for every $x \in X$, where $\dc$
denotes downward closure (equivalently, $cl (\{x\})$, where $cl$ is
topological closure).  For every continuous map
$f \colon X \to \HV Y$, there is an extension
$f^\flat \colon \HV X \to \HV Y$, defined by
$f^\flat (C) \eqdef cl (\bigcup_{x \in C} f (x))$.  The multiplication
$\mu^\Hoare$ is defined by
$\mu^\Hoare_X \eqdef \identity {\HV X}^\flat$, namely
$\mu^\Hoare_X (\mathcal C) \eqdef cl (\bigcup \mathcal C)$.

\subsection{The distributing retraction $r_\AN$, $s_\AN^\bullet$}
\label{sec:distr-retr-r_an}

Let $\bullet$ be nothing, ``$\leq 1$'', or ``$1$''.  We will see that
there is a distributing retraction of $\HV \Val_\bullet$ over the
monad $\Pred_\AN^\bullet$ of sublinear previsions.  This is very
similar to the case of $\SV$, but the analogues of $r_\DN$ and
$s_\DN^\bullet$ are only known to enjoy the required properties on a
full subcategory of $\Topcat$ consisting of spaces that were called
\emph{$\AN_\bullet$-friendly} in \citep{JGL:mscs16:errata}.  (This
repairs some mistakes in the original paper \citep{JGL-mscs16}.  See
also \citep{JGL:mscs16:rev} for a version of the latter that
incorporates the errata.)  The definition of $\AN_\bullet$-friendly
spaces is technical, and we will omit it here.  We will merely say
that every locally compact 
space is $\AN_\bullet$-friendly.
A space $X$ is \emph{locally compact} if and only if every point has a
base of compact (equivalently, compact saturated) neighborhoods.  We
do no assume Hausdorffness or even the $T_0$ property.

For every $\AN_\bullet$-friendly space $X$, there is a function
$r_\AN \colon \HV \Pred_\Nature^\bullet X \to \Pred_{\AN}^\bullet X$,
defined by $r_\AN (C) (h) \eqdef \sup_{G \in C} \allowbreak G (h)$,
and a map
$s_\AN^\bullet \colon \Pred_{\AN}^\bullet X \to \HV
\Pred_\Nature^\bullet X$ defined by
$s_\AN^\bullet (F) \eqdef \{G \in \Val_\bullet X \mid G \leq F\}$.
Both are continuous, and
$r_\AN \circ s_\AN^\bullet = \identity {\Pred_{\AN}^\bullet X}$.  This
is the content of Corollary~3.12 of \citep{JGL-mscs16}, see also the
errata \citep[Corollary~3.12, repaired]{JGL:mscs16:errata}.

\begin{definition}
  \label{defn:Top:flat}
  Let $\Topcat^\flat$ be any full subcategory of $\Topcat$ consisting
  of $\AN_\bullet$-friendly spaces and that is closed under the three
  functors $\HV$, $\Pred_\Nature^\bullet$ and $\Pred_\AN^\bullet$.
\end{definition}
The category of locally compact spaces fits if $\bullet$ is nothing or
``$\leq 1$''.  If $\bullet$ is ``$1$'', then the category of locally
compact, compact spaces fits.  We argue as follows.  Those categories
consist of $\AN_\bullet$-friendly spaces.  For every locally compact
space $X$, $\HV X$ is locally compact, by \citep[Proposition
6.11]{Schalk:diss}, and $\Val_\bullet X$ is locally compact if
$\bullet$ is nothing or ``$\leq 1$'' (see 
\citep[Theorem~4.1]{JGL:vlc}).  When
$\bullet$ is ``$1$'', $\Val_\bullet X$ is locally compact and compact
if $X$ locally compact and compact \citep[Theorem~4.1]{JGL:vlc}, and
$\HV X$ is also locally compact, as we have seen, and compact: the
proof of \citep[Proposition 6.11]{Schalk:diss} shows in particular that
$\Diamond K \eqdef \{C \in \HV X \mid C \cap K \neq \emptyset\}$ is
compact for every compact subset $K$ of $X$, and the claim follows by
taking $K$ equal to $X$.  Finally, any retract of a locally compact
(resp., compact) space is locally compact (resp., compact); this
follows from the arguments developed in the proof of \citep[Proposition
2.17]{Jung:scs:prob}, and is elementary.  Since $\Pred_\AN X$ is a
retract of $\HV \Pred_\Nature^\bullet X$ for every
$\AN_\bullet$-friendly space through $r_\AN$ and $s^\bullet_\AN$, this
entails that the category of locally compact (resp., and compact)
spaces is also closed under $\Pred_\AN^\bullet$.

We work with $\HV$ as we did with $\SV$.  We only need the following
two auxiliary lemmas.
\begin{lemma}
  \label{lemma:sup:cl}
  For every topological space $X$, for every $h \in \Lform X$, for
  every non-empty subset $A$ of $X$,
  $\sup_{x \in cl (A)} h (x) = \sup_{x \in A} h (x)$.
\end{lemma}
\begin{proof}
  The $\geq$ direction is clear, and in the $\leq$ direction, it
  suffices to show that for every $t < \sup_{x \in cl (A)} h (x)$,
  there is an $x \in A$ such that $t < h (x)$.  If
  $t < \sup_{x \in cl (A)} h (x)$, then $t < h (x)$ for some
  $x \in cl (A)$, namely the open set $h^{-1} (]t, \infty])$
  intersects $cl (A)$.  But any open set that intersects $cl (A)$ must
  intersect $A$.  Hence there is an $x \in A$ such that
  $x \in h^{-1} (]t, \infty])$, namely such that $t < h (x)$.
\end{proof}

\begin{lemma}
  \label{lemma:cl:conv}
  The closure of a convex subset of $\Pred_\AN^\bullet X$ \red{(resp.\
    $\Pred_\Nature^\bullet X$)} is convex.
\end{lemma}
\begin{proof}
  The closure of a convex subset of a semitopological cone is also
  convex \citep[Lemma~4.10~(a)]{Keimel:topcones2}, but
  $\Pred_\AN^\bullet X$ {(resp.\ $\Pred_\Nature^\bullet X$)} is not a
  cone unless $\bullet$ is nothing.  The proof argument is however the
  same, and would generalize to an appropriate notion of
  semitopological barycentric algebra.  Let
  $C \eqdef \Pred_\AN^\bullet X$ {(resp.\ $\Pred_\Nature^\bullet X$)}
  and $A$ be a convex subset of $C$.  Given $a \in [0, 1]$, we define
  a function $f \colon C \times C \to C$ by
  $f (F_1, F_2) \eqdef a F_1 + (1-a) F_2$.  This is separately
  continuous: for example, $f (\_, F_2)^{-1} ([h > r])$ is equal to
  $[h > (r - (1-a) F_2 (h)) / a]$ if $a \neq 0$ and
  $(1-a) F_2 (h) \leq r$, to the whole of $C$ if $(1-a) F_2 (h) > r$
  or if $a=0$.  Let $V$ be the (open) complement of $cl (A)$.  For the
  sake of contradiction, let us assume that there are two elements
  $F_1, F_2 \in cl (A)$ such that $f (F_1, F_2) \not\in cl (A)$.
  Since $f$ is continuous in its second argument,
  $f (F_1, \_)^{-1} (V)$ is an open neighborhood of $F_2$; hence it
  intersects $cl (A)$, and therefore also $A$, say at $F'_2$.  Then
  $f (F_1, F'_2) \in V$.  Since $f$ is continuous in its first
  argument, $f (\_, F'_2)^{-1} (V)$ is an open neighborhood of $F_1$;
  hence it intersects $cl (A)$, and therefore also $A$, say at $F'_1$.
  Then $f (F'_1, F'_2) \in V$.  But since $F'_1, F'_2 \in A$ and $A$
  is convex, $f (F'_1, F'_2) \in A$, which directly contradicts
  $f (F'_1, F'_2) \in V$.
\end{proof}

The maps $r_\AN$ and $s_\AN^\bullet$ restrict to a homeomorphism
between $\Pred_\AN^\bullet X$ and the subspace
$\HV^{cvx} \Pred_\Nature^\bullet X$ of $\HV \Pred_\Nature^\bullet X$
consisting of its convex elements \citep[Theorem~ 4.11]{JGL-mscs16},
see also the errata \citep[Corollary~4.11,
repaired]{JGL:mscs16:errata}.  In particular, $s_\AN^\bullet$ only
takes convex values, and in order to show that two elements
$\mathcal C$, $\mathcal C'$ of $\HV \Pred_\Nature^\bullet X$ are
equal, it suffices to show that they are convex and have the same
image under $r_\AN$.
\begin{lemma}
  \label{lemma:rAP:nat}
  The transformations $r_\AN$ and $s_\AN^\bullet$ are natural on
  $\Topcat^\flat$.
\end{lemma}
\begin{proof}
  Let $f \colon X \to Y$ be any continuous map.  For $r_\AN$, we show
  that for every $C \in \HV \Pred_\Nature^\bullet X$, for every
  $h \in \Lform Y$,
  $r_{\AN} (\HV (\Pred_\AN^\bullet f) (C)) (h) = \Pred_\AN^\bullet f
  (r_{\AN} (C)) (h)$.  The left-hand side is equal to
  $\sup_{G' \in \HV \Pred_\AN^\bullet (f) (C)} G' (h) = \sup_{G' \in
    cl (\{\Pred_\AN^\bullet (f) (G) \mid G \in C\})} \allowbreak G'
  (h) = \sup_{G' \in \{\Pred_\AN^\bullet (f) (G) \mid G \in C\}}
  \allowbreak G' (h)$ (by Lemma~\ref{lemma:sup:cl})
  $= \sup_{G \in C} \Pred_\AN^\bullet f (G) (h)$, while the right-hand
  side is equal to
  $r_{\AN} (C) (h \circ f) = \sup_{G \in C} G (h \circ f)$, and those
  are equal.

  For $s_\AN^\bullet$, we must show that for every
  $F \in \Pred_\AN^\bullet X$,
  $s_{\AN}^\bullet (\Pred_\AN^\bullet (f) (F)) = \HV \Pred_\AN^\bullet
  (f) \allowbreak (s_{\AN}^\bullet (F))$.  We apply the trick
  described above: we show that both sides are convex, and have the
  same image under $r_\AN$.  The left-hand side is convex because
  $s_\AN^\bullet$ only takes convex values.  The right-hand side is
  the closure of
  $A \eqdef \{\Pred_\AN^\bullet (f) (G) \mid G \in s_{\AN}^\bullet
  (F)\}$.  Hence, using Lemma~\ref{lemma:cl:conv}, it suffices to show
  that $A$ is convex.  We consider any two elements $G'_1$, $G'_2$ of
  $A$, and $r \in [0, 1]$.  By definition, there are elements $G_1$,
  $G_2$ of $s_{\AN}^\bullet (F)$ such that
  $\Pred_\AN^\bullet f (G_1) = G'_1$ and
  $\Pred_\AN^\bullet f (G_2) = G'_2$.  Since $s_{\AN}^\bullet (F)$ is
  convex, $rG_1+ (1-r)G_2 \in s_{\AN}^\bullet (F)$, and therefore
  $r G'_1 + (1-r) G'_2 = r \Pred_\AN^\bullet f (G_1) + (1-r)
  \Pred_\AN^\bullet f (G_2) = \Pred_\AN^\bullet f (rG_1 + (1-r) G_2)$
  is in $A$.  It remains to see that
  $s_{\AN}^\bullet (\Pred_\AN^\bullet (f) (F)) = \HV \Pred_\AN^\bullet
  (f) \allowbreak (s_{\AN}^\bullet (F))$ have the same image under
  $r_\AN$.  The image of the left-hand side is
  $\Pred_\AN^\bullet (f) (F)$, while that of the right-hand side is
  $\Pred_\AN^\bullet f (r_{\AN} (s_{\AN}^\bullet (F))$ (by naturality
  of $r_\AN$), hence to $\Pred_\AN^\bullet f (F)$.
\end{proof}


\begin{theorem}
  \label{thm:AN:STU}
  $\xymatrix{\HV \Val_\bullet \ar@<1ex>[r]^{r_\AN} &
    \Pred_{\AN}^\bullet \ar@<1ex>[l]^{s_\AN^\bullet}}$ is a
  distributing retraction on $\Topcat^\flat$.
\end{theorem}
\begin{proof}
  The proof is similar to that of Theorem~\ref{thm:DN:STU}, replacing
  $\SV$ by $\HV$, $\DN$ by $\AN$, and $\min$ by $\sup$.  

  Using Lemma~\ref{lemma:V=PN}, we replace $\Val_\bullet$ by the
  isomorphic monad $\Pred_\Nature^\bullet$.  Naturality is by
  Lemma~\ref{lemma:rAP:nat}.  Let $S \eqdef \HV$,
  $T \eqdef \Pred_\Nature^\bullet$, $U \edef \Pred_\AN^\bullet$,
  $r \eqdef r_\AN$, $s \eqdef s_\AN^\bullet$.  For every space $Y$,
  let $j_Y \eqdef r_Y \circ \eta^S_{TY}$; for every $G \in TY$,
  $j_Y (G) = r (\eta^S_{TY} (G)) = r (\dc G) = G \in UY$; hence $j_Y$
  is the inclusion of $TY$ into $UY$.
  
  Let $i_Y \eqdef r_Y \circ S \eta^T_Y$.  We verify that
  $i_Y (C) = (h \in \Lform Y \mapsto \sup_{y \in C} h (y))$ for every
  $C \in \HV Y$.  Equating $\eta_Y (y) = \delta_y$ with
  $h \in \Lform Y \mapsto h (y)$, we have
  $i_Y (C) (h) = r_Y (\HV \eta_Y (C)) (h) = \sup_{G \in cl (\eta_Y
    [C])} G (h) = \sup_{G \in \eta_Y [C]} G (h)$ (by
  Lemma~\ref{lemma:sup:cl})
  $= \sup_{y \in C} \eta_Y (y) (h) \sup_{y \in C} h (y)$.

  (\ref{eq:STU:eta}) For every $x \in X$,
  $(r \circ \eta^S \eta^T)_X (x) = r (\dc \delta_x)$ is the prevision
  that maps $h \in \Lform X$ to
  $\sup_{G \in \dc \delta_x} G (h) = \delta_x (h) = h (x)$, namely
  $\eta^\AN_X (x) = \eta^U_X (x)$.

  
  (\ref{eq:STU:muS}) For every $\mathcal C \in SST X$,
  $(r_X \circ \mu^S_{TX}) (\mathcal C)$ maps every $h \in \Lform X$ to
  $r_X (cl (\bigcup \mathcal C)) (h) = \sup_{G \in cl (\bigcup
    \mathcal C)} G (h) = \sup_{G \in \bigcup \mathcal C} G (h)$ (by
  Lemma~\ref{lemma:sup:cl})
  $= \sup_{C \in \mathcal C} \sup_{G \in C} G (h)$, while:
  \begin{align*}
    (\mu^U_X \circ (ir)_X) (\mathcal C) (h)
    & = \mu^\AN_X ((ir)_X (\mathcal C)) (h) \\
    & = (ir)_X (\mathcal C) (F \mapsto F (h)) \\
    & = \Pred_\AN^\bullet r (i_{\HV
      \Pred_\Nature^\bullet X} (\mathcal C)) (F \mapsto F (h)) \\
    & = i_{\HV \Pred_\Nature^\bullet X} (\mathcal C)
      ((F \mapsto F (h)) \circ r) \\
    & = i_{\HV \Pred_\Nature^\bullet X} (\mathcal C)
      (F \mapsto F (h)) \circ r) \\
    & = \sup_{C \in \mathcal C} ((F \mapsto F (h)) \circ r) (C)
    = \sup_{C \in \mathcal C} r (C) (h)
    = \sup_{C \in \mathcal C} \sup_{G \in C} G (h).
  \end{align*}
  Therefore $r \circ \mu^S T = \mu^U \circ ir$.

  (\ref{eq:STU:muT}) For every $\mathcal C \in STT X$,
  $(r_X \circ S \mu^T_X) (\mathcal C)$ maps every $h \in \Lform X$ to
  $r_X (cl (\mu^T_X [\mathcal C])) (h) = \sup_{G \in cl (\mu^T_X
    [\mathcal C])} G (h) = \sup_{G \in \mu^T_X [\mathcal C]} G (h)$
  (by Lemma~\ref{lemma:sup:cl})
  $= \sup_{\mathcal G \in \mathcal C} \mu^\Nature_X (\mathcal G) (h) =
  \sup_{\mathcal G \in \mathcal C} \mathcal G (G \mapsto G (h))$,
  while, keeping the convention that $G$ ranges over
  $\Pred_\Nature^\bullet X$ and $F$ over $\Pred_\AN^\bullet X$:
  \begin{align*}
    (\mu^U_X \circ (rj)_X) (\mathcal C) (h)
    & = \mu^\AN_X ((rj)_X (\mathcal C)) (h) \\
    & = (rj)_X (\mathcal C) (F \mapsto F (h)) \\
    & = \Pred_\AN^\bullet j (r_{\Pred_\Nature^\bullet X} (\mathcal C))
      (F \mapsto F (h)) \\
    & = r_{\Pred_\Nature^\bullet X} (\mathcal C) ((F \mapsto F (h))
      \circ j) \\
    & = r_{\Pred_\Nature^\bullet X} (\mathcal C) (G \mapsto G (h))
    & \text{since $j$ is mere inclusion} \\
    & = \sup_{\mathcal G \in \mathcal C} \mathcal G (G \mapsto G (h)).
  \end{align*}
  Therefore $r \circ S \mu^T = \mu^U \circ rj$.

  For the remaining three laws, we recall that $r$ and $s$ define a
  homeomorphism between $\Pred_\AN^\bullet X$ and
  $\HV^{cvx} \Pred_\Nature\bullet X$.  Letting $e \eqdef s \circ r$,
  it follows that $e (C) = C$ if and only if $C$ is convex.

  (\ref{eq:STU:etaS:conv}) states that $\eta^S_{TX}$ always returns a
  convex element of $ST X = \HV {\Pred^\bullet_\Nature X}$, and
  indeed, for every $G \in TX$, $\eta^S_{TX} = \dc G$ is convex.

  (\ref{eq:STU:SmuT:conv}).  \red{Using Lemma~\ref{lemma:STU:conv}, we
    verify (\ref{eq:STU:SmuT:conv'}) instead.  This states that
    $e_X (S \mu^T_X (e_{TX} (\mathcal C))) = S\mu^T (e_{TX} (\mathcal
    C))$, namely that $S\mu^T (e_{TX} (\mathcal C))$ is convex, for
    every $\mathcal C \in STTX$.  But $e_{TX} (\mathcal C)$ is convex,
    and $\mu^T_X$ is affine by Lemma~\ref{lemma:muT:lin}, so the image
    $\mu^T_X [e_{TX} (\mathcal C)]$ is convex, and therefore also its
    closure
    $cl (\mu^T_X [e_{TX} (\mathcal C)]) = S \mu^T_X (e_{TX} (\mathcal
    C))$, by Lemma~\ref{lemma:cl:conv}.}


  (\ref{eq:STU:muUS:conv}).  \red{Using Lemma~\ref{lemma:STU:conv}, we
    verify (\ref{eq:STU:muUS:conv'}) instead.  Let
    $f \eqdef s_X \circ \mu^U_X \circ j_{UX}$.  We need to show that
    $e \circ f^{\ext S} \circ e U = f^{\ext S} \circ e U$.  As in the
    previous case, it suffices to show that $f^{\ext S}$ maps} convex
  sets to convex sets.  
  For every $\mathcal G \in TUX$,
  $f (\mathcal G) = s_X (\mu^U_X (\mathcal G)) = \{G \in TX \mid G
  \leq (h \in \Lform X \mapsto \mu^U_X (\mathcal G))\} = \{G \in TX
  \mid G \leq (h \in \Lform X \mapsto \mathcal G (F \in UX \mapsto F
  (h)))\}$.
  
  Then, $\mathcal C \in STUX$,
  $f^{\ext S} (\mathcal C) = cl (\bigcup f [\mathcal C])$, and we
  claim that $\bigcup f [\mathcal C]$ is convex.  Indeed, for all
  $G_1, G_2 \in f^{\ext S} (\mathcal C)$, for every $a \in [0, 1]$, by
  definition,
  $G_1 \leq (h \in \Lform X \mapsto \mathcal G_1 (F \in UX \mapsto F
  (h)))$ for some $\mathcal G_1 \in \mathcal C$,
  $G_2 \leq (h \in \Lform X \mapsto \mathcal G_2 (F \in UX \mapsto F
  (h)))$ for some $\mathcal G_2 \in \mathcal C$, and therefore
  $a G_1 + (1-a) G_2 \leq (h \in \Lform X \mapsto \mathcal G (F \in UX
  \mapsto F (h)))$, where
  $\mathcal G \eqdef a \mathcal G_1 + (1-a) \mathcal G_2 \in \mathcal
  C$.

  Since $\bigcup f [\mathcal C]$ is convex, its closure is also convex
  by Lemma~\ref{lemma:cl:conv}, hence $f^{\ext S} (\mathcal C)$ is
  convex, and this concludes the proof.
\end{proof}

\subsection{The weak distributive law $\lambda^\AN$}
\label{sec:weak-distr-law-1}

As with $\SV$, we elucidate the corresponding weak distributive law.
There is no surprise here.
\begin{proposition}
  \label{prop:AN:wdistr}
  The weak distributive law associated with the distributing
  retraction of Theorem~\ref{thm:AN:STU} is given at each object $X$
  of $\Topcat^\flat$ by:
  \begin{align}
    \nonumber
    \lambda^\AN_X & \colon \Val_\bullet \HV X \to \HV \Val_\bullet
                       X \\
    \label{lambda:flat:h}
                     & \quad \xi \mapsto \{\nu \in \Val_\bullet X \mid \forall
                       h \in \Lform X, \int_{x \in X} h (x) \,d\nu \leq \int_{C \in \HV X}
                       \sup_{x \in C} h (x) \,d\xi\} \\
    \label{lambda:flat:U}
                     & \qquad\quad =
                       \{\nu \in \Val_\bullet X \mid \nu (U) \leq \xi (\Diamond
                       U) \text{ for every open subset $U$ of $X$}\}.
  \end{align}
\end{proposition}
\begin{proof}
  As usual, we replace $\Val_\bullet$ by the isomorphic monad
  $\Pred_\Nature^\bullet$.  Relying on (\ref{eq:lambda}),
  $\lambda^\AN_X = s_\AN^\bullet \circ \mu^\AN_X \circ (ji)_X$,
  where $i$ and $j$ are computed as in the proof of
  Theorem~\ref{thm:AN:STU}, namely $j_Y$ is the inclusion of
  $\Pred_\Nature^\bullet Y$ into $\Pred_\AN^\bullet Y$ and $i_Y$ maps
  every $C \in \HV Y$ to
  $(h \in \Lform Y \mapsto \sup_{x \in C} h (x))$, for every space
  $Y$.  For every $\mathcal G \in \Pred_\Nature^\bullet \HV X$, for
  every $H \in \Lform \Pred_\AN^\bullet X$,
  $(ji)_X (\mathcal G) (H) = \Pred_\AN^\bullet i (j_{\HV X} (\mathcal
  G)) (H) = j_{\HV X} (\mathcal G) (H \circ i) = \mathcal G (H \circ
  i)$.  Therefore, for every $h \in \Lform X$,
  $(\mu^\AN_X \circ (ji)_X) (\mathcal G) (h) = (ji)_X (\mathcal G) (F
  \mapsto F (h)) = \mathcal G (C \mapsto i (C) (h)) = \mathcal G (C
  \mapsto \sup_{x \in C} h (x))$.  Hence
  $\lambda^\AN_X (\mathcal G)$ is the collection of
  $G \in \Pred_\Nature^\bullet X$ such that
  $G (h) \leq \mathcal G (C \mapsto \sup_{x \in C} h (x))$ for every
  $h \in \Lform X$.  Through the isomorphism
  $\Val_\bullet \cong \Pred_\Nature^\bullet$, for every
  $\xi \in \Val_\bullet \HV X$,
  $\lambda^\AN_X (\xi) = \{\nu \in \Val_\bullet X \mid \forall h \in
  \Lform X, \int_{x \in X} h (x) \,d\nu \leq \int_{C \in \HV X}
  \sup_{x \in C} h (x) \,d\xi\}$.

  We prove (\ref{lambda:flat:U}).  For every
  $\nu \in \lambda^\AN_X (\xi)$, by taking the characteristic map
  $\chi_U$ of an arbitrary open subset $U$ of $X$,
  $\int_{x \in X} h (x) \,d\nu = \nu (U)$, while
  $\int_{C \in \HV X} \sup_{x \in C} h (x) \,d\xi = \int_{C \in \HV X}
  \chi_{\Diamond U} (C) \,d\xi = \xi (\Diamond U)$, so
  $\nu (U) \leq \xi (\Diamond U)$ for every open subset $U$ of $X$.
  Conversely, we assume that $\nu (U) \leq \xi (\Diamond U)$ for every
  open subset $U$ of $X$.  For every $h \in \Lform X$, let
  $h^* \colon \HV X \to \creal$ map $C$ to $\sup_{x \in C} h (x)$.
  Then $h^* \in \Lform {\HV X}$, and for every $t \in \Rp$,
  ${(h^*)}^{-1} (]t, \infty]) = \Diamond h^{-1} (]t, \infty])$.  Using
  the Choquet formula,
  $\int_{C \in \HV X} h^* (C) \,d\xi = \int_0^\infty \xi ({(h^*)}^{-1}
  (]t, \infty]))\, dt = \int_0^\infty \xi (\Diamond h^{-1} (]t,
  \infty]))\, dt \geq \int_0^\infty \nu (h^{-1} (]t, \infty]))\, dt =
  \int_{x \in X} h (x) \,d\nu$.  Since $h$ is arbitrary,
  $\nu \in \lambda^\AN_X (\xi)$.
\end{proof}

\begin{remark}
  \label{rem:AN:wdistr} 
  When $\xi$ is a finite linear combination
  $\sum_{i=1}^n a_i \delta_{C_i}$, where furthermore each $C_i$ is the
  closure (equivalently, the downward closure) of a non-empty finite
  set $E_i \eqdef \{x_{ij} \mid 1\leq j\leq m_i\}$, then we claim
  that:
  \begin{align}
    \label{lambda:flat:fin}
    \lambda^\AN_X (\xi)
    &= cl (\conv \left\{\sum_{i=1}^n a_i \delta_{x_{i
      f(i)}} \mid f \in \prod_{i=1}^n \{1, \cdots, m_i\}\right\}).
  \end{align}
  To see this, we use Proposition~\ref{prop:AN:wdistr}:
  $\nu \in \lambda^\AN_X (\xi)$ if and only if
  $G \eqdef \int_{x \in X} \_ (x) \,d\nu$ is such that $G (h)$ is
  smaller than or equal to
  $\int_{C \in \HV X} \sup_{x \in C} h (x) \,d\xi = \sum_{i=1}^n a_i
  \sup_{j=1}^{m_i} h (x_{ij})$ for every $h \in \Lform X$.  Now
  $\sum_{i=1}^n a_i \sup_{j=1}^{m_i} h (x_{ij}) = \sup_{f \in
    \prod_{i=1}^n \{1, \cdots, m_i\}} \allowbreak \sum_{i=1}^n a_i h
  (x_{i f (i)}) = r_\AN (\mathcal E)$, where $\mathcal E$ is the
  collection of (integration functionals of) the continuous valuations
  $\sum_{i=1}^n a_i \delta_{x_{i f(i)}}$, and where $f$ ranges over
  $\prod_{i=1}^n \{1, \cdots, m_i\}$.  Hence
  $\lambda^\AN_X (\xi) = s_\AN^\bullet (r_\AN (\mathcal E))$.  By
  \citep[Proposition~4.19]{JGL-mscs16}, $s_\AN^\bullet \circ r_\AN$
  maps every non-empty closed subset $C$ of $\Val_\bullet X$ to its
  closed convex hull $cl (\conv C)$, so
  $\lambda^\AN_X (\xi) = cl (\conv \mathcal E)$.
\end{remark}

\subsection{The algebras of $\Pred_\AN^\bullet$}
\label{sec:algebr-pred}

The $\HV$-algebras on $\Topcat_0$ were identified by 
\citet[Theorem~6.9]{Schalk:diss}.  They are the sober inflationary
topological semilattices, namely the semilattices $(L, \vee)$ with a
topology that makes $L$ sober and addition $\vee$ (jointly)
continuous, and where $x \leq x \vee y$ for all $x, y \in L$; the
latter is what ``inflationary'' means.  (For sobriety, see
\citep[Chapter~8]{JGL-topology}.)  Equivalently, $\vee$ is binary
supremum with respect to the specialization ordering of $L$, and is
required to exist and to be continuous, in addition to the property of
sobriety.  The actual structure maps are supremum maps
$\bigvee \colon \HV L \to L$, which automatically exist and are
continuous.  The morphisms of $\HV$-algebras are the
addition-preserving continuous maps.  It follows:
\begin{theorem}
  \label{thm:AN:alg}
  For every $T_0$ topological space $X$, there is a structure of a
  $\Pred_\AN^\bullet$-algebra on $X$ if and only if $X$ has a
  structure of a sober inflationary topological semilattice, and there
  is a $\Val_\bullet$-algebra structure
  $\beta \colon \Val_\bullet X \to X$ such that
  $\beta \circ \Val_\bullet \bigvee = \bigvee \circ \SV \beta \circ
  \lambda^\AN_X$.

  If $X$ is a sober inflationary topological semilattice, then there
  is a one-to-one correspondence between $\Val_\bullet$-algebra
  structures $\beta$ on $X$ such that
  $\beta \circ \Val_\bullet \bigvee = \bigvee \circ \SV \beta \circ
  \lambda^\AN_X$ and $\Pred_\AN^\bullet$-algebra structures $\gamma$
  on $X$, given by
  $\gamma \eqdef \bigvee \circ \HV \beta \circ s_\AN^\bullet$.

  The morphisms $f$ of $\Pred_\AN^\bullet$-algebras are the morphisms
  of underlying $\Val_\bullet$-algebras that preserve binary infima.
\end{theorem}
We specialize this to the category $\Cont$ of continuous dcpos, as we
did with locally compact sober spaces in
Corollary~\ref{corl:DN:alg:lcsober}.  A \emph{continuous
  sup-semilattice} is a continuous dcpo with a (necessarily
Scott-continuous) binary supremum operation $\vee$.  A
\emph{continuous sup-semilattice d-cone} is a continuous d-cone that
is also a continuous sup-semilattice.
\begin{corollary}
  \label{corl:AN:alg:lcsober}
  The $\Pred_\AN$-algebras in $\Cont$ are the continuous
  sup-semilattice d-cones in which scalar multiplication and addition
  distribute over binary suprema.  The maps of $\Pred_\DN$-algebras
  are the linear, sup-preserving Scott-continuous maps.
\end{corollary}
\begin{proof}
  This is proved as Corollary~\ref{corl:DN:alg:lcsober}, with $\vee$
  replacing $\wedge$ and $cl (\conv \_)$ replacing $\upc \conv \_$: we
  reduce the condition
  $\beta \circ \Val_\bullet \bigvee = \bigvee \circ \HV \beta \circ
  \lambda^\AN_X$ to a more manageable condition.

  For every finite set $\mathcal E$ of simple valuations, we note that
  $\HV \beta (cl (\conv \mathcal E)) = cl (\conv \beta [\mathcal E])$.
  The left-hand side is equal to $cl (\beta [cl (\conv \mathcal E)])$.
  Since $\beta$ is continuous,
  $\beta [cl (\conv \mathcal E)]) \subseteq cl (\beta [\conv \mathcal
  E])$, and since $\beta$ preserves convex combinations of simple
  valuations, $\beta [\conv \mathcal E] = \conv (\beta [\mathcal E])$,
  so
  $cl (\beta [cl (\conv \mathcal E)]) \subseteq cl (\conv \beta
  [\mathcal E])$.  Conversely,
  $\conv \beta [\mathcal E] = \beta [\conv \mathcal E] \subseteq cl
  (\beta [cl (\conv \mathcal E)])$, so
  $cl (\conv \beta [\mathcal E]) \subseteq cl (\beta [cl (\conv
  \mathcal E)])$.

  Also, for every non-empty finite subset $E$ of points,
  $\bigvee cl (\conv E) = \bigvee E$.  Indeed,
  $\bigvee E \leq \bigvee cl (\conv E)$ since
  $E \subseteq cl (\conv E)$, while conversely, it suffices to show
  that every open neighborhood $U$ of $\bigvee cl (\conv E)$ contains
  $\bigvee E$.  Since $\bigvee$ is the structure map of a
  $\HV$-algebra, it is in particular continuous, so there is an open
  neighborhood $\mathcal U$ of $cl (\conv E)$ in $\HV X$ whose image
  under $\bigvee$ is included in $U$.  By definition of the lower
  Vietoris topology, we may assume that $\mathcal U$ is an
  intersection $\bigcap_{k=1}^p \Diamond {U_k}$, where each $U_k$ is
  open in $X$.  If $p=0$, then $\mathcal U$ is the whole Hoare
  hyperspace, and its image under $\bigvee$, which contains
  $\bigvee \dc x$ for every point $x$, contains every point; hence $U$
  contains every point, and therefore $\bigvee E$ in particular.  We
  therefore assume $p \geq 1$.  Since $cl (\conv E) \in \mathcal U$,
  $cl (\conv E)$ intersects every $U_k$, so $\conv E$ intersects every
  $U_k$.  Hence there are convex combinations $x_k$ of elements of $E$
  such that $x_k \in U_k$ for every $k$.  Let $x \eqdef \bigvee E$.
  Every convex combination of elements of $E$ is a convex combination
  of elements smaller than or equal to $x$, hence is itself smaller
  than or equal to $x$.  Hence $x_k \leq x$ for every $k$.  Let $C$ be
  the non-empty (since $p \geq 1$) closed set
  $\dc \{x_1, \cdots, x_p\}$.  This is an element of the Hoare
  hyperspace, and $\bigvee C \leq x$.  But $\bigvee$ maps $\mathcal U$
  to $U$, and $C$ is in $\mathcal U = \bigcap_{k=1}^p \Diamond {U_k}$,
  so $\bigvee C \in U$.  Since $U$ is upwards closed, $x \in U$, as
  promised.

  The condition
  $\beta \circ \Val_\bullet \bigvee = \bigvee \circ \HV \beta \circ
  \lambda^\AN_X$ reduces to the fact that both sides coincide on
  simple valuations $\xi \eqdef \sum_{i=1}^n a_i \delta_{C_i}$, where
  each $C_i$ is the (downward) closure of a non-empty finite set
  $E_i \eqdef \{x_{ij} \mid 1\leq j\leq m_i\}$.  Then
  $\beta (\Val_\bullet \bigvee (\xi)) = \beta (\sum_{i=1}^n a_i
  \delta_{\bigvee_{j=1}^{m_i} x_{ij}}) = \sum_{i=1}^n a_i
  \bigvee_{j=1}^{m_i} x_{ij}$, while, using
  Remark~\ref{rem:AN:wdistr} and the remarks made above:
  \begin{align*}
    \bigvee (\HV \beta (\lambda^\AN_X (\xi)))
    & = \bigvee \left(\HV \beta \left(cl (\conv \left\{\sum_{i=1}^n a_i \delta_{x_{i
      f(i)}} \mid f \in \prod_{i=1}^n \{1, \cdots,
      m_i\})\right\}\right)\right) \\
      & = \bigvee \beta \left[\left\{\sum_{i=1}^n a_i \delta_{x_{i
      f(i)}} \mid f \in \prod_{i=1}^n \{1, \cdots,
      m_i\}\right\}\right] \\
    & = \bigvee_{f \in \prod_{i=1}^n \{1, \cdots, m_i\}}
      \sum_{i=1}^n a_i x_{i f (i)}.
  \end{align*}
  It is easy to see that this reduces to the fact that scalar
  multiplication and addition distribute over binary suprema.
\end{proof}

\section{Lenses, continuous valuations, and forks}
\label{sec:plotk-powerd-cont}

Combining $\SV$ and $\HV$, there is a monad of lenses, which we define
below.  The variant with the Scott topology is sometimes referred to
as the Plotkin powerdomain, or the convex powerdomain, and models
erratic non-determinism.

A \emph{lens} on a space $X$ is a non-empty set of the form $Q \cap C$
where $Q$ is compact saturated and $C$ is closed in $X$.  The
\emph{Vietoris topology} has subbasic open subsets of the form
$\Box U$ (the set of lenses included in $U$) and $\Diamond U$ (the set
of lenses that intersect $U$), for each open subset $U$ of $X$.  We
write $\Plotkin X$ for the set of lenses on $X$, and we let
$\Plotkinn X$ denote $\Plotkin X$ with the Vietoris topology.  The
specialization ordering of $\Plotkinn X$ is the \emph{topological
  Egli-Milner ordering}: $L \TEMleq L'$ if and only if
$\upc L \supseteq \upc L'$ and $cl (L) \subseteq cl (L')$
\citep[Discussion before Fact~4.1]{GL:duality}.  This is an ordering,
not just a preordering, hence $\Plotkinn X$ is $T_0$.

We will work in the category $\SComp$ of stably compact spaces and
continuous maps.  That is not the most general category we could use,
and core-compact core-coherent spaces (with an additional requirement
for compactness if $\bullet$ is ``$1$'') would be enough for us to use
the results of \citep[Section~3.3]{JGL-mscs16}.

A \emph{stably compact space} is a sober, locally compact, compact
coherent space; see \citep[Chapter~9]{JGL-topology} for definitions and
results.  Every \emph{compact pospace} $(X, \preceq)$, namely every
(necessarily Hausdorff) compact space $X$ with an ordering $\preceq$
whose graph $(\preceq)$ is closed in $X \times X$, provides two
examples of stably compact spaces: $X$ with the \emph{upward
  topology}, whose open subsets are the open subsets of $X$ that are
upwards closed with respect to $\preceq$, and $X$ with the
\emph{downward topology}, defined similarly, with ``upwards'' replaced
by ``downwards''.  Conversely, every stably compact space $X$ yields a
compact pospace $(X^\patch, \leq)$, where $X^\patch$ has the same
points as $X$ and has the \emph{patch topology} (the coarsest topology
that makes all open subsets of $X$ open and all compact saturated
subsets of $X$ closed), and $\leq$ is the specialization ordering of
$X$.  The two constructions are inverse of each other.  In particular,
the upward topology of $(X^\patch, \leq)$ is the original topology on
$X$; the downward topology consists of the complements of compact
saturated subsets of $X$, and is called the \emph{cocompact topology}.
The space $X^\dG$, whose points are those of $X$ and whose topology is
the cocompact topology, is called the \emph{de Groot dual} of $X$.


Given a stably compact space $X$, the closure $cl (L)$ of every lens
coincides with its downward closure $\dc L$
\citep[Lemma~4.2]{GL:duality}.  Hence, for two lenses $L$ and $L'$,
$L \TEMleq L'$ if and only if $\upc L \supseteq \upc L'$ and
$\dc L \subseteq \dc L'$, an ordering known as the \emph{Egli-Milner
  ordering}.  A canonical representation as $Q \cap C$ for a lens $L$
is as $\upc L \cap \dc L$.  We note that $\upc L$ is compact saturated
and that $\dc L = cl (L)$ is closed.  The equality is an easy
exercise, noting that any lens $L = Q \cap C$, with $Q$ compact
saturated and $C$ closed (hence downwards closed) must be
\emph{order-convex}, namely any point sitting between two points of
$L$ in the specialization preordering $\leq$ of $X$ must itself be in
$L$.

\begin{proposition}
  \label{prop:AD:monad}
  There is a monad
  $(\Plotkinn, \allowbreak \eta^\pl, \allowbreak \mu^\pl)$ on
  $\SComp$, whose unit and multiplication are defined at every
  topological space $X$ by:
  \begin{itemize}
  \item for every $x \in X$,
    $\eta^\pl_X (x) = 
    \{x\}$,
  \item for every $\mathcal L \in \Plotkinn \Plotkinn X$, $\mu^\pl_X
    (\mathcal L) = \bigcup \{\upc L \mid L \in \mathcal L\} \cap cl
    (\bigcup \{\dc L \mid L \in \mathcal L\})$,
  \end{itemize}
  and whose action on morphisms $f \colon X \to Y$ is given by
  $\Plotkinn f (L) = \upc f [L] \cap cl (f[L])$.  For every continuous
  map $f \colon X \to \Plotkinn Y$, the extension $f^\natural \colon
  \Plotkinn X \to \Plotkinn Y$ ($= \mu^\pl_Y \circ \Plotkinn f$) is
  given by $f^\natural (L) \eqdef \bigcup \{\upc f (x) \mid x \in L\}
  \cap cl (\bigcup \{\dc f (x) \mid x \in L\})$.
\end{proposition}
\begin{proof}
  By \citep[Theorem~6.7, item~2]{Lawson:scomp}, $\Plotkinn X$ is stably
  compact for every stably compact space $X$.  The map $\eta^\pl_X$ is
  continuous, because the inverse image of both $\Box U$ and $\Diamond
  U$ is $U$, for every open subset $U$ of $X$.

  Using Manes' characterization of monads, we define $f^\natural$, the
  desired extension of $f \colon X \to \Plotkinn Y$, by the formula
  $f^\natural (L) \eqdef \bigcup \{\upc f (x) \mid x \in X\} \cap cl
  (\bigcup \{\dc f (x) \mid x \in X\})$ for every $L \in \Plotkinn X$.
  The continuity of $f^\natural$ will follow from the following two
  properties, valid for all open subsets $V$ of $Y$:
  \begin{itemize}
  \item $(f^\natural)^{-1} (\Box V) = \Box f^{-1} (\Box V)$.  In order
    to prove this, we need to show that for every $L \in \Plotkinn X$,
    $f^\natural (L) \subseteq V$ if and only if for every $x \in L$,
    $f (x) \subseteq V$.  If the latter holds, then
    $\upc f (x) \subseteq V$ for every $x \in L$ since $V$ is upwards
    closed, so
    $f^\natural (L) \subseteq \bigcup \{\upc f (x) \mid x \in X\}
    \subseteq V$.  Conversely, if $f^\natural (L) \subseteq V$, then
    for every $x \in L$,
    $f (x) \in \upc f (x) \cap \dc f (x) \subseteq f^\natural (L)
    \subseteq V$.
  \item
    $(f^\natural)^{-1} (\Diamond V) = \Diamond f^{-1} (\Diamond V)$.
    We need to show that for every $L \in \Plotkinn X$,
    $f^\natural (L)$ intersects $V$ if and only if $f (x)$ intersects
    $V$ for some $x \in L$.  If $f (x)$ intersects $V$ for some
    $x \in L$, say at $y$, then $y$ is both in $\upc f (x)$ and in
    $\dc f (x)$, hence in their intersection, hence also in the larger
    set $f^\natural (L)$; so $f^\natural (L)$ intersects $V$.
    Conversely, if $f^\natural (L)$ intersects $V$, then then the
    larger set $cl (\bigcup \{\dc f (x) \mid x \in X\})$ also
    intersects $V$.  Since $V$ is open, it must intersect
    $\{\dc f (x) \mid x \in X\}$, hence a set $\dc f (x)$ for some
    $x \in X$.  Since $V$ is upwards closed, $V$ intersects $f (x)$.
  \end{itemize}
  The formulae for $\Plotkinn f$ and $\mu^\pl$ follow immediately.  We
  check Manes' axioms.  In order to do this, it is practical to rely
  on the easily proved fact that given any two continuous functions
  $h, k \colon X \to Y$, where $Y$ is $T_0$, if
  $h^{-1} (V) = k^{-1} (V)$ for every element $V$ of a given subbase
  of the topology of $Y$, then $h=k$.  (If the assumption holds, it
  must hold for every open subset $V$ of $Y$ as well.  Then, for every
  $x \in X$, we show that $h (x) = k (x)$ by showing that $h (x)$ and
  $k (x)$ have the same open neighborhoods $V$: $h (x) \in V$ iff
  $x \in h^{-1} (V)$, iff $x \in k^{-1} (V)$, iff $k (x) \in V$.)  We
  apply this trick to each equation $h = k$ that needs to be proved,
  recalling that all the spaces $\Plotkinn Y$ are $T_0$:
  \begin{itemize}
  \item[\ref{it:Manes:eta}] the inverse image of $\Box U$ under
    $(\eta^\pl_X)^\natural$ is $\Box {(\eta^\pl_X)^{-1} (\Box U)} =
    \Box U$, and similarly the inverse image of $\Diamond U$ is
    $\Diamond U$; $\identity {\Plotkinn X}$ has the same inverse
    images;
  \item[\ref{it:Manes:dagger:1}] the inverse image of $\Box V$ under
    $f^\natural \circ \eta^\pl_X$ is $(\eta^\pl_X)^{-1}
    ((f^\natural)^{-1} (\Box V)) = (\eta^\pl_X)^{-1} (\Box f^{-1}
    (\Box V)) = f^{-1} (\Box V)$, and similarly with $\Diamond$;
  \item[\ref{it:Manes:dagger:2}] the inverse image of $\Box V$ under
    $g^\natural \circ f^\natural$ is
    $(f^\natural)^{-1} ((g^\natural)^{-1} (\Box V)) =
    (f^\natural)^{-1} (\Box g^{-1} (\Box V)) = \Box f^{-1} (\Box
    g^{-1} (\Box V))$, while its inverse image under
    $(g^\natural \circ f)^\natural$ is
    $\Box (g^\natural \circ f)^{-1} (\Box V) = \Box f^{-1}
    ((g^\natural)^{-1} (\Box V)) = \Box f^{-1} (\Box g^{-1} (\Box
    V))$; similarly with $\Diamond$.
  \end{itemize}
%
\end{proof}

We will need the following later on.  For every stably compact space
$X$, we define $\varpi_1 \colon \Plotkinn X \to \SV X$ as mapping $L$
to $\upc L$ and $\varpi_2 \colon \Plotkinn X \to \HV X$ as mapping $L$
to $\dc L$.

\begin{lemma}
  \label{lemma:Pf:L}
  The following hold on $\SComp$:
  \begin{enumerate}
  \item\label{it:Pf:L:varpi} $\varpi_1$ and $\varpi_2$ are continuous;
  \item\label{it:Pf:L:dagger} for every morphism
    $f \colon X \to \Plotkinn Y$,
    $\varpi_1 \circ f^\natural = (\varpi_1 \circ f)^\sharp \circ
    \varpi_1$ and
    $\varpi_2 \circ f^\natural = (\varpi_2 \circ f)^\flat \circ
    \varpi_2$; explicitly, for every $L \in \Plotkinn X$,
    $\upc f^\natural (L) = (x \in X \mapsto \upc f (x))^\sharp (\upc
    L)$ and
    $\dc f^\natural (L) = (x \in X \mapsto \dc f (x))^\flat (\dc L)$;
  \item\label{it:Pf:L:Pf} for every morphism $f \colon X \to Y$,
    $\varpi_1 \circ \Plotkinn f = \SV f \circ \varpi_1$ and
    $\varpi_2 \Plotkinn f = \HV f \circ \varpi_2$;
  \item\label{it:Pf:L:mu}
    $\varpi_1 \circ \mu^\pl = \mu^\Smyth \circ \SV \varpi_1 \circ
    \varpi_1$ and
    $\varpi_2 \circ \mu^\pl = \mu^\Hoare \circ \HV \varpi_2 \circ
    \varpi_2$.
  \end{enumerate}
\end{lemma}
\begin{proof}
  \ref{it:Pf:L:varpi}.  For every open subset $V$ of $X$,  the inverse
  image of $\Box V$ (as a subbasic open subset of $\SV X$) by
  $\varpi_1$ is $\Box V$ (as a subbasic open subset of $\Plotkinn
  X$).  Similarly with $\varpi_2$ and $\Diamond$, where we must note
  that $\dc L$ is closed for every lens $L$, since $\dc L = cl (L)$,
  thanks to the fact that $X$ is stably compact.
  
  \ref{it:Pf:L:dagger}.  Let
  $f^- \eqdef (\varpi_1 \circ f) = (x \in X \mapsto \upc f (x)) \colon
  X \to \SV Y$ and
  $f^+ \eqdef (\varpi_2 \circ f) = (x \in X \mapsto \dc f (x)) \colon
  X \to \HV Y$.

  We recall that
  $f^\natural (L) = \bigcup \{\upc f (x) \mid x \in L\} \cap cl
  (\bigcup \{\dc f (x) \mid x \in L\})$, so
  $f^\natural (L) = \bigcup f^- [L] \cap cl (\bigcup f^+ [L])$.  Hence
  $f^\natural (L) \subseteq \bigcup f^- [L] \subseteq \bigcup f^-
  [\upc L] = (f^-)^\sharp (\upc L)$.  Since the latter is upwards
  closed, $\upc f^\natural (L) \subseteq (f^-)^\sharp (\upc L)$.
  Similarly, $\dc f^\natural (L) \subseteq (f^+)^\flat (\dc L)$.
  
  Conversely, every element $y$ of $(f^-)^\sharp (\upc L)$ is in
  $\upc f (x)$ for some $x \in \upc L$.  Let $x' \in L$ be such that
  $x' \leq x$.  Since $f$ is continuous, hence monotonic with
  respective to specialization preorderings, $f (x') \TEMleq f (x)$,
  and in particular $\upc f (x') \supseteq \upc f (x)$; so
  $y \in \upc f (x')$, where $x' \in L$.  Let $y' \in f (x')$ be such
  that $y' \leq y$.  Since $f (x')$ is included in both $\upc f (x')$
  and $\dc f (x')$, it is included in $f^\natural (L)$.  Therefore
  $y \in \upc f^\natural (L)$.

  Similarly, every element $y$ of $f^+ [\dc L]$ is smaller than or
  equal to some element of $f (x)$, for some $x \in \dc L$.  Let
  $x' \in L$ be such that $x \leq x'$.  Since $f$ is monotonic,
  $f (x) \TEMleq f (x')$, and in particular
  $cl (f (x)) \subseteq cl (f (x'))$.  Since closed sets are
  downwards-closed, $y \in cl (f (x))$, so $y \in cl (f (x'))$.  Since
  $f (x')$ is a lens in a compact space, $y \in \dc f (x')$, so there
  is a point $y' \in f (x')$ such that $y \leq y'$.  Now $f (x')$ is
  included in both $\upc f (x')$ and $\dc f (x')$, so it is included
  in $f^\natural (L)$.  Therefore $y \in \dc f^\natural (L)$.  This
  holds for every $y \in f^+ [\dc L]$, and $\dc f^\natural (L)$ is
  closed, since $f^\natural (L)$ is a lens in a stably compact space,
  so $cl (f^+ [\dc L]) \subseteq \dc f^\natural (L)$; namely
  $(f^+)^\flat (\dc L) \subseteq \dc f^\natural (L)$.

  \ref{it:Pf:L:Pf}.  We note that $\varpi_1 \circ \eta^\pl_Y$ maps
  every $y \in Y$ to $\upc \{y\} = \eta^\sharp_Y (y)$.  Then, by
  \ref{it:Pf:L:dagger},
  $\varpi_1 \circ \Plotkinn f = \varpi_1 \circ (\eta^\pl_Y \circ
  f)^\natural = (\varpi_1 \circ \eta^\pl_Y \circ f)^\sharp \circ
  \varpi_1 = (\eta^\sharp_Y \circ f)^\sharp \circ \varpi_1 = \SV f
  \circ \varpi_1$.  We show
  $\varpi_2 \Plotkinn f = \HV f \circ \varpi_2$ similarly.

  \ref{it:Pf:L:mu}.  For every stably compact space $X$,
  $\varpi_1 \circ \mu^\pl_X = \varpi_1 \circ (\identity {\Plotkinn
    X})^\natural = (\varpi_1 \circ \identity {\Plotkinn X})^\sharp
  \circ \varpi_1 = \varpi_1^\sharp \circ \varpi_1$, by
  \ref{it:Pf:L:dagger}; and this is equal to
  $\mu^\Smyth \circ \SV \varpi_1 \circ \varpi_1$.  Similarly with
  $\varpi_2$.
\end{proof}

In this section, $S$ will be $\Plotkinn$, $T$ will be $\Val_\bullet$
(equivalently, $\Pred_\Nature^\bullet$), and $U$ will be the following
monad of forks.  A \emph{fork} on a space $X$ is any pair $(F^-, F^+)$
of a superlinear prevision $F^-$ on $X$ and of a sublinear prevision
$F^+$ on $X$ satisfying \emph{Walley's condition} (or is
\emph{medial}):
\[
  F^- (h+h') \leq F^- (h) + F^+ (h') \leq F^+ (h+h')
\]
for all $h, h' \in \Lform X$ \citep{Gou-csl07,KP:predtrans:pow}.  A
fork is \emph{subnormalized}, resp.\ \emph{normalized} if and only if
both $F^-$ and $F^+$ are.

We write $\Pred_{\ADN} X$ for the set of all forks on $X$, and
$\Pred^{\leq 1}_{\ADN} X$, $\Pred^1_{\ADN} X$ for their subsets of
subnormalized, resp.\ normalized, forks.  The \emph{weak topology} on
each is the subspace topology induced by the inclusion into the larger
space $\Pred_{\DN} X \times \Pred_{\AN} X$.  A subbase of the weak
topology is composed of two kinds of open subsets: $[h > r]^-$,
defined as $\{(F^-, F^+) \mid F^- (h) > r\}$, and $[h > r]^+$, defined
as $\{(F^-, F^+) \mid F^+ (h) > r\}$, where $h \in \Lform X$,
$r \in \real^+$.  The specialization ordering of spaces of forks is
the product ordering $\leq \times \leq$, where $\leq$ denotes the
pointwise ordering on previsions.  In particular, all those spaces of
forks are $T_0$.

We will write $\pi_1$ and $\pi_2$ for first and second projection,
from $\Pred^\bullet_\ADN X$ to $\Pred_\DN^\bullet X$ and
$\Pred_\AN^\bullet X$ respectively.  Those are continuous maps.
Additionally, a function $f \colon X \to \Pred_\ADN^\bullet Y$ is
continuous if and only if $\pi_1 \circ f$ and $\pi_2 \circ f$ are.

The following is similar to \citep[Proposition~3]{Gou-csl07}, except
that the Scott topology was considered instead of the weak topology on
spaces of forks.
\begin{proposition}
  \label{prop:ADN:monad}
  Let $\bullet$ be nothing, ``$\leq 1$'', or ``$1$''.  There is a
  monad
  $(\Pred_\ADN^\bullet, \allowbreak \eta^\DN, \allowbreak \mu^\DN)$ on
  $\Topcat$, whose unit and multiplication are defined at every
  topological space $X$ by:
  \begin{itemize}
  \item for every $x \in X$, $\eta^\ADN_X (x) \eqdef (\eta^\DN_X (x),
    \eta^\AN_X (x))$;
  \item for every
    $(\mathcal F^-, \mathcal F^+) \in \Pred_\ADN^\bullet
    {\Pred_\ADN^\bullet X}$,
    $\mu^\ADN_X (\mathcal F^-, \mathcal F^+) \eqdef (\pi_1^{\ext \DN}
    (\mathcal F^-), \pi_2^{\ext \AN} (\mathcal F^+))$.
  \end{itemize}
  Its action on morphisms $f \colon X \to Y$ is given by
  $\Pred_\ADN^\bullet (f) (F^-, F^+) = (\Pred_\DN^\bullet (f) (F^-),
  \Pred_\AN^\bullet (f) (F^+))$.  For every continuous map
  $f \colon X \to \Pred_\ADN^\bullet Y$, its extension $f^{\ext \ADN}$
  is defined by
  $f^{\ext \ADN} (F^-, F^+) \eqdef ((\pi_1 \circ f)^{\ext \DN} (F^-),
  (\pi_2 \circ f)^{\ext \AN} (F^+))$.
\end{proposition}
\begin{proof}
  First, $\eta^\ADN_X (x)$ is a fork for every $x \in X$.  For this,
  we need to check Walley's condition, which reduces to
  $(h+h') (x) \leq h (x) + h' (x) \leq (h+h') (x)$.  The inverse image
  of $[h > r]^+$ is $(\eta^\DN_X)^{-1} ([h > r])$, which is open, and
  similarly with $[h > r]^-$, so $\eta^\ADN_X$ is continuous.

  Second, we check that $f^{\ext \ADN} (F^-, F^+)$ satisfies Walley's
  condition, where $f$ is any continuous map
  $f \colon X \to \Pred_\ADN^\bullet Y$ and
  $(F^-, F^+) \in \Pred_\ADN^\bullet X$.  For all
  $h, h' \in \Lform X$, we need to show that
  $(\pi_1 \circ f)^{\ext \DN} (F^-) (h+h') \leq (\pi_1 \circ f)^{\ext \DN} (F^-)
  (h) + (\pi_2 \circ f)^{\ext \AN} (F^+) (h') \leq (\pi_2 \circ f)^{\ext \AN}
  (F^+) (h+h')$, as follows:
  \begin{align*}
    & (\pi_1 \circ f)^{\ext \DN} (F^-) (h+h') \\
    & = F^- (x \in X \mapsto \pi_1 (f (x)) (h+h')) \\
    & \leq F^- (x \in X \mapsto \pi_1 (f (x)) (h) + \pi_2 (f (x))
      (h')) \\
    & \quad \text{since $f (x)$ satisfies Walley's condition and $F^-$
      is monotonic} \\
    & \leq F^- (x \in X \mapsto \pi_1 (f (x)) (h)) + F^+ (x \in X
      \mapsto \pi_2 (f (x)) (h')) \\
    & \quad \text{since $(F^-, F^+)$ satisfies Walley's condition} \\
    & = (\pi_1 \circ f)^{\ext \DN} (F^-)  (h) + (\pi_2 \circ f)^{\ext \AN} (F^+)
      (h')
  \end{align*}
  and:
  \begin{align*}
    & (\pi_1 \circ f)^{\ext \DN} (F^-) (h) + (\pi_2 \circ f)^{\ext \AN} (F^+)
      (h') \\
    & = F^- (x \in X \mapsto \pi_1 (f (x)) (h)) + F^+ (x \in X
      \mapsto \pi_2 (f (x)) (h')) \\
    & \leq F^+ (x \in X \mapsto \pi_1 (f (x)) (h) + \pi_2 (f (x))
      (h')) \\
    & \quad \text{since $(F^-, F^+)$ satisfies Walley's condition} \\
    & \leq F^+ (x \in X \mapsto \pi_2 (f (x)) (h+h')) \\
    & \quad \text{since $f (x)$ satisfies Walley's condition and $F^+$
      is monotonic} \\
    & = (\pi_2 \circ f)^{\ext \AN} (F^+) (h+h').
  \end{align*}
  The fact that $f^{\ext \ADN} (F^-, F^+)$ is subnormalized if $\bullet$ is
  ``$\leq 1$'' and normalized if $\bullet$ is ``$1$'' follows from the
  fact that $(\pi_1 \circ f)^{\ext \DN} (F^-)$ and
  $(\pi_2 \circ f)^{\ext \AN} (F^+)$ are.  In order to show that $f^{\ext \ADN}$
  is continuous, it suffices to observe that $\pi_1 \circ f^{\ext \ADN}$
  and $\pi_2 \circ f^{\ext \ADN}$ are.  But
  $\pi_1 \circ f^{\ext \ADN} = (\pi_1 \circ f)^{\ext \DN} \circ \pi_1$, and
  similarly with $\pi_2$.

  The monad equations, in the form of Manes' equations
  \ref{it:Manes:eta}, \ref{it:Manes:dagger:1} and
  \ref{it:Manes:dagger:2}, follow immediately from the same equations
  for the $\Pred_\DN^\bullet$ and $\Pred_\AN^\bullet$ monads.  The
  formulae for multiplication and $\Pred_\ADN^\bullet (f)$ also follow
  immediately.
%
\end{proof}

\subsection{The distributing retraction $r_\ADN$, $s_\ADN^\bullet$}
\label{sec:distr-retr-r_adn}

By \citep[Proposition~3.32]{JGL-mscs16}, for every stably compact space
$X$, there is a retraction of $\Plotkinn \Pred_\Nature^\bullet X$ onto
$\Pred_\ADN^\bullet X$, defined by:
\begin{align*}
  r_\ADN
  & \colon \Plotkinn \Pred_\Nature^\bullet X \to
    \Pred_\ADN^\bullet X
  & s_\ADN^\bullet
  & \colon \Pred_\ADN^\bullet X \to \Plotkinn \Pred_\Nature^\bullet X
  \\
  & \quad L \mapsto (r_\DN (L), r_\AN (L))
  && \quad (F^-, F^+) \mapsto s_\DN^\bullet (F^-) \cap s_\DN^\bullet (F^+),
\end{align*}
where $r_\DN$ and $r_\AN$ are extended to lenses by the usual
formulae: $r_\DN (L) (h) \eqdef \min_{G \in L} G (h)$,
$r_\AN (L) (h) \eqdef \sup_{G \in L} G (h)$.  It is clear that
$r_\DN (L)$ is equal to $r_\DN (\upc L)$, and $r_\AN (L)$ is equal to
$r_\AN (cl (L))$ by Lemma~\ref{lemma:sup:cl}, hence to $r_\AN (\dc L)$
on a stably compact space.  This defines a homeomorphism between the
subspace $\Plotkinn^{cvx} \Pred_\Nature^\bullet X$ of convex lenses
and $\Pred_\ADN^\bullet X$ \citep[Proposition~4.8]{JGL-mscs16}.

As a result, the $\Pred_\ADN^\bullet$ monad restricts to $\SComp$.
This is a consequence of the following.
\begin{proposition}
  \label{prop:U:scomp}
  Let $\bullet$ be nothing, ``$\leq 1$'' or ``$1$''.  For every stably
  compact space $X$, $\Pred_\Nature^\bullet X$, $\Val_\bullet X$ and
  $\Pred_\ADN^\bullet X$ are stably compact.
\end{proposition}
\begin{proof}
  $\Val_\bullet X$ is stably compact: this is due to
  \citet{Jung:scs:prob,AMJK:scs:prob} when $\bullet$ is ``$\leq 1$'' or
  ``$1$'', and to \citet[Corollary~1]{Plotkin:alaoglu} in the general
  case.  It follows that the homeomorphic space
  $\Pred_\Nature^\bullet X$ is also stably compact.  We have already
  noted that $\Plotkinn X$ is stably compact \citep[Theorem~6.7,
  item~2]{Lawson:scomp}.  Hence $\Plotkinn \Pred_\Nature^\bullet X$ is
  stably compact, and we conclude since any retract of a stably
  compact space is stably compact
  \citep[Proposition~2.17]{Jung:scs:prob}.  (The result is originally
  due to 
  \citet[page 154, second paragraph]{Lawson:versatile},
  who mentions it in terms of compact subersober spaces.)
\end{proof}

\begin{lemma}
  \label{lemma:sADP:updown}
  Let $X$ be a stably compact space.  For every
  $(F^-, F^+) \in \Pred_\ADN^\bullet X$,
  $\upc s_\ADN^\bullet (F^-, F^+) = s_\DN^\bullet (F^-)$ and
  $\dc s_\ADN^\bullet (F^-, F^+) = s_\AN^\bullet (F^+)$.
\end{lemma}
\begin{proof}
  The set $s_\ADN^\bullet (F^-, F^+)$ is convex, being the
  intersection of two convex sets, and therefore its upward closure is
  convex, too.  Therefore
  $s_\DN^\bullet (r_\DN (\upc s_\ADN^\bullet (F^-, F^+))) = \upc
  s_\ADN^\bullet (F^-, \allowbreak F^+)$, since $r_\DN$ and
  $s_\DN^\bullet$ form a homeomorphism between
  $\SV \Pred_\Nature^\bullet X$ and $\Pred_\DN^\bullet X$.  Since
  $r_\ADN$ and $s_\ADN^\bullet$ form a retraction,
  $r_\ADN (s_\ADN^\bullet (F^-, F^+)) = (F^-, F^+)$, and therefore,
  looking at first components,
  $r_\DN (\upc s_\ADN^\bullet (F^-, F^+)) = F^-$.  It follows that
  $s_\DN^\bullet (F^-) = s_\DN^\bullet (r_\DN (\upc s_\ADN^\bullet
  (F^-, F^+))) = \upc s_\ADN^\bullet (F^-, F^+)$.  We prove that
  $s_\AN^\bullet (F^+) = \dc s_\ADN^\bullet (F^-, F^+)$ similarly.
\end{proof}

\begin{lemma}
  \label{lemma:rADP:nat}
  The transformations $r_\ADN$ and $s_\ADN^\bullet$ are natural on
  $\SComp$.
\end{lemma}
\begin{proof}
  Let $f \colon X \to Y$ be a continuous map.  For every
  $L \in \Plotkinn \Pred_\Nature^\bullet X$,
  \begin{align*}
    (\Pred_\ADN^\bullet (f) \circ r_\ADN) (L)
    & = (\Pred_\DN^\bullet (f) (r_\DN (\upc L)), \Pred_\AN^\bullet (f)
      (r_\AN (\dc L))) \\
    & = (r_\DN (\SV \Pred_\Nature^\bullet (f) (\upc L)),
      r_\AN (\HV \Pred_\Nature^\bullet (f) (\dc L)))
    & \text{since $r_\DN$, $r_\AN$ are natural,}
  \end{align*}
  while $(r_\ADN \circ \Plotkinn \Pred_\Nature^\bullet (f)) (L)$ is
  equal to $r_\ADN (L')$ where
  $L' \eqdef \Plotkinn \Pred_\Nature^\bullet (f)) (L)$.  Now
  $r_\ADN (L') = (r_\DN (\upc L'), r_\AN (\dc L'))$, which is equal to
  $(r_\DN (\SV \Pred_\Nature^\bullet (f) (\upc L)), r_\AN (\HV
  \Pred_\Nature^\bullet (f) (\dc L)))$ by Lemma~\ref{lemma:Pf:L},
  item~\ref{it:Pf:L:Pf}.

  As for $s_\ADN^\bullet$,
  \begin{align*}
    & (\Plotkinn \Pred_\Nature^\bullet (f) \circ s_\ADN^\bullet) (F^-,
      F^+) \\
    & = \upc \Pred_\Nature^\bullet (f) [\upc s_\ADN^\bullet (F^-,
      F^+)]
      \cap cl (\Pred_\Nature^\bullet (f) [\dc s_\ADN^\bullet (F^-,
      F^+)])
      & \text{by def.\ of $\Plotkinn$ on morphisms} \\
    & = \upc \Pred_\Nature^\bullet (f) [s_\DN^\bullet (F^-)] \cap
      cl (\Pred_\Nature^\bullet (f) [s_\AN^\bullet (F^+)])
    & \text{by Lemma~\ref{lemma:sADP:updown}} \\
    & = \SV \Pred_\Nature^\bullet (f) (s_\DN^\bullet (F^-))
      \cap \HV \Pred_\Nature^\bullet (f) (s_\AN^\bullet (F^+))
    & \text{by definition of $\SV$, $\HV$} \\
    & = s_\DN^\bullet (\Pred_\DN^\bullet (f) (F^-))
      \cap s_\AN^\bullet (\Pred_\AN^\bullet (f) (F^+))
    & \text{nat.\ of $s_\DN^\bullet$, $s_\AN^\bullet$
      (Lemmas~\ref{lemma:rDP:nat}, \ref{lemma:rAP:nat})} \\
    & = s_\ADN^\bullet (\Pred_\ADN^\bullet (f) (F^-, F^+))
    & \text{by definition of $s_\ADN^\bullet$.}
  \end{align*}
\end{proof}

Now that we know that $r_\ADN$ and $s_\ADN^\bullet$ are natural, we
can reason 2-categorically, in the category (0-cell) $\SComp$.  The
following equations hold, where $f$ ranges over the 2-cells (natural
transformations) from some 1-cell $F$ to some 1-cell $G$.
(\ref{eq:ADN:STU:s}) is Lemma~\ref{lemma:sADP:updown},
(\ref{eq:ADN:STU:mu:var}) and (\ref{eq:ADN:STU:P:var}) are
Lemma~\ref{lemma:Pf:L}, items~\ref{it:Pf:L:mu} and~\ref{it:Pf:L:Pf}
respectively; (\ref{eq:ADN:STU:mu}) is because
$\pi_1 \circ \mu^\ADN = \pi_1^{\ext \DN} \circ \pi_1 U = \mu^\DN \circ
\Pred_\DN^\bullet \pi_1 \circ \pi_1 U = \mu^\DN \circ \pi_1 \pi_1$ and
similarly with $\pi_2$; the other equations are easy.
\begin{align}
  \label{eq:ADN:STU:r}
  \pi_1 \circ r_\ADN & = r_\DN \circ \varpi_1 T
  & \pi_2 \circ r_\ADN & = r_\AN \circ \varpi_2 T \\
  \label{eq:ADN:STU:s}
  \varpi_1 T \circ s_\ADN^\bullet & = s_\DN^\bullet \circ \pi_1
  & \varpi_2 T \circ s_\ADN^\bullet & = s_\AN^\bullet \circ \pi_2 \\
  \label{eq:ADN:STU:eta}
  \pi_1 \circ \eta^\ADN & = \eta^\DN
  & \pi_2 \circ \eta^\ADN & = \eta^\AN \\
  \label{eq:ADN:STU:mu}
  \pi_1 \circ \mu^\ADN & = \mu^\DN \circ \pi_1 \pi_1 
  & \pi_2 \circ \mu^\ADN & = \mu^\AN \circ \pi_2 \pi_2 
  \\
  \label{eq:ADN:STU:P}
  \pi_1 G \circ \Pred_\ADN^\bullet (f)
                     & = \Pred_\DN^\bullet (f) \circ \pi_1 F
  &\pi_2 G \circ \Pred_\ADN^\bullet (f)
                       & = \Pred_\AN^\bullet (f) \circ \pi_2 F \\
  \label{eq:ADN:STU:eta:var}
  \varpi_1 \circ \eta^\pl & = \eta^\Smyth
  & \varpi_2 \circ \eta^\pl & = \eta^\Hoare
  \\
  \label{eq:ADN:STU:mu:var}
  \varpi_1 \circ \mu^\pl & = \mu^\Smyth \circ \varpi_1 \varpi_1
  & \varpi_2 \circ \mu^\pl & = \mu^\Hoare \circ \varpi_2 \varpi_2 \\
  \label{eq:ADN:STU:P:var}
  \varpi_1 G \circ \Plotkinn f
                     & = \SV f \circ \varpi_1 F
  & \varpi_2 G \circ \Plotkinn f
                       & = \HV f \circ \varpi_2 F.
\end{align}
Let $e \eqdef s_\ADN^\bullet \circ r_\ADN$,
$i \eqdef r_\ADN \circ \Plotkinn \eta^\Nature$ and
$j \eqdef r_\ADN \circ \eta^\pl T$.  We also define
$e_\DN \eqdef s_\DN^\bullet \circ r_\DN$,
$i_\DN \eqdef r_\DN \circ \SV \eta^\Nature$,
$j_\DN \eqdef r_\DN \circ \eta^\Smyth T$,
$e_\AN \eqdef s_\AN^\bullet \circ r_\AN$,
$i_\AN \eqdef r_\AN \circ \HV \eta^\Nature$,
$j_\AN \eqdef r_\AN \circ \eta^\Hoare T$.  Then, using
(\ref{eq:ADN:STU:s}) and (\ref{eq:ADN:STU:r}), (\ref{eq:ADN:STU:r})
and (\ref{eq:ADN:STU:P:var}), or (\ref{eq:ADN:STU:r}) and
(\ref{eq:ADN:STU:eta:var}), or (\ref{eq:ADN:STU:mu:var}) and
(\ref{eq:ADN:STU:P:var}) respectively, we obtain:
\begin{align}
  \label{eq:ADN:STU:e}
  \varpi_1 T \circ e & = e_\DN \circ \varpi_1 T
  & \varpi_2 T \circ e & = e_\AN \circ \varpi_2 T \\
  \label{eq:ADN:STU:i}
  \pi_1 \circ i & = i_\DN \circ \varpi_1
  & \pi_2 \circ i & = i_\AN \circ \varpi_2 \\
  \label{eq:ADN:STU:j}
  \pi_1 \circ j & = j_\DN
  & \pi_2 \circ j & = j_\AN \\
  \label{eq:ADN:STU:dagger}
  \varpi_1 G \circ f^\natural & = (\varpi_1 G \circ f)^\sharp \circ \varpi_1 F
  & \varpi_2 G \circ f^\natural & = (\varpi_2 G \circ f)^\flat \circ \varpi_2 F,
\end{align}
for every 2-cell $f \colon F \dto \Plotkinn G$.


\begin{theorem}
  \label{thm:ADN:STU}
  $\xymatrix{\Plotkinn \Val_\bullet \ar@<1ex>[r]^{r_\ADN} &
    \Pred_{\ADN}^\bullet \ar@<1ex>[l]^{s_\ADN^\bullet}}$ is a
  distributing retraction on $\SComp$.
\end{theorem}
\begin{proof}
  In order to check equations $h=k$ between morphisms with a codomain
  of the form $\Pred_\ADN^\bullet Z$, it suffices to verify that
  $\pi_1 \circ h = \pi_1 \circ k$ and $\pi_2 \circ h = \pi_2 \circ k$.
  Similarly, for morphisms with a codomain of the form $\Plotkinn Z$,
  it suffices to verify that $\varpi_1 \circ h = \varpi_1 \circ k$ and
  $\varpi_2 \circ h = \varpi_2 \circ k$; $\varpi_1$ and $\varpi_2$
  were defined before Lemma~\ref{lemma:Pf:L}.
  Then, verifying the six equations of a distributing retraction
  reduces to the corresponding equations of the distributing
  retractions of the Smyth and Hoare cases (Theorem~\ref{thm:DN:STU},
  Theorem~\ref{thm:AN:STU}).  For example, in order to verify
  (\ref{eq:STU:muS}), we check that
  $\pi_1 \circ r_\ADN \circ \mu^S T = r_\DN \circ \mu^\Smyth T \circ
  \varpi_1 \varpi_1 T$ and
  $\pi_1 \circ \mu^U \circ ir = \mu^\DN \circ i_\DN r_\DN \circ
  \varpi_1 \varpi_1 T$, which are equal by the Smyth case of
  (\ref{eq:STU:muS}), and similarly with $\pi_2$; so
  $r_\ADN \circ \mu^S T = \mu^U \circ ir$.  We leave the other cases
  as exercises.
\end{proof}

\subsection{The weak distributive law $\lambda^\ADN$}
\label{sec:weak-distr-law-2}

\begin{proposition}
  \label{prop:ADN:wdistr}
  The weak distributive law associated with the distributing
  retraction of Theorem~\ref{thm:AN:STU} is given at each object $X$
  of $\Topcat$ by:
  \begin{align}
    \nonumber
    \lambda^\ADN_X
    & \colon \Val_\bullet \Plotkinn X \to \Plotkinn \Val_\bullet X \\
    \label{lambda:natural:h}
    & \quad \xi \mapsto
      \{\nu \in \Val_\bullet X \mid \forall
      h \in \Lform X, \\
    \nonumber
    & \qquad\quad\int_{L \in \Plotkinn X}
      \min_{x \in L} h (x) \,d\xi \leq \int_{x \in  X} h (x) \,d\nu
      \leq \int_{L \in \Plotkinn X} \sup_{x \in L} h (x) \,d\xi\} \\
    \label{lambda:natural:U}
    & \quad \xi \mapsto
      \{\nu \in \Val_\bullet X \mid
      \xi (\Box U) \leq \nu (U) \leq \xi (\Diamond U)
      \text{ for every open subset $U$ of $X$}\}.
  \end{align}
\end{proposition}
\begin{proof}
  Let us continue to use
  $T \eqdef \Pred_\Nature^\bullet \cong \Val_\bullet$.  By
  Proposition~\ref{prop:STU:wdistr},
  $\lambda^\ADN \eqdef s_\ADN \circ \mu^\ADN \circ ji$, so
  $\varpi_1 T \circ \lambda^\ADN = s_\DN^\bullet \mu^\DN \circ j_\DN
  (i_\DN \circ \varpi_1) = \lambda^\DN \circ T \varpi_1$ by
  (\ref{eq:ADN:STU:s}), (\ref{eq:ADN:STU:mu}), (\ref{eq:ADN:STU:j}),
  (\ref{eq:ADN:STU:i}).  Similarly,
  $\varpi_2 T \circ \lambda^\ADN = \lambda^\AN \circ T \varpi_2$.  For
  every $\xi \in \Val_\bullet \Plotkinn X$, it follows that
  $\lambda^\ADN_X (\xi) = \upc \lambda^\ADN_X (\xi) \cap \dc
  \lambda^\ADN_X (\xi) = \lambda^\DN_X (\Val_\bullet \varpi_1 (\xi))
  \cap \lambda^\AN_X (\Val_\bullet \varpi_2 (\xi)) = \lambda^\DN_X
  (\varpi_1 [\xi]) \cap \lambda^\AN_X (\varpi_2 [\xi])$.

  We now use Proposition~\ref{prop:DN:wdistr} and Proposition~\ref{prop:AN:wdistr}.
  \begin{align*}
    \lambda^\ADN_X (\xi)
    & = \{\nu \in \Val_\bullet X \mid \forall
      h \in \Lform X, \int_{x \in X} h (x) \,d\nu \geq \int_{Q \in \SV X}
      \min_{x \in Q} h (x) \,d\varpi_1 [\xi]\} \\
    & \quad \cap \{\nu \in \Val_\bullet X \mid \forall
      h \in \Lform X, \int_{x \in X} h (x) \,d\nu \leq \int_{C \in \HV X}
      \sup_{x \in C} h (x) \,d\varpi_2 [\xi]\} \\
    & = \{\nu \in \Val_\bullet X \mid \forall
      h \in \Lform X, \int_{x \in X} h (x) \,d\nu \geq \int_{L \in \Plotkinn X}
      \min_{x \in \upc L} h (x) \,d\xi\} \\
    & \quad \cap \{\nu \in \Val_\bullet X \mid \forall
      h \in \Lform X, \int_{x \in X} h (x) \,d\nu \leq \int_{L \in \Plotkinn X}
      \sup_{x \in \dc L} h (x) \,d\xi\} \\
    & \qquad\text{by the change of variable formula (\ref{eq:chgvar}).}
  \end{align*}
  Taking minima of (continuous, hence) monotonic maps over $\upc L$ is
  the same as taking them over $L$, and taking suprema over $\dc L$ is
  the same as taking them over $L$, leading to
  (\ref{lambda:natural:h}).  We obtain (\ref{lambda:natural:U}) in a
  similar fashion.
\end{proof}

\begin{remark}
  \label{rem:ADN:wdistr}
  Let us write $\langle E \rangle$ for the order-convex closure
  $\upc E \cap \dc E$ of a set $E$.  When $\xi$ is a finite linear
  combination $\sum_{i=1}^n a_i \delta_{L_i}$, where furthermore each
  $L_i$ is the order-convex closure $\langle E_i \rangle$ of a
  non-empty finite set $E_i \eqdef \{x_{ij} \mid 1\leq j\leq m_i\}$,
  then we claim that:
  \begin{align}
    \label{lambda:natural:fin}
    \lambda^\DN_X (\xi)
    &= \left\langle \conv \left\{\sum_{i=1}^n a_i \delta_{x_{i
      f(i)}} \mid f \in \prod_{i=1}^n \{1, \cdots, m_i\}\right\} \right\rangle,
  \end{align}
  This follows from
  $\lambda^\ADN_X (\xi) = \lambda^\DN_X (\varpi_1 [\xi]) \cap
  \lambda^\AN_X (\varpi_2 [\xi])$ and Remarks~\ref{rem:DN:wdistr}
  and~\ref{rem:AN:wdistr}, once we observe that
  $\varpi_1 [\xi] = \sum_{i=1}^n a_i \delta_{\upc E_i}$ and
  $\varpi_2 [\xi] = \sum_{i=1}^n a_i \delta_{\dc E_i}$.

  This is equivalent to the formula given by Goy and Petri\c san for
  their weak distributive law of the powerset monad over the finite
  distribution monad on $\Setcat$ \citep[Lemma~3.1]{goypetr-dp}.  Note
  that their powerset includes the empty set, while our monad is one
  of non-empty compact sets.  As shown by Aristote (see below), this
  does make a difference.
\end{remark}

\subsection{The Vietoris and Radon monads on $\KHaus$}
\label{sec:viet-radon-monads}

There is a special case of what we have just done, which is perhaps
more interesting to classical topologists, and which is obtained by
restricting to the category $\KHaus$ of compact Hausdorff spaces, and
letting $\bullet$ be ``$1$''.

In a Hausdorff space $X$, the specialization ordering is equality.
Hence we can dispense with $\upc$ and $\dc$ symbols.  Then $(X, =)$ is
a compact pospace, the upward and downward topologies coincide with
the original topology on $X$, and in particular $X$ is stably compact.
Additionally, the lenses are exactly the non-empty compact subsets of
$X$, as one checks easily.  Then $\Plotkinn X$ is the space of
non-empty compact subsets of $X$ with the Vietoris topology.  This
space was first studied by 
\citet{Vietoris:hyper},
then by 
\citet[\S 28]{Hausdorff:Mengen} and 
\citet{Michael:hyper} among others, and is compact Hausdorff.
Hence $\Plotkinn$ restricts to a monad on $\KHaus$, which has become
to be known as the (non-empty) \emph{Vietoris monad} $\mathsf V_*$
\citep{Garner:weak:distr,Goy:PhD,goypetr-dp,Aristote:mono:WDL}.
Beware, though, that the Vietoris monad sometimes refer to a monad
$\mathsf V$ where $\mathsf V X$ is the space of all compact subsets of
$X$, including the empty set.  Aristote showed that there is no
monotone weak distributive law of $\mathsf V$ over the Radon monad
$\mathsf R$ on $\KHaus$ \citep[Theorem~35]{Aristote:mono:WDL},
suggesting that there is perhaps no weak distributive law of
$\mathsf V$ over $\mathsf R$ at all.  (Note that Aristote calls weak
distributive law of $T$ over $S$ what we or 
\citet{Garner:weak:distr} call a weak distributive law of $S$ over
$T$.)

The Radon monad $\mathsf R$ on $\KHaus$ is isomorphic to our $\Val_1$,
as we now see.  A Borel measure $\mu$ on a Hausdorff space $X$ is
\emph{inner regular} if and only if the measure of every Borel subset
$E$ is the supremum of the measures of the compact subsets of $E$.
(Hausdorffness is required so that compact subsets are closed, hence
Borel.)  Following \citet{AMJK:scs:prob}, we call \emph{Radon measure}
on a Hausdorff space any inner regular Borel measure $\mu$ such that
$\mu (K) < \infty$ for every compact subset $K$.  There are alternate
definitions, which are all equivalent on locally compact Hausdorff
spaces.  On a compact Hausdorff space, all bounded Borel measures are
Radon, where a measure $\mu$ is bounded if $\mu (X) < \infty$; this is
the case of all \emph{probability} measures, for which $\mu (X) = 1$.
$\mathsf R X$ is the space of all Radon probability measures on
$X^\patch$, with the \emph{vague topology}, which is the coarsest one
that makes $\mu \mapsto \int_{x \in X} h (x) \,d\mu$ continuous for
every continuous map from $X^\patch$ to $\real$ \citep[Definition
28]{AMJK:scs:prob}.  Then $\mathsf R X$ and $\Val_1 X$ are isomorphic
\citep[Theorem~26, Theorem~36]{AMJK:scs:prob}; in one direction, every
Radon measure yields a probability valuation when restricted to the
open sets, and conversely every probability valuation extends to a
unique Radon measure on all Borel subsets of $X^\patch$
\citep[Theorem~27]{AMJK:scs:prob}.  (See also \citet{KL:measureext}
for other, possibly more general similar extension results.)  This
defines a monad $\mathsf R$ on $\SComp$, which is isomorphic to the
monad $\Val_1$.

\begin{proposition}
  \label{prop:KHaus:restr}
  The monads $\Plotkinn$ ($\cong \mathsf V_*$), $\Val_1$
  ($\cong \mathsf R$), $\Pred_\Nature^1$ and $\Pred_\ADN^1$ restrict
  to $\KHaus$.
\end{proposition}
\begin{proof}
  We have just argued the case of $\Plotkinn$.  For the other two, we
  recall that a \emph{regular} space $X$ is one in which for every
  point $x$, every open neighborhood $U$ of $x$ contains a closed
  neighborhood $C$ of $x$ \citep[Exercise 4.1.21]{JGL-topology};
  equivalently, if every open set is the union of the directed family
  of interiors $\interior C$ of closed subsets $C$ of $U$.  Every
  compact Hausdorff space is regular \citep[Proposition 4.4.17,
  Exercise 4.1.23]{JGL-topology}.

  Let $X$ be regular.  Given any two distinct elements
  $\mu, \nu \in \Val_\bullet$, there must be an open subset $U$ of $X$
  such that $\mu (U) \neq \nu (U)$, say $\mu (U) < \nu (U)$.  Let $r$
  be such that $\mu (U) < r < \nu (U)$.  Since
  $r < \nu (U) = \sup_C \nu (\interior C)$, where $C$ ranges over the
  closed subsets of $U$, we have $r < \nu (\interior C)$ for some
  closed subset $C$ of $U$.  Then $\nu$ is in $[\interior C > r]$,
  $\mu$ is in $[X \diff C > 1-r]$ (since
  $1 = \mu (X) = \mu ((X \diff C) \cup U) \leq \mu (X \diff C) + \mu
  (U) < \mu (X \diff C) + r$, hence $\mu (X \diff C) > 1-r$), and
  $[\interior C > r]$ and $[X \diff C > 1-r]$ are disjoint; for the
  latter, any probability valuation in the intersection must map the
  disjoint union $\interior C \cup (X \diff C)$ to a number strictly
  larger than $r + (1-r)=1$, which is impossible.  Hence $\Val_1 X$ is
  Hausdorff.

  We now assume $X$ compact Hausdorff---hence regular and stably
  compact.  Let $(F_1^-, F_1^+)$ and $(F_2^-, F_2^+)$ be two distinct
  elements of $\Pred_\ADN^1 X$.  Since $r_\ADN$ and $s_\ADN^1$ form a
  retraction, in particular $s_\ADN^1$ is injective, so
  $\mathcal L_1 \eqdef s_\ADN^1 (F_1^-, F_1^+)$ and
  $\mathcal L_2 \eqdef s_\ADN^1 (F_2^-, F_2^+)$ are distinct.  By
  symmetry, let $\nu$ be an element of $\mathcal L_1$ that is not in
  $\mathcal L_2$.  Every compact Hausdorff space is normal, namely one
  can separate disjoint closed subsets by disjoint open neighborhoods
  \citep[Proposition 4.4.17]{JGL-topology}.  $\Val_1 X$ is compact
  Hausdorff, $\mathcal L_2$ is compact, hence closed, and $\{\nu\}$ is
  also closed, so there are disjoint open subsets $\mathcal U$ and
  $\mathcal V$ of $\Val_1 X$ such that $\nu \in \mathcal U$ and
  $\mathcal L_2 \subseteq \mathcal V$.  Then
  $\mathcal L_1 \in \Diamond {\mathcal U}$,
  $\mathcal L_2 \in \Box {\mathcal V}$, and it is easy to see that
  $\Diamond {\mathcal U}$ and $\Box {\mathcal V}$ are disjoint.  It
  follows that $(s_\ADN^1)^{-1} (\Diamond {\mathcal U})$ and
  $(s_\ADN^1)^{-1} (\Box {\mathcal V})$ are disjoint open
  neighborhoods of $(F_1^-, F_1^+)$ and $(F_2^-, F_2^+)$ respectively.
%
%
%
\end{proof}

We obtain the following from Proposition~\ref{prop:KHaus:restr} and
Theorem~\ref{thm:ADN:STU}.
\begin{corollary}
  \label{corl:ADN:STU}
  $\xymatrix{\Plotkinn \Val_1 \ar@<1ex>[r]^{r_\ADN} &
    \Pred_{\ADN}^1 \ar@<1ex>[l]^{s_\ADN^1}}$ is a
  distributing retraction on $\KHaus$.
\end{corollary}
In particular, $\lambda^\ADN$ is a weak distributive law of
$\Plotkinn$ ($\cong \mathsf V_*$) over $\Val_1$ ($\cong \mathsf R$) on
$\KHaus$.  Aristote further showed that this weak distributive law is
monotone \citep[Theorem~35]{Aristote:mono:WDL}.  Our import is,
perhaps, to give an explicit description of the combined monad, and
that is $\Pred_\ADN^1$.

\section*{Acknowledgements}

I would like to thank Richard Garner, Alexandre Goy, Quentin Aristote and Daniela
Petri\c san for discussions that the forerunner \citep{JGL:wdistr} to
this paper spurred between us.

\bibliographystyle{newapa}
\ifarxiv

\else
\iftac
  \newcommand{\etalchar}[1]{$^{#1}$}

\else
\bibliography{wdistr}
\fi
\fi

\end{document}